\documentclass[12pt, english]{report}
\usepackage{amsmath}
\usepackage[colorlinks = true, linkcolor = black, urlcolor = black, citecolor=black, bookmarksopen = true]{hyperref}
\usepackage{cleveref}		
\usepackage{ccaption}
\usepackage{fullpage}
\usepackage{setspace}
\usepackage{amssymb}

\numberwithin{equation}{chapter}	
\numberwithin{figure}{chapter}

\crefformat{equation}{#2equation~#1#3}
\Crefformat{equation}{#2Equation~#1#3}

\usepackage{bm}
\usepackage{graphics}

\usepackage{undertilde}
\usepackage{bbold}
\usepackage{verbatim}

\usepackage[titletoc]{appendix}
\usepackage{natbib}
\usepackage{epstopdf}
\usepackage{floatrow}
\usepackage{float}
\usepackage[linesnumbered,ruled]{algorithm2e}
\usepackage{algorithmic}
\usepackage{subfigure}
\usepackage[nottoc,notlof,notlot]{tocbibind} 

\usepackage{floatrow}				
\usepackage{multirow} 
\usepackage{amsfonts}
\usepackage{soul}
\usepackage[T1]{fontenc}
\usepackage[utf8]{inputenc}
\usepackage{amsthm}
\usepackage[titletoc]{appendix}
\theoremstyle{plain}
\newtheorem{theorem}{Theorem}[chapter]
\newtheorem{corollary}{Corollary}[theorem]
\newtheorem{lemma}[theorem]{Lemma}
\newtheorem{proposition}{Proposition}[theorem]
\theoremstyle{definition}

\theoremstyle{remark}
\newtheorem*{remark}{Remark}
\usepackage{blindtext}
\usepackage{epsfig}
\usepackage[titletoc]{appendix}
\usepackage{siunitx}			
\usepackage[font=small,labelfont=bf]{caption}\usepackage[pdftex,dvipsnames,usenames]{color}
\usepackage{sectsty}
\makeatletter
\renewcommand\subsubsection{\@startsection{subsubsection}{3}{\z@}%
{-3.25ex\@plus -1ex \@minus -.2ex}%
{-1.5ex \@plus .2ex}%
{\normalfont\normalsize\bfseries}}
\makeatother
\begin{document}
\newpage
\thispagestyle{empty}
\clearpage
\vspace*{15pt}
\begin{center}
	\textbf{Optimal p-value weighting with independent information}\\[15pt]
	by\\[15pt]
	\textsc{Mohamad S. Hasan}\\[15pt]
	(Under the Direction of Paul Schliekelman)\\[15pt]
	Abstract\\[15pt]
\end{center}
The large scale multiple testing inherent to high throughput biological data necessitates very high statistical stringency and thus true effects in data are difficult to detect unless they have high effect sizes. One solution to this problem is to use an independent information to prioritize the most promising features of the data and thus increase the power to detect them. Weighted p-values provide a general framework for doing this in a statistically rigorous fashion. However, calculating weights that incorporate the independent information and optimize statistical power remains a challenging problem despite recent advances in this area. Existing methods tend to perform poorly in the common situation that true positive features are rare and of low effect size. We introduce covariate based weighting methods for calculating optimal weights conditioned on the effect sizes of the tests. This approach uses the probabilistic relationship between covariate and test effect size to calculate more informative weights that are not diluted by null effects as is common with group-based methods. This relationship can be calculated theoretically for normally distributed covariates or estimated empirically in other cases. We show via simulations and applications to data that this method outperforms existing methods by a large margin in the rare/low effect size scenario and has at least comparable performance in all scenarios.

\newpage
\thispagestyle{empty}
\textbf{Index words:} Multiple hypothesis test, P-value weighting, Covariates, Filter statistics, Gaussian covariates, Data-driven p-value weighting, Family wise error rate, False discovery rate, Power, Ranking, High-throughput data, RNA-Seq analysis, Genome wide association study, Proteomics, QTL data analysis,


\newpage
\thispagestyle{empty}
\vspace*{15pt}
\begin{center}
	\textbf{Optimal p-value weighting with independent information}\\[15pt]
	by\\[15pt]
	\textsc{Mohamad S. Hasan}\\[15pt]
B.S., University of Dhaka, 2008\\
M.S., University of Dhaka, 2009\\
M.S., University of Central Florida, 2012\\[70pt]

A Dissertation Submitted to the Graduate Faculty\\
of The University of Georgia in Partial Fulfillment\\
of the\\
Requirements for the Degree\\
Doctor of Philosophy\\[50pt]

Athens, Georgia\\
2017
\end{center}

\pagenumbering{roman}

\newpage
\thispagestyle{empty}
\vspace*{200pt}
\begin{center}
	\textcircled{c} 2017\\
Mohamad Shakil Hasan\\
All Rights Reserved\\
\end{center}


\newpage
\thispagestyle{empty}
\vspace*{15pt}
\begin{center}
	\textbf{Optimal p-value weighting with independent information}\\[15pt]
	by\\[15pt]
	\textsc{Mohamad S. Hasan}
\end{center}
\vskip 100pt
\begin{flushleft}\singlespacing
	\vskip 12pt
	\hspace*{200pt}\makebox[96pt][l]{Major Professor:} Paul Schliekelman\\
	\vskip 12pt
	\hspace*{200pt}\makebox[100pt][l]{Committee:        }Lynne Billard\\
	\hspace*{200pt}\makebox[100pt][l]{~                 }Nicole Lazar\\
	\hspace*{200pt}\makebox[100pt][l]{~                 }Romdhane Rekaya\\
	\hspace*{200pt}\makebox[100pt][l]{~                 }Jeongyoun Ahn\\
	
	\vskip 100pt
	Electronic Version Approved:\\
	\vskip 12pt 
	Suzanne Barbour\\
	Dean of the Graduate School\\
	The University of Georgia\\
	December 2017\\
\end{flushleft}



\newpage
\vspace*{15pt}
\section*{\centering Acknowledgments}
\hspace{.22in} First and most importantly, I would like to express my sincere appreciation toward my advisor Dr. Paul Schliekelman. His supervision and encouragement through my Ph.D. study into various innovative statistical methods is the key factor to the success of this dissertation. Dr. Paul Schliekelman is a brilliant and passionate statistician, and his insightful ideas in both statistical theories and applications inspired me. It is a great honor and privilege to
have the opportunity to work closely and learn from him.

This work is also the product of brilliant questions and suggestions from a number of researchers. In particular, I would like to thank Dr. Lynne Billard from the Department of Statistics. I benefited from her expertise in mathematical statistics. Without her, it would have been difficult to solve one of the mathematical proofs regarding the integration of cumulative density function of the normal distribution. I also benefited from Dr. Nicole Lazar. Her questions and comprehensive exam provided me an opportunity to have a solid literature review and directions for improving my writing skills. In addition, I am thankful to other professors on my dissertation committee. Dr. Romdhane Rekaya raised valuable questions and provided me with improvised suggestions with a deep perception which greatly helped me to improve my dissertation. I am also greatly influenced by Dr. Jeongyoun Ahn’s critical questions about Type I error and the independence of the tests and the covariates.

Finally, I would like to thank the Department of Statistics at the University of Georgia. My solid theoretical and applied background in statistics is built from this department. This foundation is not only beneficial to my research but also priceless for my career. I also would like to thank my dear family members for their support and understanding all these years. It has been paramount to my success.

\newpage

\tableofcontents
\listoffigures  
\listoftables 
\listofalgorithms


\newpage
\vspace*{15pt}
\begin{center}
	\textbf{Optimal p-value weighting with independent information}\\[18pt]
	\textsc{Mohamad S. Hasan}
\end{center}
\hspace{.22in}This dissertation is devoted to the study of p-value weighting in a multiple hypothesis test setting. The main goal is to compute weights for the p-values to obtain optimal power. However, since power depends on the effect size, it requires estimating the effect sizes. To overcome this limitation scientists use external information frequently termed as covariates. The key goal is to have an effective selection and uses of the covariates to assign the weights to the p-values and consequently obtain the optimal power of the tests. In this dissertation, the methodological and theoretical research, as well as a considerable portion of the applied work addresses development and effective uses of the covariates to estimate the p-value weights. An innovative contribution of the work is the establishment of a new perspective on the analysis of high throughput data for which relative effect sizes are very low and the true effects are hard to detect with the usual statistical tools.

In this dissertation, we showed three concrete approaches of p-value weighting with an independent information--covariate. In Chapter \ref{ch:Intro}, we discussed the background of the p-value weighting and its advancement. In Chapters \ref{ch:method_crw} - \ref{ch:effective_data_app}, we showed three p-value weighting methods and their validity and applications. In Chapter \ref{ch:discussion}, we provided a summary discussion of the methods that we proposed. We also provided additional materials such as three \textbf{R} packages to perform the computation regarding the proposed methods and additional proofs and plots in the Appendices \ref{ch:proofs} - \ref{ch:addtional_plots}. 

We presented the rigorous theoretical development of the proposed methods and simulation study to support our methods. We called our proposed method Covariate Ranking Weight (CRW), Gaussian Covariate Weight (GCW), and Data-driven Covariate Weight (DCW). We also presented an approximate version of the Gaussian Covariate Weight (GCW2). The key contribution of our methods is that we showed that weight can be computed without directly using the effect sizes. In the proposed methods, we used a probabilistic relationship between the test effect sizes and the covariates or the function of the covariates. We showed exact mathematical expressions as well as approximations of the relationships. We also presented applications of the several data sets to the methods: GWAS, RNA-Seq, Proteomics, and QTL data application. In addition, we introduced and described three \textbf{R} packages: $OPWeight$, $OPWpaper$ and $empOPW$ to reproduce the results and conduct future research. The package $OPWeight$ is designed to apply the algorithms and methodologies that we proposed, package $OPWpaper$ is designed to reproduce the simulation results and the figures presented in this dissertation, and the third package $empOPW$ is designed to show the application of the proposed Data-driven method.

\newpage

\pagenumbering{arabic}


\chapter{Introduction} \label{ch:Intro}
\hspace{.22in} Modern technological advancements in automated data production have produced a large increase in the scale and resolution of information frequently term as high-throughput data, leading to a massive increase in the number of hypothesis tests. Unfortunately, an increased number of tests require more stringent criteria for declaring significance and as a result, the power to detect true effects is frequently very low. Therefore, increasing the resolution of the data may result in a decrease in the probability of identifying true effects \citep[e.g.,][]{kim2015prioritizing}. 

The multiple testing problem is fundamental to high-throughput data and the necessary statistical stringency makes detection of features of interest difficult unless their effect size is large. Unless this problem can be overcome, the promise of high-throughput data may never be realized. Scientists prefer multiple hypothesis over other high-throughput data analysis methods because it is convenient and provides a simple means of studying many features simultaneously, which may not be possible in a single test. Although multiple hypothesis provides an opportunity to test many features simultaneously, it often requires high compensation for doing so \citep{stephens2016false}. One potential solution to this problem is using external information to prioritize the hypothesis tests most likely to yield true positive effects. One means of doing so is p-value weighting, i.e., redefining the p-values by incorporating the weights. 

Weighed p-values provide a framework for using external information to prioritize the data features that are most likely to be true effects. Many statistical methods have been proposed and still advancing to weigh p-values in a multiple hypothesis setting. The first weighting scheme was discussed by \cite{holm1979simple}. A similar discussions can be found in \cite{benjamini1997false, roeder2009genome}, and \cite{gui2012weighted}, among others.

Consider a multiple hypothesis testing situation in which there are $m$ hypothesis tests, and we want to test the null hypothesis $H_{0i}:\varepsilon_i=0$ against the alternative hypothesis $H_{ai}:\varepsilon_i > 0; i=1,\ldots,m,$ where the effect size is $\varepsilon_i=\frac{\sqrt{n_i}\mu_i}{\sigma_i}$ for the i-th alternative mean $\mu_i$ with corresponding sample size $n_i$ and standard deviation $\sigma_i$. For illustrative purpose, we consider only right-tailed tests; however, one can easily generalize to two-tailed tests. Furthermore, we define a set of hypotheses $H=\{H_1,\ldots, H_m\}$, and the corresponding test statistics $Z=\{Z_1,\ldots, Z_m\}$ and p-values $P=\{p_1,\ldots, p_m\}$. In addition, there are $m_0$ true nulls and $m_1$ true alternative tests. 

Generally, the multiple testing procedure rejects a number of hypotheses: some of them are false positives $(V)$, and the rest of them are true positives $(S)$. For the reader’s convenience, the possible outcomes are summarized in Table \ref{table:test_outcome}. 
\begin{table}[ht]
	\caption{Possible outcomes of multiple hypothesis tests}
	\begin{center}
	\begin{tabular}{lccc}
		\hline
		& Not-rejected & Rejected & Total \\ 
		\hline
		True nulls & $m_0 - V$ & $V$ & $m_0$ \\ 
		True alternatives & $m_1 - S$ & $S$ & $m_1$ \\
		\hline 
		Total & $m - R$ & $R$ & $m$ \\ 
		\hline
	\end{tabular}
\end{center}
\label{table:test_outcome}
\end{table}

At the first step, the researchers are concerned about controlling the Family Wise Error Rate (FWER) at a pre-specified significance level of $\alpha$. The FWER is the probability of falsely rejecting at least one true null hypothesis, i.e., $FWER = P(V \ge 1)$. Thus, if $R$ and $m_0$ refer to the total number of rejected null and true null hypotheses, respectively, then $R$ controls the FWER if $P(R \cap m_0 ) \le \alpha$.  For a multiple hypothesis setting, the pre-specified Bonferroni significance level for each test turns out to be $\frac{\alpha}{m}$. In many cases, FWER is too conservative, especially if the number of the hypotheses $m$ is very large. Therefore, an alternative yet popular error-measuring tool is False Discovery Rate (FDR) \citep{benjamini1995controlling}. FDR is defined as the expected value of the proportion of the false positives, which is expressed as 
\begin{equation}
FDR = E\bigg(\frac{V}{max(R, 1)} \bigg)             
\end{equation}

The advantage of the FDR is that the error penalty adapts to the signal of the data and the number of tests \citep{benjamini2010discovering} and has an intuitive application in multiple hypothesis testing settings. Scientists are particularly interested in the Benjamini and Hochberg FDR procedure \citep{benjamini1995controlling}: denote $R(t)$ as the number of rejected null hypothesis, where $t$ is the significance level, and the test will be rejected if $p_i \le t$. Then the objective problem of the FDR is maximizing the number of rejections while controlling the error rate, which can be expressed as
\begin{equation}
max \  R(t), \ s.t. \ \frac{mt}{R(t)} \le \alpha.                   
\end{equation}
As we observe, the FDR method maximizes the number of rejected null hypotheses and maintains the false positive rate, simultaneously, via adjusting the value of $t$.

Another advantage is that the BH procedure can be extended to the Bayesian framework. As observed by Efron \citep{efronlarge}, the frequentist FDR and Bayesian BH FDR coincide at the same point because each hypothesis can be considered as a Bernoulli random variable. Given the Bernoulli outcome, the conditional distribution of the p-values follows either $F_0 \sim uniform(0,1)$ with jointly independent distribution or some other distribution $(F_1)$ with limit $[0,1]$  under the null and alternative hypothesis, respectively, where $F_1$ is stochastically smaller than $F_0$. Therefore, the objective problem becomes $FDR = \frac{mt}{R(t)} =\frac{\pi_0 F_0 (t)}{F(t)}$ if we are strictly conservative, i.e., the proportion of the true null tests $\pi_0 = 1$, where $F = \pi_0 F_0 + (1 - \pi_0 ) F_1$ and $F_0 (t)=t$ (CDF of $uniform(0, 1)$). Thus, the conservativeness of the FDR procedure controls the error while increasing the power of the tests.  
 
Although the FDR procedure shows promising efficacy to increase the power, the multiple hypothesis testing burden is so high that it can only provide a suboptimal power if the number of tests is very large (e.g. thousands). Another major shortcoming of the BH procedure is that it disregards heterogeneity across hypotheses, consequently lose the power of the tests. Theoretical advancement also provides mechanisms to estimate $m_0$ and incorporate with FDR to improve the power. However, the estimation and incorporation of the true null hypothesis $m_0$ somewhat increases the power but is not very effective when the true effect size is very low, such as close to $1$ (this is a very common scenario in high-throughput data) \citep{ignatiadis2016data}. 

Since the available data or existing information is not sufficient enough to produce the optimal power in a multiple hypothesis testing setting, scientists have so far proposed different mechanisms to incorporate external information so that the optimal power can be obtained. Some notable mechanisms are applying covariates, filtering, two-stage testing, etc. In this dissertation, we focused on the incorporation of the external information as covariates. These covariates and the corresponding information are used to estimate the p-value weights. 
 
Considering the above setup, the simple Bonferroni p-value weighting procedure \citep{spjotvoll1972optimality, holm1979simple, benjamini1997false} for a set of non-negative weights $w=\{w_1,\ldots,w_m\}$ rejects a null hypothesis if 
\begin{equation}
i \ \epsilon \ R=\Big\{i:\frac{P_i}{w_i} \le \frac{\alpha}{m}\Big\},
\end{equation}
Since under the null hypothesis the p-values are $uniform(0,1)$, the expected number of false rejections is $\sum_{i=0}^m P\big(p_i \le \frac{\alpha w_i}{m}\big) =  \frac{\alpha}{m}  \sum_{i=1}^{m}w_i$. Thus, this weighting scheme will control FWER if the average weight equals to $1$ \citep{roeder2009genome} or $\sum_{i=1}^{m}w_i = m$. Similarly, \cite{genovese2006false} showed that a weighted FDR procedure can be conducted by using weighted p-values in the standard procedure of \cite{benjamini1995controlling}. The greatest benefit of the p-value weighting comes with the robustness of the method against misspecification of the weights, i.e., an appropriately chosen weight substantially increases the power but the loss of power by the misspecification of the weights is very minimal \citep{roeder2009genome}.

Although many theoretical mechanisms have been established, the fundamental question is how to estimate weight so that the power can be optimized. Many studies  \citep{holm1979simple, benjamini1997false, westfall2004weighted, kropf2004nonparametric, genovese2006false, rubin2006method, wasserman2006weighted, ionita2007genomewide, Finos2007, roeder2009genome, roquain2009optimal, gui2012weighted, dobriban2015optimal, ignatiadis2016data} have proposed methods of p-value weighting, and techniques have been steadily advancing. \cite{holm1979simple} presented a simple and broadly applicable sequential procedure of multiple hypothesis test, i.e., hypotheses are rejected one at a time until no further rejections can be done. The power gain using this method depends upon the alternative. It is small if all the hypotheses are true, but it could be substantial if a number of hypotheses are incorrect. 

\cite{benjamini1997false} presented new approaches as well as extensions to the general weighted multiple hypothesis problems by examining the test formulation and the error rate. \cite{roeder2009genome} provided a comprehensive framework regarding the weighting schemes of multiple hypotheses testing while \cite{gui2012weighted} discussed a comparative study of the weighted testing by providing implications for FWER, FDR, and SFDR (Stratified False Discovery Rate). \cite{wasserman2006weighted} developed a theory of optimal weight and provided some direction to estimate the optimal weight. Considering the above conditions an unweighted $(w_i=1)$ power for a right-tailed alternative hypothesis is
\begin{equation}
\beta(\varepsilon_i;1)=P\big(Z_i > Z_{\frac{\alpha}{m}}\mid \varepsilon_i \big)I(\varepsilon_i > 0)=\bar\Phi\big(Z_{\frac{\alpha}{m}}-\varepsilon_i\big)I(\varepsilon_i > 0)
\end{equation}
Then one can formulate the weighted power \citep{wasserman2006weighted} along this line for the $i^{th}$ right-tailed test statistic $Z_i$ as
\begin{equation}
\beta(\varepsilon_i;w_i)=P\big(Z_i > Z_{\frac{\alpha w_i}{m}}\mid \varepsilon_i\big)I(\varepsilon_i > 0)=\bar\Phi\big(Z_{\frac{\alpha w_i}{m}}-\varepsilon_i\big)I(\varepsilon_i > 0),
\end{equation}
where $Z_{\frac{\alpha w_i}{m}}$ refers to the weighted critical value corresponding to the desired FWER $\alpha$ with multiple hypothesis testing correction $\frac{\alpha w_i}{m}$  for the $i^{th}$ test, and $\Phi=1-\bar \Phi$ refers to the cumulative function of the standard normal distribution. Then the average power for the $m_1$ true alternative hypothesis is
\begin{equation}\label{eq:AvgPower}
\frac{1}{m_1}\sum_{i=1}^{m_1}\beta(\varepsilon_i;w_i)I(\varepsilon_i > 0).
\end{equation}
This formulation requires that the weight, $w_i \ge 0$ and average weight, $\bar w = 1$. For $w_i > 1$, the power increases and decreases for $w_i < 1$. The optimal weight is attained by setting the derivatives of \eqref{eq:AvgPower} to zero and solving the equation after incorporating the weighting constraint. This procedure leads to the following weight identity in terms of the unknown effect size $\varepsilon_i:$
\begin{equation}
\hat w_i= \Big(\frac{m}{\alpha}\Big)\bar\Phi \Big(\frac{\varepsilon_i}{2} + \frac{c}{\varepsilon_i}\Big) I(\varepsilon_i>0),
\end{equation}
where c is constant so that $\sum_{i=1}^{m}\hat w_i=m.$ Unfortunately, calculating these weights requires knowledge of the effect sizes, which is not generally known.\\

\cite{rubin2006method} proposed using a data splitting approach to estimate the effect sizes and weights, i.e., randomly split the data into two parts and use the first part as a training set to estimate $\varepsilon_i$ and the corresponding optimal weights $w_i$. These weights are then applied to the testing set. However, this approach was shown to be inappropriate in terms of power loss. The power gain resulting from the split data procedure cannot compensate the loss of power if the entire data were available to use \citep{roeder2009genome}. 
 
Instead, in order for p-value weighting methods to be effective, i.e., controlling the Type-I error and increasing the power, weights must be estimated from independent data \citep{bourgon2010independent}. To understand why independence is so crucial, consider the same set-up as above and suppose that in addition there is a vector $\bar Y$ of covariates generated from independent data, with the $i^{th}$ element of the vector corresponding to the $i^{th}$ hypothesis test. These covariates will tend to be higher for more promising tests and lower for less promising ones. Furthermore, the covariates are independent under the null hypothesis but correlated to the test statistics under the alternative. These covariates are easily available to the scientists \citep{ignatiadis2016data}.

If the p-values and the covariates under the null hypothesis are not independent, then the Type-I error control will be lost, because in a filtering method (tests are filtered out by the covariates), the p-values that do not pass the first stage will not contribute to the Type-I error; rather, the p-values passing the first stage will. The p-values that passed the first stage are conditionally distributed, and therefore, the corresponding test statistics have heavier tailed distribution compare to the unconditional test statistics. The independence under the null hypothesis between the covariates and test statistics make the independent p-values less likely to pass the first stage. Some test statistics, even though they pass the first stage, will have the same conditional and unconditional distribution because of the independence. Consequently, the p-values in the two stages will have the correct size and thus control the Type-I error \citep{bourgon2010independent}.

The covariates are also sometimes referred to the filter statistics. P-value weighting can be thought of in terms of filtering the test statistics by some independent external information, i.e., giving zero weight to the test statistics implies that the tests are filtered out in the first stage and will not be considered for the next step. However, both are not identical, because the p-values with zero weight can contribute to the overall power, whereas the p-values filtered out in the first stage do not \citep{bourgon2010independent}. 

Appropriately chosen covariates can substantially improve the power, whereas inappropriately chosen covariates can reduce the power. There are many commonly used pairs of test statistics and the covariates that satisfy the independence condition. These covariates are also related to true effect sizes. Some examples of the covariates are for RNA-Seq analysis the mean read count (baseMean) per genes, for GWAS analysis minor allele frequency or p-values from another study, for t-tests overall variances, for two-sided test the sign of the effect, etc \citep{ignatiadis2016data}.

Some information that generally is ignored in a single hypothesis test can also be valuable covariates in a multiple hypothesis testing. For example, consider a single hypothesis $H_i;i=1,\ldots,m$. If the sample observations of the hypothesis test are independently and identically normally distributed, then, under the null hypothesis, which ignores the label of the observations, the two-sample t-test statistic $T_i$ is an ancillary statistic, and the estimated pooled mean and variance $(\hat{\mu_i},s_i^2)$ is a complete sufficient statistic. Thus, by Basu’s theorem, $T_i$ and  $(\hat{\mu_i},s_i^2)$ are statistically independent under the null but informative under the alternative. Therefore, the overall means or variances can be used as covariates \citep{ignatiadis2016data}.

In this dissertation, our focus is not on finding the appropriate covariates; rather, we are interested in showing a few methods which can apply covariates effectively. Various authors have proposed methods for weighting based on the external information-covariates \citep[e.g.][]{satagopan2002two, westfall2004weighted, ionita2007genomewide,  roeder2009genome, kim2015prioritizing, dobriban2015optimal, ignatiadis2016data}. 

\cite{roeder2009genome} proposed a method of breaking the hypothesis tests by the covariates into groups and estimating effect sizes for each group from the data and then calculating the weights for the groups based on these effect sizes. In order to maintain Type-I error control, the group sizes must be large enough that individual null features with chance large test statistics will not inflate estimates of effect size, thus boosting themselves erroneously. However, sufficiently large groups will also likely contain mixtures of true negative and true positive features so that estimates of the effect sizes for the positive features will be diluted, and the resulting weights may be poor. 

Very recently, \cite{ignatiadis2016data}. proposed a related method that they term independent hypothesis weighting (IHW). They split the hypothesis tests into groups by an independent covariate that is believed to be informative about the statistical properties of the hypotheses. Weights are calculated for each group based on a computational approach that maximizes the number of rejections while maintaining a pre-specified FDR. A key component of IHW is that the hypothesis tests are split into $k$ folds. For each fold, the IHW optimization procedure is applied to the other $k-1$ folds, and the resulting weights are applied to the remaining fold. Thus, the weights for each fold do not depend on p-values in the fold itself. Without this procedure, Type-I error control would be lost because a chance small p-value in a group will tend to inflate the weight and thus elevate that p-value to significance. 

Although this method is promising, it suffers from the same problem as the \cite{roeder2009genome} procedure. That is, if true alternative tests are only a small fraction of all tests (as is common for many types of high-throughput data) then all groups will be diluted by many null tests and thus no groups will be strongly weighted. This is particularly true when most of the true alternative tests have small effect sizes. An alternative approach is to weight groups using a mathematical function based on some criterion that is assumed to provide good weighting. 

\cite{westfall2004weighted} proposed using the certain data-driven quadratic forms. \cite{ionita2007genomewide} proposed an exponential weighting scheme for tests sorted by an external covariate. Each subsequent group is twice the size of the previous one, and each weight is $\frac{1}{4}$ of the previous one. \cite{kim2015prioritizing} proposed weights based on a vector $\bar n$ that defines group boundaries among covariate-sorted p-values and uses weights $w_j=\frac{m}{\lambda n_{k(j)}}$, where $n_{k(j)}$ is the smallest value from the vector $\bar n$  that is greater than $j$, and $\lambda$ is a coefficient that makes the weights average to one. 

Although these approaches show good properties in applications, they are not based on any rigorous optimality criterion. Most importantly, the methods are based on group analysis, hence, they cannot provide weight for an individual test. Therefore, they are not effective when the effect sizes are very low. 	
	
Ideally, weights would be calculated individually for each test. The problem is how to do this. \cite{dobriban2015optimal} proposed a method that they called $Bayes\ weights (BW)$ in which they estimated individual weight by assuming Gaussian prior effect sizes. The advantage of the method is that it does not over-fit the model. However, the caveat is that the method assumes the same prior and posterior effect sizes. Furthermore, it may not obtain the optimal power in many cases since the method uses only the prior information without intervening the p-values, especially when the prior and the posterior distribution are not the same, for example, RNA-Seq data, where the prior information could be count data and the posterior distribution is a Gaussian. 

In summary, all existing methods require either 1) difficult to attain knowledge of effect sizes or effect size distributions or 2) grouping tests by the covariates and then estimating properties of the groups.\\ 
 
The starting point for the approaches proposed in this dissertation is the optimal weights of \cite{wasserman2006weighted} and \cite{roeder2009genome}. These oracle weights require knowledge of the true effect sizes, which is not a reasonable requirement. However, we believe that knowledge of the probabilistic relationship between true effect sizes and the covariates or function of the covariates such as ranking of tests by an independent covariate is a reasonable requirement, and we derive optimal weights given this relationship. We also provide methods for estimating this relationship. We show that this approach outperforms other approaches in many scenarios, in particular when the true alternate tests are only a small fraction of all tests and have low effect sizes. Furthermore, we also present data-driven procedures to estimate these relationships.


\chapter{Covariate rank weighting (CRW)}\label{ch:method_crw}
\section{Power formulation} 
\hspace{.25in} In QTL mapping study, the true effect is a genetic polymorphism, which affects the traits of interest, and the genetic marker is the number of hypothesis tests. In RNA sequence analysis, the total number of genes is the number of hypothesis test and differentially expressed genes are the true effect. The numbers of hypotheses generated from these studies are significantly higher than the actual true effect. So the proper ranking, i.e., the most significant tests are ranked higher, of the test statistics can significantly reduce the number of tests and increase the power by up-weighting the true effect while maintaining the FWER. Therefore, we formulate a modified version of the usual power formula by incorporating the probability of rank of the test statistic, ranking by an independent covariate.

Suppose there are $m$ total hypothesis tests with $m_1$ true effects and $m_0 = m-m_1$ null effects. The goal is to test the null hypothesis $H_{0i}:\varepsilon_i=0$ and the alternative hypothesis $H_{1i}:\varepsilon_i>0;i=1,\ldots,m$. There is a test statistic $Z_i$ associated with each hypothesis test and we will assume that this test statistic follows a normal distribution. In addition, each hypothesis test $i$ has an associated covariates $Y_i$. This covariate is calculated from some independent data that is believed to contain information about the effect sizes of the hypothesis tests in the sense that hypothesis tests with larger covariate values are more likely to be true alternate tests. We rank the hypothesis tests by the covariate so that each hypothesis test has a covariate rank $r_i$. We assume that we know the distribution $f(\varepsilon \mid r_i)$, the probabilistic relationship between the covariate ranking and the effect size. We can calculate this directly under certain distributional assumptions, or estimate it from data. Our goal is to derive p-value weights based on this quantity. Define $\beta(r_i;w_i)$ as the power for the test $i$ with rank $r_i$ and weight $w_i$: 
\begin{equation}
\beta(r_i;w_i)=P\big(Z_i > Z_{\frac{\alpha w_i}{m}}\mid r_i\big)
\end{equation}
This can be reformulated by incorporating the generic effect $\varepsilon$ into the power identity as
\begin{equation}
\beta(r_i;w_i)=\frac{\int {P\big(Z_i > Z_{\frac{\alpha w_i}{m}},\varepsilon, r_i \big)I(\varepsilon > 0)d\varepsilon}}{P(r_i)}.
\end{equation}
After simplification by using the intersection rule of the probability theory, we have
\begin{equation}\label{eq:powerLikelihood}
\beta(r_i;w_i )=\int{\bar\Phi\big(Z_{\frac{\alpha w_i}{m}}-\varepsilon\big)f(\varepsilon \mid r_i)I(\varepsilon > 0)d\varepsilon.}
\end{equation}
The conditional probability of the effect given its rank, $f(\varepsilon \mid r_i)$, is not easily attainable. However, by using Bayes rule, we have $f(\varepsilon \mid r_i) = \frac{P(r_i \mid \varepsilon )f(\varepsilon)}{P(r_i )},$
where $P(r_i \mid \varepsilon )$ is the rank probability of the test statistic by an independent covariate given the effect size $\varepsilon$, $f(\varepsilon)$ is the probability density function of the effect size, and $P(r_i)$ is the probability of the rank. 

We also use the prior probability of the covariate rank. This is given by $P(r_i )=\frac{1}{m}$, because, lacking any additional information, all tests are equally likely to receive any covariate rank. Furthermore, we derived an exact and an approximate mathematical formula to compute $P(r_i \mid \varepsilon )$ and the corresponding weights, which are also verified via simulations.
 
\section{Weight}
In this section, we show the derivation of our proposed weighting method CRW when the effect sizes are continuous as well as binary. We also derive the probabilistic relationship between the ranks and the effect sizes.
\subsection{Weight given continuous effect} From equation (\ref{eq:powerLikelihood}), we obtain the average power for the $m$ tests, which is
\begin{equation}
\frac{1}{m}\sum_{i=1}^{m}\beta(r_i;w_i )=\frac{1}{m}\sum_{i=1}^{m}\int{\bar\Phi\big(Z_{\frac{\alpha w_i}{m}}-\varepsilon\big)mP(r_i \mid \varepsilon)f(\varepsilon)I(\varepsilon > 0)d\varepsilon.}
\end{equation}
In order to make this weighting scheme meaningful and valid the average weight has to be equal to $1$, i.e., $\sum_{i=1}^{m}w_i=m$ . Therefore, it requires  normalization of the weights to make the average weight equals to $1$. 

By considering the above weighting constraint, we can formulate a likelihood equation: 
\begin{equation}\label{eq:Likelihood}
L(w_i;r_i) =\frac{1}{m}\sum_{i=1}^{m}\int{\bar\Phi\big(Z_{\frac{\alpha w_i}{m}}-\varepsilon\big)mP(r_i \mid \varepsilon)f(\varepsilon)I(\varepsilon > 0)d\varepsilon} - \delta\Big(\frac{1}{m}\sum_{i=1}^{m}w_i-1\Big),
\end{equation}
where $\delta$ refers to the Lagrange multiplier. Our goal is to find the vector of weights $\bar w$ that maximizes the average power subject to this constraint. Consider the weight $w_k$ for a specific test $k$. We need to maximize the expression above with respect to $w_k$. Therefore, we apply Lagrangian optimization to derive weights: 
\begin{equation}
\frac{dL}{dw_k}=\frac{1}{m}\int \frac{-\phi\Big(\bar\Phi^{-1}(\frac{\alpha w_k}{m})-\varepsilon\Big)(\frac{\alpha}{m})}{-\phi\Big(\bar\Phi^{-1}(\frac{\alpha w_k}{m})\Big)}mP(r_k \mid \varepsilon)f(\varepsilon)d\varepsilon-\frac{\delta}m,
\end{equation}
where $\phi(.)$ and $\Phi(.)$, respectively, refer to the probability density function (PDF) and cumulative density function (CDF) of the standard normal distribution, and $\bar\Phi=1-\Phi$. Equating $\frac{dL}{dw_k}=0$ and performing simple algebra, we obtain
\begin{equation}\label{eq:finalLikelihood}
\int\left(e^{Z_{\frac{\alpha w_k}{m}\varepsilon}-\frac{\varepsilon^{2}}{2}}\right)P(r_k \mid \varepsilon)f(\varepsilon)d\varepsilon=\frac{\delta}{\alpha}.
\end{equation}
The above integration is not easily tractable. An exact mathematical expression could be obtainable depending on the form of $P(r_i \mid \varepsilon)$ and $f(\varepsilon)$; however, almost always an exact numerical estimation of the weights is available. One approach of the numerical estimation could be applying the convolution with Fourier Transformation, which is, in this context, computationally very costly, especially if the test size is very large. Therefore, we obtain a general weight expression that will be applicable in a wide variety of situations; we will derive an approximate form. Let
\begin{equation}\label{eq:taylorFunction}
g(\varepsilon)=\left( e^{Z_{\frac{\alpha w_k}{m}\varepsilon}-\frac{\varepsilon^{2}}{2}}\right) P(r_k \mid \varepsilon).
\end{equation}	
Since $\varepsilon$ is a random variable and $g(\varepsilon)$  is differentiable, by the first order Taylor series expansion of $g(\varepsilon)$ we have 
$E(g(\varepsilon)) \approx g(E(\varepsilon)).$
Consequently, equation (\ref{eq:finalLikelihood}) reduces to 
\begin{equation}\label{eq:taylorApprox}
Eg(\varepsilon) \approx \left( e^{Z_{\frac{\alpha w_k}{m}E(\varepsilon)}-\frac{\big(E(\varepsilon)\big)^{2}}{2}}\right) P\Big(r_k \mid E(\varepsilon)\Big) \approx \frac{\delta}{\alpha}.
\end{equation} 
After simple algebraic manipulation, an approximate version of the weight is obtained from equation (\ref{eq:taylorApprox}), which is
\begin{equation}\label{eq:ContWeight}
w_i \approx \Big(\frac{m}{\alpha}\Big) \bar \Phi \Bigg (\frac{E(\varepsilon)}{2} + \frac{1}{E(\varepsilon)} log\bigg(\frac{\delta}{\alpha P\big(r_i \mid E(\varepsilon)\big)}\bigg)\Bigg).
\end{equation}
where the Lagrange multiplier $\delta \ge 0$, can be obtained by numerical optimization so that $\sum_{i=1}^{m} w_i = m$. In order to control the Type-I error, we need to obtain $P\big(r_i \mid E(\varepsilon)\big)$ from an independent information such as covariate, which is denoted by $P\big(r_i \mid E(\tau)\big)$. Here $\varepsilon$ and $\tau$ refer to the test and covariate effect sizes, respectively; and $r_i$ refer to the rank of a test by the covariate. Note that $E(\varepsilon)$ refers to the expected value of the true alternative test-effect sizes, i.e., $E(\varepsilon)= E(\varepsilon \mid \varepsilon >0)$.

We also compared the approximate weights and the exact weights (Figure \ref{fig:exact vs. approx weights}). To compute the approximate and the exact weights, we used equation (\ref{eq:ContWeight}) and applied the numerical integration to equation (\ref{eq:finalLikelihood}), respectively. We applied these procedures to the two data sets analyzed in Chapter \ref{ch:CRW_data_application}. Both the approximate and the exact approaches provided the similar results. 

To obtain the optimal value of $\delta$, we applied Newton-Raphson and Grid search algorithms. Let us define a function of $\delta$ by $f(\delta)$, then applying the weighting constraint $\frac{1}{m}\sum_{i=1}^{m} w_i = 1$, we have
\begin{equation}\label{eq:fun_delta}
f(\delta) = \sum_{i=1}^{m} \bar \Phi \Bigg (\frac{E(\varepsilon)}{2} + \frac{1}{E(\varepsilon)} log\bigg(\frac{\delta}{\alpha P\big(r_i \mid E(\varepsilon)\big)}\bigg)\Bigg) - \alpha.
\end{equation}
We solved $f(\delta)$ for $\delta \ \epsilon \ [0, \infty)$. Generally, we applied Newton-Raphson (NR) algorithm to obtain the optimal value of $\delta$ when $f(\delta) > 0$. Although Newton-Raphson is computationally very fast, it heavily depends on the correct guess of the initial value. NR method is also sensitive to non-convex problem. From the simulation results, it is evident that if the effect size is very low, we will have flatter weights, then NR method does not converge and unable to find the roots. Therefore, sometimes we used grid search algorithm to obtain the optimal value of $\delta$.  

The approximation of the weights (equation \ref{eq:ContWeight}) is almost always true regardless of whether the information of $\varepsilon$ and $P(r_i \mid \varepsilon)$ is known. The main advantage of this weight is that we do not need to know the effect sizes, we only need the expected value of the effect sizes. In some situation such as a binary effect case (described later), an explicit version of the weight is attainable before obtaining the actual form of $P(r_i \mid \varepsilon)$.

\subsection{Ranks probability given continuous effect \boldmath $P(r_i \mid \tau_i)$}\label{sec:Probability of Rank for Continuous} 
\hspace{.22in} In this section, we show a mathematical derivation of the rank probability given the effect size $P(r_i \mid \tau_i)$.
\begin{theorem}
If there are $m$ hypothesis tests and $r_i, r_{0i}$ and $r_{1i}$; $i=1, \ldots, m$ are the ranks of a random test statistic $t$ given the effect size $\tau_i$ among all tests, among the true nulls, and among the true alternatives, respectively, then the rank probability $P(r_i \mid \tau_i)$ of the test $t$ is the expectation of the pmf of the sum of two Binomial random variables $r_{0i}$ and $r_{1i}$, and the expectation is over $t$. 
\end{theorem}
\begin{proof}
Consider a multiple hypothesis testing situation where $m$ tests are being conducted. Suppose there are $m_0$ true null hypotheses and $m_1=m-m_0$ true alternatives. Let, $x_1,\ldots,x_{m_0}$ be the test statistics under the null model, and $y_1,\ldots,y_{m_1 }$ be the test statistics under the alternative model, and also consider a test statistic $t$. Our goal is to find the distribution for the rank of the test statistic $t$ among the other tests.  
Let $k_0-1$ and $k_1-1$ refer to the number of success, i.e., the number of tests greater than $t$ under the null and alternative models, respectively.
Define, $I_j$ and $I_l$ as follows
\begin{equation}
I_j= 
\begin{cases}
1, & \text{if } X_j > t\\
0, & \text{otherwise}
\end{cases}; \ j = 1,\ldots,m_0
\end{equation}
and		
\begin{equation}
I_l= 
\begin{cases}
1, & \text{if } Y_l > t\\
0, & \text{otherwise}
\end{cases}; \ l = 1,\ldots,m_1.
\end{equation}	
Take $S_0 = \sum_{j=1}^{m_0} I_j$, then $S_0 + 1$  will give the rank of $t$ among the null tests. $I_j$ is a binary random variable and, assuming independence among the tests, the number of tests exceeding $t$ follows a binomial distribution with $m_0$ trials with success probability $P(X_j>t)$. $P(X_j>t)$ can be computed as $P(X_j>t)=1-P(X_j<t)=1-F_0 (t,\tau_j=0)$, where $F_0 (.)$ is the cumulative density function of $X_j$  under the null model. For simplicity, we will denote it $F_0$. Thus, if $k_0$  is the rank of $t$ denoted by a random variable $r_{0j}$, then
\begin{equation}
P(r_{0j}=k_0 \mid t, \tau_j )=\binom{m_0-1}{k_0-1} (1-F_0 )^{k_0-1} F_0^{m_0-k_0}; \ 1 \le k_0 \le m_0.
\end{equation}
Similarly, for the alternative hypotheses we have $S_1 = \sum_{l=1}^{m_1} I_l$; therefore, the probability of success is $P(Y_l>t)=1-P(Y_l<t)=1-F_1 (t,\tau_l)=1-\int F_1 (t,\tau_l)f(\tau_l)d\tau_l$. $F_1(t)$ is the probability of a randomly chosen covariate statistic exceeding $t$. This depends on the effect size $\tau_l$. This effect size is unknown and thus we integrate over possible values. $F_1 (t,\tau_l)$ refers to the cumulative density function of $Y_l$, with $\tau_l$ known.	We will denote $P(Y_l<t)$ as $F_1$. Thus, if $k_1$  is the rank of $t$ denoted by a random variable $r_{1l}$, then
\begin{equation}\label{eq:contAltRanksProb}
P(r_{1l}=k_1 \mid t, \tau_l)=\binom{m_1-1}{k_1-1} (1-F_1 )^{k_1-1} F_1^{m_1-k_1}; \ 1 \le k_1 \le m_1.
\end{equation}
In practice, we do not know $m_0$ and $m_1$. These parameters can be estimated by the method of \cite{storey2003statistical}. In fact, we only need either $m_0$ or $m_1$. Our goal is to obtain rank $k$ under both the null and alternative models, simultaneously, for a random test $t$. If the rank of $t$ is $k$, then there are $k-1$ tests that are higher and $m-k$ tests that are lower than the test $t$, assuming that $t$ is counted as one of the $m$ tests. Let us suppose a random variable $r_i$ refers to the rank of $t$ under both the null and the alternative models, simultaneously. Then
\begin{equation} 
P(r_i=k \mid \tau_i )=P(r_{0j}+ r_{1l}-1=k \mid \tau_i ),
\end{equation}
where $i=1,\ldots,m; \ 0 \le k \le m$; and $m_0+m_1=m$. This is a joint distribution of $r_{0j}$ and $r_{1l}$. Thus, by applying convolution for the discrete case, we have
\begin{equation}\label{eq:ranksProbFirst} 
P(r_i=k \mid \tau_i )=\frac{\sum_{k_0=1}^{k}P(r_{1l}+r_{0j}=k+1,\tau_i,r_{0j}=k_0 ) }{p(\tau_i)}.
\end{equation}
Up to this point $t$ has been assumed to be a fixed value. So, incorporating the test statistic $t$ as a random variable into equation (\ref{eq:ranksProbFirst}) produces to
\begin{equation}\label{eq:ranksProbFinal}
P(r_i=k \mid \tau_i)=\sum_{k_0=1}^{k} \int_t P(r_{1i}=k-k_0+1 \mid \tau_i,t) P(r_{0i}=k_0 \mid \tau_i,t )P(t \mid \tau_i)dt.
\end{equation}
If $t$ is one of the $m_0$ null tests, and we want to compute $P(r_{0i}=k_0 \mid \tau_i,t)$, then the number of null trials in the binomial is $m_0-1$ and the number of alternative trials is $m_1$. If $t$ is one of the $m_1$  alternative tests, then the number of null trials is $m_0$  and the number of alternative trials is $m_1-1$. Similar arguments hold for calculating $P(r_{0i}=k-k_0+1 \mid \tau_i,t)$. Therefore, 
\begin{equation}
P(r_{1i}=k-k_0+1 \mid \tau_i,t)= 
\begin{cases}
\binom{m_1}{k-k_0} (1-F_1 )^{k-k_0} F_1^{m_1-(k-k_0)}, & \text{if } \tau_i = 0\\
\binom{m_1-1}{k-k_0} (1-F_1 )^{k-k_0} F_1^{(m_1-1)-(k-k_0)}, & \text{if } \tau_i > 0
\end{cases},
\end{equation}
and 
\begin{equation}
P(r_{0i}=k_0 \mid \tau_i,t)= 
\begin{cases}
\binom{m_0-1}{k_0-1} (1-F_0 )^{k_0-1} F_0^{m_0-k_0}, & \text{if } \tau_i = 0\\
\binom{m_0}{k_0-1} (1-F_0 )^{k_0-1} F_0^{m_0-(k_0-1)}, & \text{if } \tau_i > 0
\end{cases}.
\end{equation}
The expression $P(r_i=k \mid \tau_i)$ is, in fact, the required probability of rank of the test given the effect size. Since $P(t\mid \tau_i )$ is the probability density function of $t$, equation (\ref{eq:ranksProbFinal}) can be simplified in terms of expectation ($E_T$) over the random variable $t$ as
\begin{equation}
P(r_i=k \mid \tau_i )=\sum_{k_0=1}^{k}E_T \Big[ P(r_{1i}=k-k_0+1 \mid \tau_i,t )P(r_{0i}=k_0 \mid \tau_i,t )\Big] 
\end{equation}
This finishes the proof.
\end{proof}
 
Equation (\ref{eq:ranksProbFinal}) is not easily tractable, and finding a closed form solution is difficult. However, this equation can be solved numerically and can also be easily simulated. We solved this equation using the importance sampling method of the Monte Carlo (MC) simulation. Importance sampling is an MC simulation approach in which the integral is expressed as an expectation of a function of a random variable. Then, the density of the random variable is chosen appropriately which allows reducing the variance of the estimate of the integral. 

Furthermore, if we interchange the expectation and the summation, then the term inside the expectation is the sum of two independent binomial random variables, which is a special case of the poison-binomial PDF. There are many algorithms to solve the poison-binomial PDF numerically \citep{fernandez2010closed}. However, we propose a normal approximation of the term, which introduces a procedure that is faster and easier than many other algorithms. The next result will show a form of normal approximation.
\begin{proposition}
Rank probability $P(r_i=k \mid \tau_i )$ is the expected value of the normal PDF, i.e.,
\begin{equation}
P(r_i=k \mid \tau_i) = 
\begin{cases}
E_T N(\mu_0,\sigma_0^2), & \text{if } \tau_i=0\\
E_T N(\mu_1,\sigma_2^2), & \text{if } \tau_i > 0
\end{cases},
\end{equation}
where
\begin{equation}
\begin{cases}
\mu_0 = (m_0-1)(1-F_0 ) + m_1 (1-F_1 ) + 1\\
\mu_1 = m_0(1-F_0 ) + (m_1-1)(1-F_1 ) + 1\\
\sigma_0 = (m_0-1)(1-F_0 )F_0 + m_1 (1-F_1 )F_1\\
\sigma_1 = m_0(1-F_0 )F_0 + (m_1-1)(1-F_1 )F_1
\end{cases}.
\end{equation}
\end{proposition}

\begin{proof}
If $\tau_i = 0,$ then from the equation (\ref{eq:ranksProbFinal}) $P(r_i=k \mid \tau_i)$ can be written as
\begin{equation*}
P(r_i=k \mid \tau_i) = \sum_{k_0=1}^{k} E_T \bigg[\binom{m_0-1}{k_0-1}(1-F_{0})^{k_0-1}F_{0}^{m_0-k_0}.\binom{m_1}{k-k_0}(1-F_{1})^{k-k_0}F_{1}^{m_1-(k-k_0)} \bigg]. 
\end{equation*}
Equivalently,
\begin{equation*}
P(r_i=k \mid \tau_i) = E_T \sum_{k_0=1}^{k} \bigg[\binom{m_0-1}{k_0-1}(1-F_{0})^{k_0-1}F_{0}^{m_0-k_0}.\binom{m_1}{k-k_0}(1-F_{1})^{k-k_0}F_{1}^{m_1-(k-k_0)} \bigg].
\end{equation*}
Suppose we have two Binomial random variables $X=k_1-1\sim Binom(m_1,1-F_1)$ and $Y=k_0-1 \sim Binom(m_0-1,1-F_0)$; and we want to obtain $Z=X+Y$. Then
\begin{equation*}
\begin{split}
P(Z=k-1) &=P(X+Y=k-1)\\
&=P(X=k-1-Y)\\
&=\sum_{k_0-1=0}^{k-1}P(X=k-k_0)P(Y=k_0-1)\\
&=\sum_{k_0-1=0}^{k-1}\binom{m_1}{k-k_0}(1-F_{1})^{k-k_0}F_{1}^{m_1-(k-k_0)}.\binom{m_0-1}{k_0-1}(1-F_{0})^{k_0-1}F_{0}^{m_0-k_0}.
\end{split}
\end{equation*}
Therefore,
\begin{equation*}
P(r_i=k \mid \tau_i)=E_T P(Z=k-1)=E_T P(Z+1=k),
\end{equation*}
where $Z$ is a sum of two independent binomial random variables. Thus, the PDF of $Z+1$ can be approximated by normal PDF with mean $= E(X+Y+1)$ and variance $= var(X+Y+1)$. Similarly, we can also obtain an approximation when $\tau_i > 0$.
\end{proof}

An useful result which ensures that CRW strictly controls Type-I error. 
\begin{corollary}
If all tests are from the true null models, then the rank is uniformly distributed on $r_i \ \epsilon \ \left[0, m \right] $, i.e., $P(r_i=k \mid \tau_i) = \frac{1}{m}; \ 1 \le k \le m.$
\end{corollary}  
\begin{proof}
	If all tests are from the true null models then 
	\begin{equation*}
	P(r_i=k \mid \tau_i )=\sum_{k_0-1=0}^{k}\int_t\binom{m_0-1}{k_0-1}\binom{m_1}{k-k_0} (1-F_0)^{k-1} F_0^{m-k}f_0(t)dt.
	\end{equation*}
	Let, $F_0=F_0(t)=u$, then $f_0(t)dt=du,$ and $u \hspace{.1cm}  \epsilon \hspace{.1cm} (0, 1)$. Thus,
	\begin{equation*}
	\begin{split}
	P(r_i=k \mid \tau_i ) &= \sum_{k_0-1=0}^{k}\int_{0}^{1}\binom{m_0-1}{k_0-1}\binom{m_1}{k-k_0} (1-u)^{k-1} u^{m-k}du\\
	&=\sum_{k_0-1=0}^{k-1}\binom{m_0-1}{k_0-1}\binom{m_1}{(k-1)-(k_0-1)}Beta(m-k+1, k). 
	\end{split}
	\end{equation*}
	By vendermonde convolution, we have 
	\begin{equation*}
	\begin{split}
	P(r_i=k \mid \tau_i ) &= Beta(m-k+1, k)\binom{m_1+m_0-1}{k-1}\\
	&= Beta(m-k+1, k)\binom{m-1}{k-1}\\ 
	&= \frac{1}{m}
	\end{split}
	\end{equation*}
\end{proof}
\begin{remark}
Following the above corollary, it is evident from equation (\ref{eq:ContWeight}) that as $m \longrightarrow \infty $, $w \longrightarrow 0$; thus, strictly controlled the Type-I error.
\end{remark}

\textbf{Example:} Suppose test statistic $X_j \sim N(0,1)$ under the null model, then $P(X_j>t)=1-F_0(t,\tau_j=0)=1-\Phi(t)=\Phi(-t)$. However, the effect size under the alternative model is a random variable, unlike the null model where $\tau_j=0$. Thus we have to consider the $l^{th}$ effect for the  $l^{th}$ test and exclude all other possible effect sizes. This can be done by integrating out all other possible values of the effect sizes. Therefore, $Y_l \sim N(\tau_l,1)$ if $\tau_l$ is a random variable of the effect sizes, then $P(Y_l>t)=1-\int F_1 (t,\tau_l)f(\tau_l)d\tau_l=1-\int \Phi(t-\tau_l)f(\tau_l )d\tau_l$, where $f(\tau_l)$ is the probability density function of $\tau_l$. 

Sometimes an explicit result of the integral is attainable. For example, if we assume that the effect sizes follow a uniform distribution such that $\tau_l \sim U(a,b)$ ; a reasonable choice because all effect sizes can be similar in size, then we have $P(Y_l > t)=E_{\tau_l} \Phi(\tau_l-t)=\frac{1}{b-a} [(b-t)\Phi(b-t)-(a-t)\Phi(a-t)+\phi(b-t)-\phi(a-t)]$. Similarly, if the effect sizes follow an exponential distribution with rate parameter $\lambda$ then $P(Y_l>t)=\Phi(-t)+e^{-\lambda t} e^{\frac{\lambda^2}{2}} \Phi(t-\lambda)$ (details in Appendix \ref{ch:proofs}). Consequently, the probability of the rank of the tests given the effect size $P(r_i=k \mid \tau_i )$ is obtained by plugging $F_0$ and $F_1$ in equation (\ref{eq:ranksProbFinal}).

\begin{remark}
It is crucial to note that the effect size in (\ref{eq:ranksProbFinal}) is the covariate statistic effect size, whereas that in (\ref{eq:ContWeight}) is the test statistic effect size. The basis of our method is that there is a positive association between these two effect sizes and thus tests that are higher ranked by the covariate statistic will tend to have higher effect sizes in their test statistics. Note that the weight equation (\ref{eq:ContWeight}) only requires the expected value of the test statistic effect size. Initially, we will assume for simplicity that the expected value of the covariate statistic effect size is equal to the expected value of the test statistic effect size. Later, we will explore more complicated relationships between the two effect sizes. For the reader's convenience, the computational procedures of weights is outlined below: 
\end{remark}

\vspace*{-.1in}

\begin{algorithm}[H]
	\caption{Method to compute CRW weights}
	\label{alg:cont_weight}
	\SetAlgoLined
	\textbf{Input:}	$m \gets$ total number of hypothesis tests; 
	$\alpha \ \epsilon \ (0, 1) \gets$ significance threshold;	$E(\varepsilon) \gets$ mean of test-effects;
	$P(r_i \mid E(\tau)) \gets$ ranks probabilities\\
	$f \gets$ function of $\delta$ from equation (\ref{eq:fun_delta})\\
	$f' \gets$ first derivative of $f$\\
	Denote $nmax = 100$ and let a initial value $x_0$\\
	$n = 1$\\
	\eIf{$(n <= nmax))$}
	{
		\While{($f(x_0) < 0$)}
		{
			\eIf{$(f(x_0+0) > f(x_0+.5))$}
			{
				$x_0 = x_0 - .5$ 		
			}{
			$x_0 = x_0 + .5$
		}
		$n = n + 1$
	}
	solve $f$ for $\delta \ \epsilon \ (0, \infty)$ using Newton-Raphson method and return $\delta_{opt}$
}{

Generate a sequence of $\delta = \delta_1,\ldots,\delta_n$, where $\delta \ \epsilon \ (0,1)$ and $\lvert\delta_j-\delta_k\rvert \le .001; \ j,k=1,\ldots,n;j \ne k$.\\

For each $\delta_j$, compute the sum of weights $sw_j(\delta_j)=\sum_{i=1}^{m}w_i$ via equation (\ref{eq:ContWeight})\\  
Denote by $\delta_{opt}$ the optimal value of $\delta$, which satisfy $\underset{\delta_{opt}}{\text{min}} \lvert sw_j (\delta_j)- m\rvert$
}
Compute the weights $w$ by using $\delta_{opt}$\\  
Normalize the weights $w$ so that the weights sum to $m$.\\
\textbf{Return:} $w$, $\delta_{opt}$
\end{algorithm}

\subsection{Weight and ranks probability given binary effect \boldmath $P(r_i \mid \tau)$} 
\hspace{.22in}We also develop a weighting scheme for binary effect size referred as binary or stratified weight. Binary weights are sometimes optimal \citep{roeder2009genome}, and a closed form expression of the weight is easily obtainable. In this weighting procedure, we up weight a fixed fraction of the hypotheses having the fixed effect size $\varepsilon$. Let us suppose there are $m$ hypothesis tests in order  to find $m_1$ true effects, and we want to test $H_{0i}:\varepsilon_i=0$ vs. $H_{1i}:\varepsilon_i=\varepsilon>0;i=1, \ldots,m$, where the effect size is $\varepsilon=\frac{\sqrt{n}\mu}{\sigma}$ for the common alternative mean $\mu$ with corresponding sample size $n$ and the standard deviation $\sigma$. After considering the weighting constraint, we can formulate a likelihood equation similar to the equation (\ref{eq:Likelihood}) in the continuous case; without integration, because of the discreteness of $\varepsilon_i$, which takes only $\varepsilon$ under the alternative model. Then by differentiating with respect to $w_k$ for a specific test $k$ and equating $\frac{dl}{dw_k}=0$, we obtain 
\begin{equation}
\left( e^{Z_{\frac{\alpha w_k}{m}\varepsilon}-\frac{\varepsilon^{2}}{2}}\right)P(r_k \mid \varepsilon)P(\varepsilon)=\frac{\delta}{\alpha}.
\end{equation}
For the binary case,
\begin{equation} 
P(\varepsilon)=
\begin{cases}
p & \text{if } \varepsilon_i = \varepsilon\\
1-p & \text{if } \varepsilon_i =0
\end{cases},
\end{equation}	 
and $p=\frac{m_1}{m}$, initially, then the weight for the binary case can be computed as,
\begin{equation}\label{eq:BinWeight}
w_i = \Big(\frac{m}{\alpha}\Big) \bar \Phi \Bigg (\frac{\varepsilon}{2} + \frac{1}{\varepsilon} log\bigg(\frac{\delta m}{\alpha m_1 P(r_i \mid \varepsilon)}\bigg)\Bigg)I(\varepsilon = \varepsilon).
\end{equation}
To control the Type-I error, we need to obtain $P(r_i \mid \varepsilon)$ from an independent covariate, which is denoted by $P(r_i \mid \tau)$. The main advantage of the binary case is that it is fairly easy to calculate, and an explicit version of the weight is available before knowing the exact form of $P(r_i \mid \tau)$. 

Computing the probability of rank of the test given the effect, $P(r_i \mid \tau)$, for the binary case is similar to the computation performs for the continuous case, except all true alternative effects are the same. Under the null model, the effect size is $0$, which is not different from the continuous section. Therefore, the probability of rank of the test under the null model is same as the probability under the null model for the continuous case. 

On the other hand, under the alternative model, the effect size is a constant; unlike the effect size in the continuous case, the probability of success $P(Y_l>t)=1-P(Y_l<t)=1-F_1 (t, \tau)=1-F_{1B}$, where $F_{1B} (.)$ refers to the cumulative distribution function of $Y_l$ for the binary case, which depends on the common effect size, $\tau$, of the alternative models. Therefore, the probability of rank of the test under the alternative model is obtained by replacing $F_1$ by $F_{1B}$ into equation (\ref{eq:contAltRanksProb}). 

Consequently, if $t$ is a random test value, and $k_0$, $k_1$, and $k$ are the ranks of the random variables $r_{0j}$, $r_{1l}$, and $r_i$ across the null, alternative, and all test statistics, respectively, then the probability of the rank, $P(r_i=k \mid \tau)$, for the binary case has a very similar expression to the expression derived for the continuous case  except $F_1$ will be replaced by $F_{1B}$. Furthermore, a normal approximation can also be obtained.
	
\textbf{Example:} The example presented for the continuous case can be modified for the binary case. The information under the null model would not change. However, the effect size, $\tau$, under the alternative model is not a random variable anymore; it is a fixed value unlike the alternative version of the continuous case. Therefore,  $Y_l \sim N(\tau,1)$ and $P(Y_l>t)=1-F_{1B} (t,\tau)=1-\Phi(t-\tau)=\Phi(\tau-t)$. Furthermore, by following the similar procedures described in Algorithm~\ref{alg:cont_weight}, one can obtain the weights (equation (\ref{eq:BinWeight})) for the binary effect sizes.

\section{Relationship between test effect and covariate effect}
\hspace{.22in} For each test $i$, take $\varepsilon_i$ as the effect size for the test statistic and $\tau_i$ as the effect size for the covariate. The fundamental idea of our approach is that there is a positive relationship between the two effect sizes and hence higher values of the covariate indicate more promising tests. If the two effect sizes are equal, i.e. $\varepsilon_i=\tau_i$, then equation (\ref{eq:ranksProbFinal}) is the required rank probability $P(r_i=k \mid \varepsilon_i)$ of the test given the effect size in the weight equation (\ref{eq:ContWeight}). 

We note further that the rank probability is unaffected by a linear transformation of all effect sizes. Thus, if there is a linear relationship between test effect sizes and covariate effect sizes, then it is still valid to assume $\varepsilon_i=\tau_i$. The weight equation (\ref{eq:ContWeight}) only requires $P\big(r_i \mid E(\varepsilon_i)\big)$. Thus, we conduct a linear regression of estimated covariate effect sizes on estimated test effect sizes and then use this to find the corresponding value of the covariate effect. This value is then used to calculate the rank probability used in the weight equation. In practice, the relationship between test effect and covariate effect is likely to be noisy. We explore the effect of such noise in the simulations below.
\begin{lemma}
Linear transformation of the effect size does not change $P(r_i=k \mid \tau_i)$, i.e., replacing $\tau_i$ by $\alpha + \beta \tau_i$ will not change the rank probability $P(r_i=k \mid \tau_i)$. 	
\end{lemma}
\begin{proof}
Let, $F_0  = F_0 (t; \tau_j  = 0)$ refers to the cumulative density function of a test statistic $X_j$ under the null model and $F_1  = F_1 (t; \tau_l > 0)$ refers to the cumulative density function of a test statistic $Y_l$ with $\tau_l$ known, under the alternative model.
 
As an example, suppose the effect size is $\tau_l\sim Normal(\eta, \sigma^2)$ and the test statistic $t \sim Normal (\tau_l, 1)$ then $P(Y_l > t)$ can be computed as (see proof in Appendix \ref{ch:proofs}-ii):
\begin{equation}
F_1 = 1-P(Y_l < t)=1-\int \Phi(t-\tau_l) \frac{1}{\sigma}\phi\bigg(\frac{\tau_l-\eta}{\sigma}\bigg)d\tau_l.
\end{equation}
If we replace $\tau_l$ by $\alpha + \beta \tau_l$ then the distribution of the effect size becomes $\alpha+ \beta \tau_l \sim Normal(\alpha+ \beta \eta, \beta^2 \sigma^2).$ Denote the new test statistics and $F_1$ by $t'$ and $F_1'$ respectively. Then the new test statistics must be $t' \sim Normal(\alpha+ \beta \tau_l, \beta^2)$. Consequently, 
\begin{equation}
F_1' = 1-P(Y_l < t)=\int \Phi\Bigg(\frac{(\alpha + \beta t)-(\alpha + \beta \tau_l)}{\beta}\Bigg) \frac{1}{\sigma}\phi\Bigg(\frac{(\alpha + \beta \tau_l)-(\alpha + \beta \eta)}{\beta \sigma}\Bigg)d\tau_l = F_1.
\end{equation}
Similar arguments holds for $F_0$. Furthermore, since $F_0$ and $F_1$ can be expressed in terms of the location and scale family, the above results can be generalized to other cases when the effect size does not follow the normal distribution. 
\end{proof}
 
We also wanted to see the relationship between the test effect and the rank obtained from the covariate effect. One way of doing so is to compute the probability of the rank of the test when the effect size is given from the data instead of the external source. Let us suppose $\varepsilon_y$ is the covariate effect and the corresponding rank is $r_y$, and $\varepsilon_t$ is the test effect, Then the relationship between the probability of the rank given the test effect can be defined in terms of the probability of the rank given the covariate effect. If the effects are assumed continuous, then the relationship can be expressed as
\begin{equation}
P(r_y=k \mid \varepsilon_t )=\frac{P(r_y,\varepsilon_t )}{f(\varepsilon_t)}= \frac{\int P(r_y,\varepsilon_y,\varepsilon_t )d\varepsilon_y}{f(\varepsilon_t)} =\int P(r_y \mid \varepsilon_y, \varepsilon_t)f(\varepsilon_y \mid \varepsilon_t)d\varepsilon_y.
\end{equation}
Since including $\varepsilon_t$ does not add additional information once the rank is computed from $\varepsilon_y$; and given $\varepsilon_t$, $f(\varepsilon_y \mid \varepsilon_t)$ is just the conditional PDF of $\varepsilon_y$. Therefore, for a fixed rank, $r_y=k$, the above expression becomes 
\begin{equation}
P(r_y=k \mid \varepsilon_t )=\int P(r_y \mid \varepsilon_y)f(\varepsilon_y \mid \varepsilon_t )d\varepsilon_y=E_{\varepsilon_y}P(r_y \mid \varepsilon_y ); -\infty < \varepsilon_y, \varepsilon_t <\infty,
\end{equation}
where $P(r_y=k \mid \varepsilon_y )$ is the same as $P(r_i=k \mid \varepsilon_i )$ in equation (\ref{eq:ranksProbFinal}). In order to obtain $P(r_y=k \mid \varepsilon_t )$, we need to know the conditional distribution of $\varepsilon_y$ given $\varepsilon_t$.

\textbf{Example} Suppose the joint distribution of $\varepsilon_y$ and $\varepsilon_t$  is a Bivariate Normal, and the marginal distributions of $\varepsilon_y$ and $\varepsilon_t$ are univariate normal, i.e., $\varepsilon_y,\varepsilon_t \sim BVN(\mu_y,\mu_t,\sigma_y^2,\sigma_t^2,\rho)$ and $\varepsilon_y \sim N(\mu_y,\sigma_y^2)$ and $\varepsilon_t \sim N(\mu_t,\sigma_t^2)$, where means $\mu_y, \mu_t$, variances $\sigma_y^2,\sigma_t^2$, and correlation coefficient $\rho$ can be chosen arbitrarily. Then, the conditional distribution of $\varepsilon_y$ given $\varepsilon_t$ can be specified as
$\varepsilon_y \mid \varepsilon_t \sim N\Big(\mu_y+\rho\big(\frac{\sigma_y}{\sigma_t}\big)(\varepsilon_t-\mu_t),(1-\rho^2)\sigma_y^2\Big),$
which is a univariate normal distribution. Thus, to obtain $P(r_y=k \mid \varepsilon_t )$, one needs to take the expectation over $P(r_y=k \mid \varepsilon_y)$, which is a function of the random variable $\varepsilon_y \mid \varepsilon_t$. Figure \ref{fig:filterVsTest_ralation_.9} shows the simulated relationship between the covariate effect and the test effect.


\chapter{Simulation study for CRW}\label{ch:CRW_simulation} 
\hspace{.22in} In this section, we present preliminary results and performance of the proposed CRW method. We performed extensive simulations by considering various aspects such as test parameters, test sizes, effect sizes, etc.; however, we present only a portion of the simulations (For additional results see Appendix \ref{ch:addtional_plots}).
\section{Ranks probability} 
\hspace{.22in} Figure \ref{fig:ranksProb_cont} shows the probability of rank of a test from three different approaches: 1) Simulation approach 2) Exact numerical solution to the CRW method, and 3) Normal approximation of the CRW method. Figure \ref{fig:prob_and_weight_vs_rank} shows the probabilities of the ranks and the corresponding weights for the different effect sizes and proportion of the true null hypothesis.

\begin{figure}
	\begin{center}
		\includegraphics[scale=.6]{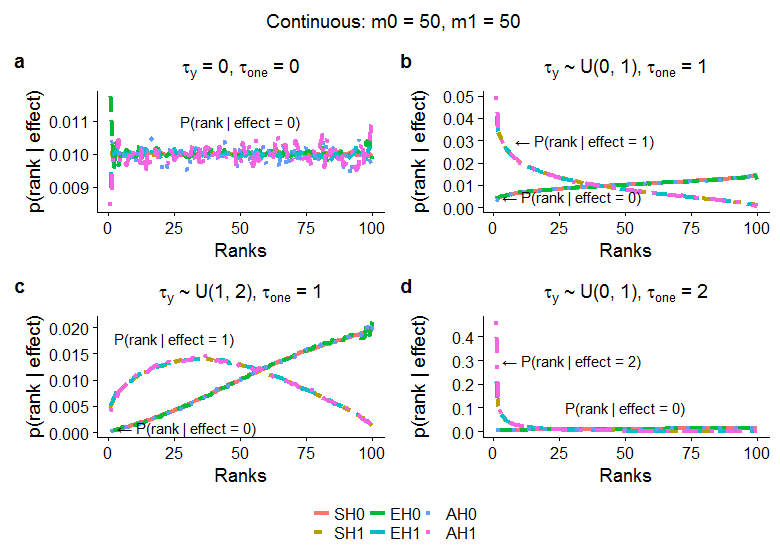}
	\end{center}
	\caption[Plots of $P(r_i=k \mid \tau)$ given the continuous effects]{\footnotesize{This Figure shows $P(r_i=k \mid \tau_i)$ for the continuous case for three different scenarios: 1) from the simulation, 2) from the exact CRW method, and 3) from the CRW normal approximation. To generate these plots, we assumed that there are $m=100$ tests of which $m_0=50$ are true nulls, and $m_1=50$ are true alternatives. The true null tests have mean $0$ and the true alternate tests have the distribution $\tau_y$ as shown at the top of each plot. Each plot consists of six curves, SH0, SH1, EH0, EH1, AH0, and AH1, where the first letter represents the method (S = simulated, E = exact, and A = approximate), and H0 and H1 represent the hypothesis type. Three curves (SH0, EH0, and AH0) starting from the bottom-left represent $P(r_i=k \mid \tau_i=0)$, and the remaining three curves (SH1, EH1, and AH1) starting from top-left represent $P(r_i = k \mid \tau_i = \tau_{one})$, where $\tau_{one} = \{0,1,2\}$. All the curves show the probability of the rank of a statistic with effect size either $\tau_i = 0$ or $\tau_i = \tau_{one}$ across all tests. (a) The rank probabilities of a null test across all tests when all tests are from the true nulls. (b) Rank probabilities of an alternative test with effect size $\tau_{one} = 1$, when $50$ test statistics are from the null models with effect size $0; 49$ test statistics with effect size $\tau_y \sim uniform(0,1)$  and one test statistic with the effect size $\tau_{one}=1$ are from the alternative models. Similarly, plots (c) and (d) show the probabilities for the different effect sizes. All plots of the simulation suggest nearly perfect alignment with the CRW (exact and approximate) methods.}}
	\label{fig:ranksProb_cont}
\end{figure}\
The simulation of Figure \ref{fig:ranksProb_cont} was conducted to verify the ranks probability of a test of the CRW method, $P(r_i=k \mid \tau_i)$. For the approach 1, to perform the simulations, we generated $s=10,000,000$ samples; each sample is composed of $m=100$ test statistics from the normal distribution. We assumed that $m_0=50$ test statistics are from the null models with effect size $0; 49$ test statistics with effect size $\tau_y \sim uniform(0,1)$  and one test statistic with the effect size $\tau_{one}=1$ are from the alternative models. Furthermore, we assumed that the test statistics are normally distributed with standard deviation $1$.  For the simplicity, we considered the test statistics to be normally distributed with mean $\tau_i= 0$ and $\tau_i>0$ under the null and alternative models, respectively, and standard deviation $1$. That is, $F_0$  and $F_1$ are the cumulative functions of the normal distributions. One can easily expand the idea to other distributions.

We combined both groups of tests by placing the $50$ alternative tests first, and the $50$ null tests later. We picked the first test and the $51^{st}$ test and saved their ranks. Since the first test is the alternative and $51^{st}$ test is the null, we would expect that the rank of the first test is always higher; however, this may not be always true, especially if the effect size is very low. We followed the same procedure for $10,000,000$ samples and computed the relative frequency of the ranks, which are the ultimate probabilities of the ranks of the tests by the simulation approach.

For both the exact and the normal approximation approach, to conduct the simulation, we considered $100,000$ replications in the MC simulation to obtain the rank probability curves. We assumed that there are $m =100$ tests of which $m_1=50$ test statistics are from the null model with effect size $0$, and $50$ test statistics are from the alternative model. Of these, $49$  have effect size $\tau_y \sim uniform(0,1)$ and one has effect size $\tau_{one}=\{1 \ or \ 2 \}$. In addition, we directly computed $F_0$ and $F_1$ from $F_0=P(X_j > t) = \Phi(-t)$  and $F_1=(1-t)\Phi(1-t)-(-t)\Phi(-t)+\phi(1-t)-\phi(-t)$, respectively, to obtain the binomial probability of a test is higher than a specific test $t$.

\section{Ranks probabilities and weights}
\hspace{.22in} To calculate the ranks probabilities in Figure \ref{fig:prob_and_weight_vs_rank}, we applied the CRW normal approximation method. We assumed that there are $m=10,000$ hypotheses tests; the test statistics follow normal distributions with mean $0$ and $\varepsilon_i$  under the null and alternative models, respectively, and standard deviations $1$. We generated alternative effect sizes from $uniform(a,b)$ so that the mean of the alternative effect sizes is $E(\tau_i)=\frac{a+b}{2}$. In addition, we directly computed $F_0$ and $F_1$ from $F_0=P(X_j>t)=\Phi(-t)$  and $F_1=P(Y_l>t)=\Phi\big(E(\tau_i)-t\big)$, respectively, to obtain the binomial probability of a test being higher than a specific test $t$ (see Example in the section \ref{sec:Probability of Rank for Continuous}). 
\begin{figure}[ht]
	\begin{center}
		\includegraphics[scale=.58]{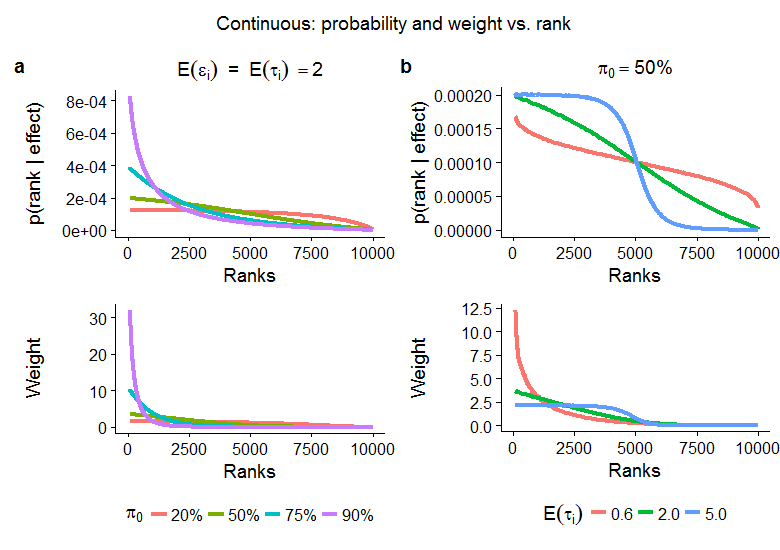}
	\end{center}
	\caption[Plots of $P\big(r_i=k \mid E(\tau)\big)$ and the corresponding normalized weight]{\footnotesize{The ranks probabilities, $P\big(r_i=k \mid E(\varepsilon_i )\big)$, and the corresponding normalized weights $w_i$ (a) for the different proportion of the true null hypotheses $(\pi_0)$ and (b) for the different mean covariate-effect sizes $E(\tau_i)$. The effect size is the same for all alternate hypothesis tests and equal to the value shown. There are $m=10,000$ total tests. The curves were smoothed to eliminate random variation due to Monte Carlo integration.}}
	\label{fig:prob_and_weight_vs_rank}
\end{figure}

To calculate the corresponding normalized weight, we applied the ranks probabilities in the weight equation (\ref{eq:ContWeight}) and assumed that $E(\varepsilon_i)=E(\tau_i)$. Then, numerically optimized with respect to the Lagrangian multiplier $\delta$ so that weights sum to $m$. We performed Monte Carlo (MC) simulations with importance sampling approach and considered $100,000$ replications in the MC simulation to obtain smoother probability curves. For the additional plots of other combination of the parameters and the binary effect sizes see Appendix \ref{ch:addtional_plots}.

Figure \ref{fig:prob_and_weight_vs_rank} shows ranks probabilities and their corresponding weights. In this figure, all alternate hypothesis tests have the same covariate effect size, and the rank probability and corresponding weight are shown for one of these tests. In the left-hand plots, the alternate hypothesis effect size is fixed at a value of 2, while the proportion of true nulls varies between the different curves. In the right-hand plots, the proportion of true nulls is fixed at $50\%$ and the alternate hypothesis effect size varies between curves. 

As is evident from inspection of the weight equation (3), higher weight is assigned to tests that have a higher rank probability. When the alternate hypothesis tests are more differentiated from the background, then there is a higher probability of higher rank and thus higher weight for higher ranked tests. This is observed when comparing the different null proportion curves in Figure \ref{fig:prob_and_weight_vs_rank}a. When true alternate hypothesis tests are rare (e.g. $90\%$ null proportion), then a given true alternate hypothesis test is highly likely to be ranked highly and thus it gets a high weight. When there are many true alternate hypothesis tests with similar effect sizes (e.g. $20\%$ null proportion), then true alternate tests will be spread out over many ranks and thus no rank receives the high weight. 

Similar trends are seen across varying effect sizes in Figure \ref{fig:prob_and_weight_vs_rank}b. When the covariate effect size is high (blue curve), then there is a high probability that all true alternate tests will be ranked above all true null tests. In this case, all ranks in the upper $50\%$ are highly likely to be true alternate tests and thus all receive the same intermediate size weight. When the covariate effect size is weak (red curve), then null tests are likely to spread over all except the highest ranks. Thus, the high ranks get high weights because they are highly differentiated from the background. This behavior is a highly desirable aspect of the CRW method because true alternate tests that are rare and of low effect will receive the highest boost from weighting. In contrast, such effects will tend to be washed out in group-based approaches to weighting (see Figure \ref{fig:group_effect} and \ref{fig:group_effect_ihw}). 

\section{Power}\label{power}
\hspace{.22in} We used simulations to compare the statistical performance of four methods: the proposed CRW method, the Benjamini and Hochberg FDR procedure with no weighting (BH) \citep{benjamini1995controlling}, the weighting method of Roeder and Wasserman (RDW) \citep{roeder2009genome}, the Independent Hypothesis Weighting (IHW) \citep{ignatiadis2016data}. Simulated data were generated for various scenarios and used to estimate power, FDR, and FWER for each method. 

The simulations for power were divided into three groups based on the proportion of the true null hypothesis. The three groups were composed of $\pi_0 =\{50\%, 90\%, 99\%\}$ true null tests. For each group of simulations, we considered the combination of the correlations and the effect sizes as $\rho \times E(\varepsilon_i) = \{0,.3,.5,.7,.9\} \times \{E(\tau_i),N\big(E(\tau_i),CV.E(\tau_i)\big)\}$, where $E(\varepsilon_i)$, $E(\tau_i)$, and $CV$ refers to the mean test-effect, mean covariate-effect, and coefficient of variation. The correlation was between test statistics. For the mean effect sizes and the $CV$, we considered a vector of $\{.2,.4,.6,.8,1,2,3,8\}$ and $CV=\{0,.5,1,3,10\}$. Different $CV$ were used to take into account the influence of the variance of the effect sizes. When the $CV = 0$, we kept the mean test and the mean covariate effect same; however, for the other cases the mean test-effects were generated from the normal distribution with mean $E(\tau_i)$ and standard deviation $CV.E(\tau_i)$.
\begin{figure}[ht]
	\begin{center}
		\includegraphics[scale=.5]{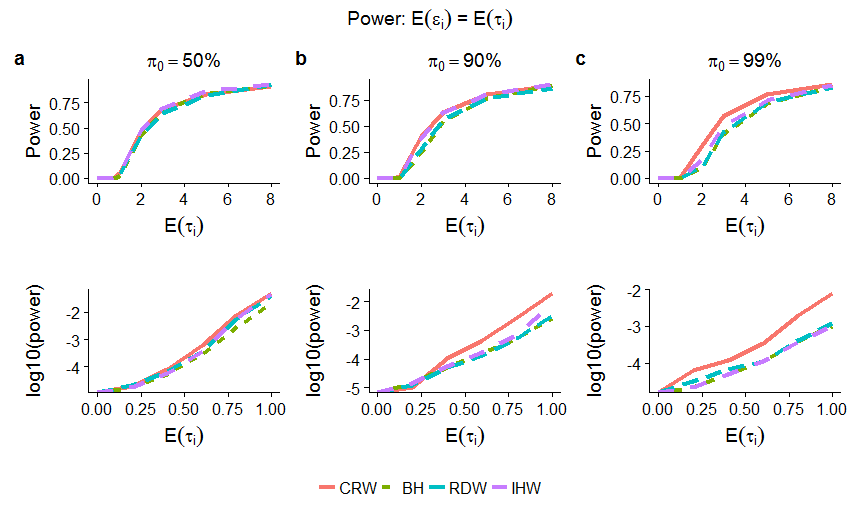}
	\end{center}
	\caption[Comparison of power for the different effect sizes when the mean covariate effect and the mean test effect are the same]{\footnotesize{The Power of four methods when the mean test-effect $E(\varepsilon_i)$ is equal to mean covariate-effect $E(\tau_i)$. Each plot consists of four curves of CRW, BH, RDW, and IHW methods. The first row shows the power for the low to high effect sizes, and the second row shows $log10(power)$ for the low effect sizes. Three columns represent three groups of $50\%,90\%$, and $99\%$ true null hypothesis.}}
	\label{fig:power1}
\end{figure}

For the non-correlated and correlated cases, we generated tests from the normal distribution and the multivariate normal distribution, respectively. The correlation matrix of the multivariate normal distribution had $10,000$ rows and $10,000$ columns; however, the matrix is diagonally split into $100$ blocks, and each block consists of $100$ rows and $100$ columns. All the remaining cells are filled with zeros. Thus, the hypothesis tests were assumed to be structured as $100$ blocks each with $100$ correlated tests. This type of structure might be observed, for example, with correlated expression in gene expression data or linkage disequilibrium between correlated markers in a genome-wide association study (GWAS). We conducted a simulation of $1,000$ replicates and assumed that there were $m=10,000$ hypotheses tests. For each replicate, we determined the proportion of true positive results then calculated the average across the replicates. 

To obtain the results corresponding to the BH method, we used $p.adjust$ function from $R$ software and followed the Benjamini and Hochberg FDR procedure \citep{benjamini1995controlling}. To implement the RDW method and estimate the weights, we followed the procedures described in the paper \cite{roeder2009genome}, section 5; and to implement the IHW method, we applied the $ihw$ function from the $R$ package $IHW$ and kept the default settings.
\begin{figure}[ht]
	\begin{center}
		\includegraphics[scale=.53]{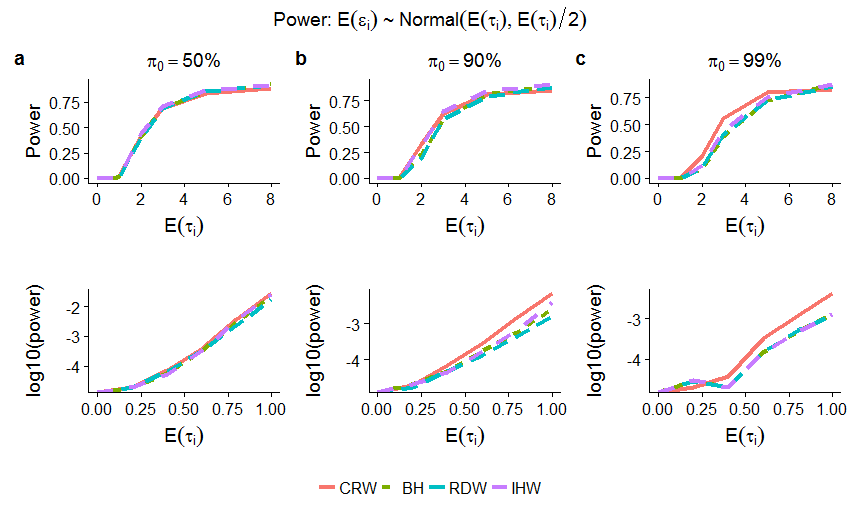}
	\end{center}
	\caption[Comparison of power for the different effect sizes when the mean covariate effect and the mean test effect are not the same]{\footnotesize{The power for the same parameters described in Figure \ref{fig:power1} except the mean test-effect $E(\varepsilon_i)$ is not equal to the mean covariate-effect $E(\tau_i)$; rather $E(\varepsilon_i) \sim Normal\big(E(\tau_i) ,E(\tau_i)/2\big)$, where CV=1/2 (for the other cases see Appendix \ref{ch:addtional_plots}).}}
	\label{fig:power2}
\end{figure}

When the number of the true null hypothesis is low (Figure \ref{fig:power1}a), the IHW method outperforms CRW if the mean effects size is large; however, for all other situations, CRW equals or outperforms all of the three methods. IHW and RDW are based on group analysis, assigning the weights in groups based on tests in the group. When the power is low and/or true alternate tests are rare, then the effects of true alternate tests are washed out by null tests and group weighting loses effectiveness. Figure \ref{fig:group_effect} and \ref{fig:group_effect_ihw} in the simulation results demonstrate that in many scenarios all groups will be dominated by null tests. Unlike IHW, CRW computes the weight for each test; therefore, it can more easily detect the true alternative even though the proportion of the true alternatives and the effect sizes is relatively low.  The relative benefit of CRW over other methods is highest when power is low. In this case, it can have a tenfold advantage in power over IHW. In the data application section and discussion, we will argue that the low power range is of major importance for many types of high-throughput data.

We also observed the power of CRW when the mean test-effect and the mean covariate-effect are different (Figure \ref{fig:power2}). The power of CRW does not affect significantly for the higher mean covariate-effect and the proportion of the true null, and if the coefficient of variation ($CV$) of the test-effect is close to $1$ (see Appendix \ref{ch:addtional_plots}); however, for the low effects and higher $CV$, the power decreases compare to other methods.  
\begin{table}[ht]
	\caption{POWER of the CRW vs. IHW method}
	\begin{center}
		\begin{tabular}{c|c|l|ccccc}
			\hline
			\multirow{ 3}{*}{\% H0} & \multirow{ 3}{*}{$\rho$} & \multirow{ 2}{*}{Methods} & \multicolumn{5}{c}{Mean effect size}  \\ 
			\cline{4-8} & & & 0.6 & 0.8 & 1.0 & 2.0 & 3.0 \\
			\hline
			\multirow{ 6}{*}{50}
			& \multirow{ 2}{*}{.3}
			& CRW & 0.001 & 0.008 & 0.047 & 0.474 & 0.685 \\ 
			& & IHW & 0.001 & 0.007 & 0.043 & 0.472 & 0.696 \\
			\cline{2-8}
			& \multirow{2}{*}{.5}   
			& CRW & 0.001 & 0.009 & 0.046 & 0.474 & 0.684 \\ 
			& & IHW & 0.001 & 0.008 & 0.044 & 0.473 & 0.695 \\
				\cline{2-8}
			& \multirow{2}{*}{.9}  
			& CRW & 0.002 & 0.010 & 0.045 & 0.470 & 0.686 \\ 
			& & IHW & 0.003 & 0.012 & 0.046 & 0.470 & 0.697 \\
			\hline
			\multirow{6}{*}{90}
			& \multirow{2}{*}{.3} 
			& CRW & 0.001 & 0.003 & 0.020 & 0.401 & 0.634 \\ 
			& & IHW & 0.000 & 0.001 & 0.008 & 0.380 & 0.630 \\
				\cline{2-8}
			& \multirow{2}{*}{.5}  
			& CRW & 0.001 & 0.003 & 0.020 & 0.402 & 0.635 \\ 
			& & IHW & 0.000 & 0.001 & 0.009 & 0.380 & 0.630 \\
				\cline{2-8}
			& \multirow{2}{*}{.9}  
			& CRW & 0.001 & 0.005 & 0.022 & 0.403 & 0.636 \\ 
			& & IHW & 0.001 & 0.002 & 0.013 & 0.378 & 0.626 \\
			\hline
			\multirow{6}{*}{99}
			& \multirow{2}{*}{.3}  
			& CRW & 0.001 & 0.002 & 0.008 & 0.307 & 0.572 \\ 
			& & IHW & 0.000 & 0.000 & 0.001 & 0.175 & 0.477 \\
				\cline{2-8}
			& \multirow{2}{*}{.5} 
			& CRW & 0.001 & 0.002 & 0.008 & 0.306 & 0.568 \\ 
			& & IHW & 0.000 & 0.001 & 0.001 & 0.177 & 0.474 \\
				\cline{2-8}
			& \multirow{2}{*}{.9} 
			& CRW & 0.001 & 0.004 & 0.012 & 0.313 & 0.577 \\ 
			& & IHW & 0.000 & 0.001 & 0.002 & 0.179 & 0.474 \\
			\hline 
		\end{tabular}
\end{center}
\footnotesize{$\%$ H0 = Percentage of the true null hypothesis, $\rho$ = Correlation}
\label{table:Power}
\end{table}

\section{FWER and FDR}\label{fwer_fdr_crw}
\hspace{.22in} We conducted simulations to verify that the CRW method controls the Family Wise Error Rate (FWER) and False Discovery Rate (FDR). To observe the FWER when all tests are true nulls, we performed $1,000$ replications, and for each replication, we generated a data set composed of $m=10,000$  observations of three variables: 1) test statistics, 2) p-values, and 3) covariate statistics. We generated two sets of p-values from $unifrom(0,1)$ then one set of p-values is used to compute the covariate statistics and the another set is used to compute the test statistics. We converted the p-values into the statistics by the inverse of the standard normal CDF, i.e., $T = \Phi^{-1}(1-P)$. Then we performed simple regression to obtain a relationship between the test and covariate statistics in which covariate statistics were regressed on the test statistics. 

We computed predicted covariate statistics for the mean and median test statistics and used the predicted values as the estimate of the mean covariate effect size for the continuous and binary cases, respectively. We then applied the predicted covariate statistic to compute the probability of the rank and the corresponding mean test statistic to compute the weight for the continuous case. Similarly, we used the predicted covariate statistic and corresponding median test statistic to compute the probability and the weight for the binary case, respectively. At the end, we compared the sorted p-values with the weighted significance thresholds and then computed the sum of the rejected null hypotheses and averaged across replications to obtain the FWER. Since all test statistics are generated from the null models, we would expect the number of false positive to be below the significance level of $\alpha$. From Figure \ref{fig:fwer}, we see the CRW methods control the FWER and perform similarly to the Bonferroni correction.
\begin{figure}
	\begin{center}
		\includegraphics[scale=.5]{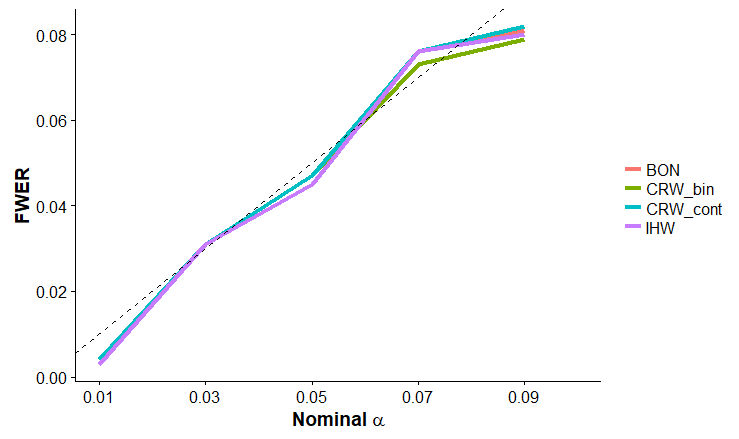}
	\end{center}
	\caption[Type-I error rate of the CRW method]{\footnotesize{This Figure shows the FWER for the different significance levels of $\alpha$. In the legend, the representation is: BON = Bonferroni, CRW\_bin = CRW binary, CRW\_cont = CRW continuous, and IHW = Independent Hypothesis Weighting. To generate these plots, we performed 1,000 replications of $m=10,000$ hypothesis tests. Consequently, to obtain the FWER, the test statistics of the hypothesis tests were generated from the standard normal distribution, then computed the sum of the rejected null hypotheses and average across replications.}}
	\label{fig:Type_I_Error}
\end{figure}
\begin{figure}
	\begin{center}
		\includegraphics[scale=.77]{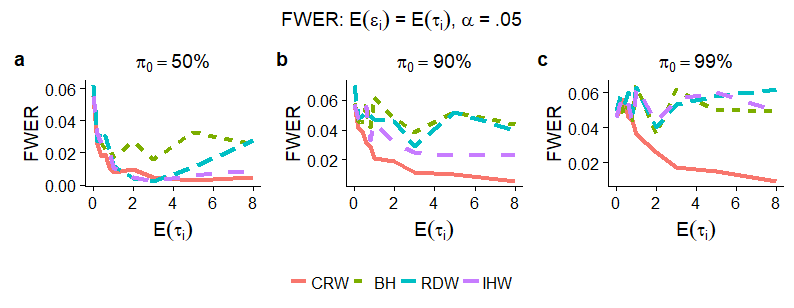}
	\end{center}
	\caption[FWER of the CRW method when the mean test effect and the mean covariate effect are the same]{\footnotesize{This Figure shows the simulated FWER for the different mean effect sizes when the mean test effect $E(\varepsilon_i)$ is equal to mean covariate effect $E(\tau_i)$. Three columns are based on three groups composed of $50\%,90\%$, and $99\%$ true null hypothesis. To generate these plots, we conducted 1,000 replications and assumed that there were $m=10,000$ hypotheses tests.}}
	\label{fig:fwer}
\end{figure}
\begin{figure}
	\begin{center}
		\includegraphics[scale=.57]{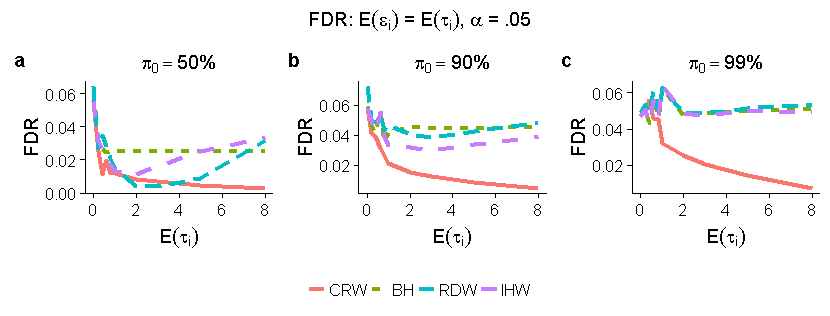}
	\end{center}
	\caption[FDR of the CRW method when the mean test effect and the mean covariate effect are the same]{\footnotesize{This Figure shows the simulated FDR for the different mean effect sizes when the mean test effect $E(\varepsilon_i)$ is equal to mean covariate effect $E(\tau_i)$. Three columns are based on three groups composed of $50\%,90\%$, and $99\%$ true null hypothesis. To generate these plots, we conducted 1,000 replications and assumed that there were $m=10,000$ hypotheses tests.}}
	\label{fig:fdr}
\end{figure}

To observe the FWER and FDR for the situation when the effect sizes are from both the null and alternative models, we considered the same setup described in Section \ref{power} for the Power calculation. When the mean effect sizes are very low the CRW methods tend to lose control although that is not significantly worse. However, both the FWER and the FDR of the CRW method drastically decrease as the effect sizes increase and performed significantly well over other methods, although all other methods control the FDR as well. To obtain the FWER and FDR, we computed the sum of the rejected true nulls and the proportion of the false rejections for each replication, then average over replications.

\section{Power vs. test correlation}
\begin{figure}[ht]
\begin{center}
	\includegraphics[scale=.58]{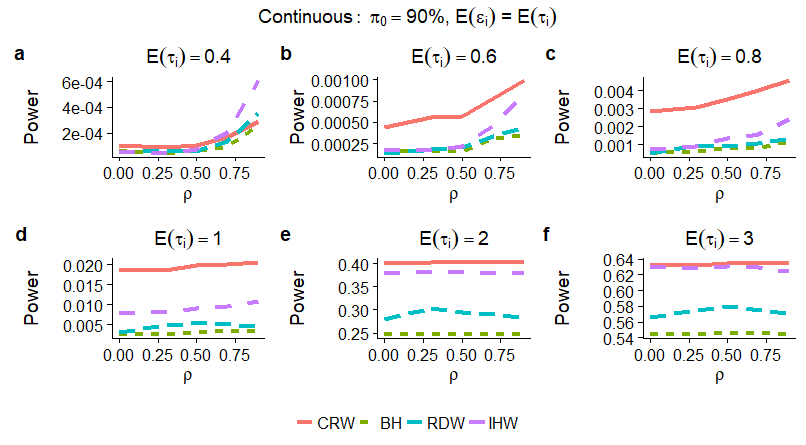}
\end{center}
\caption[Power of the CRW method against the *different correlations of the tests]{\footnotesize{This Figure shows the simulated power across different correlations between the test statistics for the different effect sizes. Each plot consists of four curves based on the proposed CRW, BH, RDW, and IHW methods. To generate these plots, we assumed that there were $m=10,000$ hypotheses and of which $90\%$ are true null. We also assumed that the mean test effect size $E(\varepsilon_i)$ and the mean covariate effect size $E(\tau_i)$ are the same. }}
\label{fig:testCorrelationVsPower_null90}
\end{figure}
\hspace{.22in} we observed the influence of the test correlation to the power of the CRW method (Figure \ref{fig:testCorrelationVsPower_null90}). As we see, the CRW method stays pretty much same across different correlations between the test statistics.  If the mean effect sizes for the test and the covariate are low then the higher correlation increase the power of the CRW as well as the Power of the other methods; however, for the large effect size, the effect of the correlations is negligible. To conduct the simulation, we also adapted the same setup described in the Section \ref{power}. Here, we only presented the case when there is $90\%$ true null hypothesis. For the other cases see Appendix \ref{ch:addtional_plots}.

\section{Power vs. proportion of true null}
\begin{figure}[ht]
	\begin{center}
		\includegraphics[scale=.55]{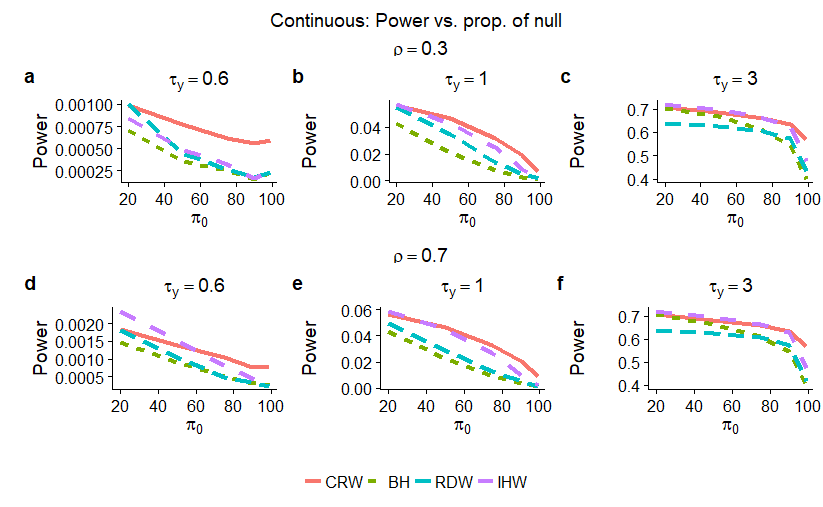}
	\end{center}
	\caption[Power of the CRW method against the different proportion of the true null tests]{\footnotesize{This Figure shows the simulated power across the proportion of the true null hypothesis for the different test correlations. To generate these plots, we assumed that there were $m=10,000$ hypotheses. We also assumed that the mean test effect size $E(\varepsilon_i)$ and the mean covariate effect size $E(\tau_i)$ are the same.}}
	\label{fig:FDRpower_influence_null}
\end{figure}
\hspace{.22in} we observed the influence of the proportion of the true null hypothesis to the power of the CRW methods (Figure \ref{fig:FDRpower_influence_null}). As we see, the CRW method performs well when the proportion of the true alternatives is low, especially when it is below $10\%$. Generally, in the multiple hypothesis settings, this is the actual scenario in which only a fraction of the tests is from the true alternatives. Our method does not affect by the null proportion if the mean covariate effect size $E(\tau_i)$ and the correlation between the test statistics are low. For the moderate mean effect size and the high correlation, our method performs better as long as the proportion of the true null test is above $60\%$. In addition; the CRW method always performs well if the null proportion is high, approximately above $90\%$, regardless of the test correlations and the sizes of the mean covariate effect.

\section{Relationship between test-effect and covariate-effect}
\begin{figure}[ht]
	\begin{center}
		\includegraphics[scale=.55]{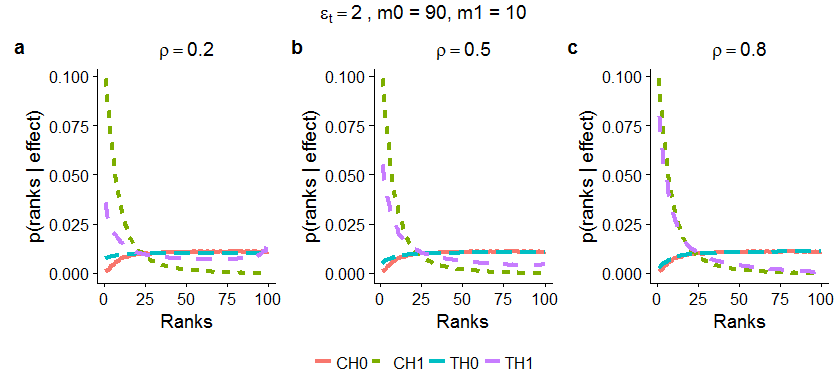}
	\end{center}
	\caption[Relationship between the test effect and the covariate effect]{\footnotesize{This figure shows the relationship between the test effect $(T)$ and the covariate effect $(C)$ in terms of the ranks probabilities of a test given the test effect size, $P(r_y=k \mid \varepsilon_t)$; and compare to the rank probability of the actual CRW method. In the legend, the first letter represents the source of the effects, and $H0$ and $H1$ represent the null and the alternative hypothesis, respectively. To generate the plots, we assumed that the number of hypothesis tests was $m=100$, of which $m_0=90$ are true null and $m_1=10$ are true alternative tests; the mean test effect size of the alternative test is $\varepsilon_t=2$; and the correlation varies by $\rho=\{.2,.5,.8\}$. We performed $10,000$ replications to compute the probability of a specific rank, and for a specific rank; we generated $5,000$ observations of $\varepsilon_y$ from $Normal(\rho \varepsilon_t,1-\rho^2 )$ then computed the expectation of  $P(r_y \mid \varepsilon_y)$  to obtain $P(r_y=k \mid \varepsilon_t)$.}}
	\label{fig:filterVsTest_ralation_.9}
\end{figure}
\hspace{.22in} For the simulation, we assumed that the joint distribution of $\varepsilon_y$ and $\varepsilon_t$  is a Bivariate Normal, and the marginal distributions of $\varepsilon_y$ and $\varepsilon_t$ are univariate normal, i.e., $\varepsilon_y,\varepsilon_t \sim BVN(0,0,1,1,\rho)$ and $\varepsilon_y  \sim N(0,1)$ and $\varepsilon_t \sim N(0,1)$, where the correlation coefficient, $\rho$, was chosen arbitrarily. Consequently, the conditional distribution of $\varepsilon_y$ given $\varepsilon_t$ is $\varepsilon_y \mid \varepsilon_t \sim Normal(\rho \varepsilon_t,1-\rho^2 )$, which is a univariate normal distribution. 

Our goal is to observe the change of the relationship between the covariate-effect and the test-effect with the change of $\rho$ in terms of the ranks probability. We would expect that the probability plots computed from test effect sizes are similar to the probability plots computed from the covariate effect sizes. Figure \ref{fig:filterVsTest_ralation_.9} shows the simulated relationship between the covariate effect and the test effect for the different parameters. To conduct the simulation, we adapted the importance sampling approach of the Monte Carlo simulation. As we see, in general, when the correlation is low $(\rho=.2)$, the probability curves generated from the test effect (TH0 and TH1) are completely separated from the probability curves generated from the covariate effect (CH0 and CH1) and show no change with the covariate curves; however, for the high correlation case, the curves are aligned almost perfectly.

\section{Group effect}
\hspace{.22in}This simulation is conducted to show how alternate tests are washed out by null tests in group-based methods like IHW. We chose IHW because IHW is the best among the existing group based methods. To conducted the simulation, we generated $m = 40,000$ test statistics and covariate statistics from the normal distribution. The mean vector of the statistics is composed of true null and alternative effect sizes, where the proportion of the true nulls $\pi_0$ was pre-specified. To see the overall group effect, we ranked the effect sizes by the covariates then split the effect sizes into $10$ groups. We then computed the proportion of the true alternative effects $(\varepsilon_i > 0)$ in the best group and the mean effect size in the best group. To see the group effect on the IHW method, we applied their $ihw$ function from the $R$ package $IHW$. For each of the procedures, we conducted $100$ replications. 
\begin{figure}[ht]
	\begin{center}
		\includegraphics[scale=.55]{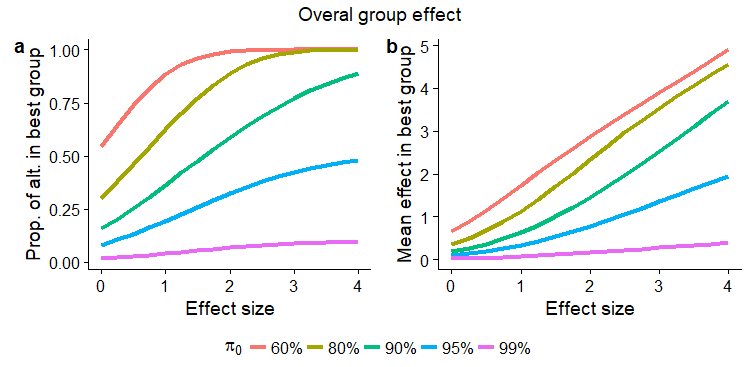}
	\end{center}
	\caption[Group effect that demonstrates the disadvantages of the group-based p-value weighting]{\footnotesize{(a) Shows the proportion of true alternate in the best group and (b) Mean effect size in the best group for the different proportion of the true null tests $\pi_0$.}}
	\label{fig:group_effect}
\end{figure}

The point is to show that even the best group has many nulls in a group-based method. Therefore, the weights become flat as the effect sizes get small and $\pi_0$ gets large, i.e., When there are a large number of null tests, all groups will have many nulls and therefore the weights would be flatter. It is expected that as the effect size gets lower than the top group would get the highest weight because then it is very hard for anything to be significant.  
\begin{figure}[ht]
	\begin{center}
		\includegraphics[scale=.55]{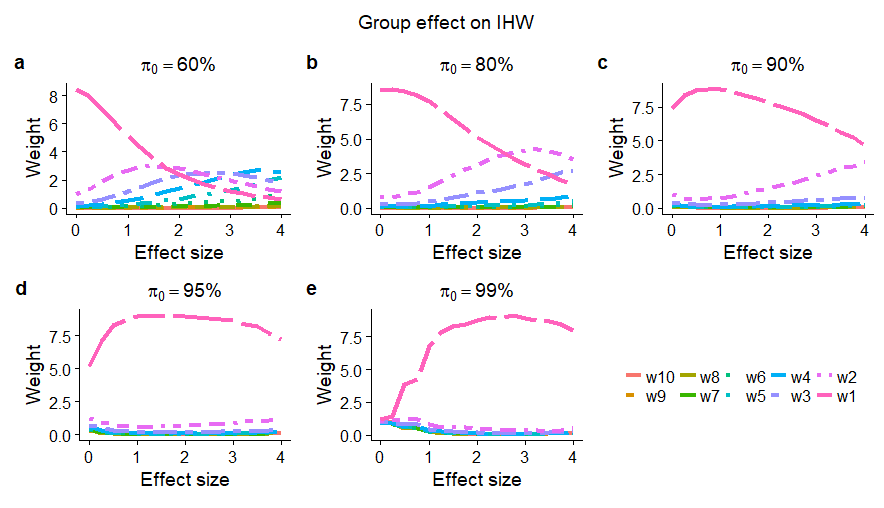}
	\end{center}
	\caption[Group effect on the IHW method]{\footnotesize{This figure shows the group effect on IHW for the different proportion of the null tests $\pi_0$. In the legend, $w1 - w10$ refers to the weights from the First to the $10^{th}$ group.}}
	\label{fig:group_effect_ihw}
\end{figure}\\

These figures verify the $group \ dilution$, i.e., many nulls dilute the groups. However, for the CRW method, the high ranks get high weights because they are highly differentiated from the background (see Figure \ref{fig:prob_and_weight_vs_rank}). This behavior is a highly desirable aspect of the CRW method because true alternate tests that are rare and of low effect will receive the highest boost from weighting. In contrast, such effects will tend to be washed out in group-based approaches to weighting. 

IHW and RDW are based on group analysis, assigning the weights in groups based on tests in the group. When the power is low and/or true alternate tests are rare, then the effects of true alternate tests are washed out by null tests and group weighting loses effectiveness. Unlike IHW, CRW computes the weight for each test; therefore, it can more easily detect the true alternative even though the proportion of the true alternatives and the effect sizes is relatively low. The relative benefit of CRW over other methods is highest when power is low. In this case, it can have a tenfold advantage in power over IHW.


\chapter{Data application to CRW}\label{ch:CRW_data_application}
\section{Data sources and pre-screening analysis}
\hspace{.22in} In this section, we showed two data examples: 1) Bottomly-RNA-Seq data and 2) Proteomics data. The description and pre-screening analysis of the data are given below:

\textbf{Bottomly:} This a RNA-Seq data set obtained and downloaded from \cite{frazee2011recount}. The authors generated single end RNA-Seq reads from $10 B6$ and $11 D2$ mice ($21$ lanes on three Illumina GAIIx flowcells). After applying the method described in the paper the data had read counts for $36,229$ genes of which $12,632$ genes had no reads across all lanes and $7,414$ genes had no reads for at least one $B6$ and one $D2$ sample. Therefore, the subsequent analyses were based on the remaining $16,183$ genes.

We used $DESeq2$ package \citep{Love2014} from $Bioconductor$ software to obtain the p-values and the mean of normalized counts (baseMean) for genes. We maintained default settings for the $design \sim strain$ to obtain the statistics and used the mean of normalized counts $baseMean$ for genes of the samples as the covariate. The $baseMean$ is used in $DESeq2$ for estimating the fitted dispersion. \cite{Love2014} argued that the covariates $baseMean$ and the test statistics are approximately independent under the null hypothesis. Under the null hypothesis t-test $T_i; \ i=1,\ldots,m$ is an ancillary statistic, and the estimated pooled mean and variance $(\hat{\mu_i},s_i^2)$ is a complete sufficient statistic. Thus, $T_i$ and  $(\hat{\mu_i},s_i^2)$ are independent under the null but informative under the alternative, which satisfies the requirement for the $baseMean$ to be the covariates. Therefore, we used $baseMean$ as covariates.
\begin{figure}[ht]
	\begin{center}
		\includegraphics[scale=.56]{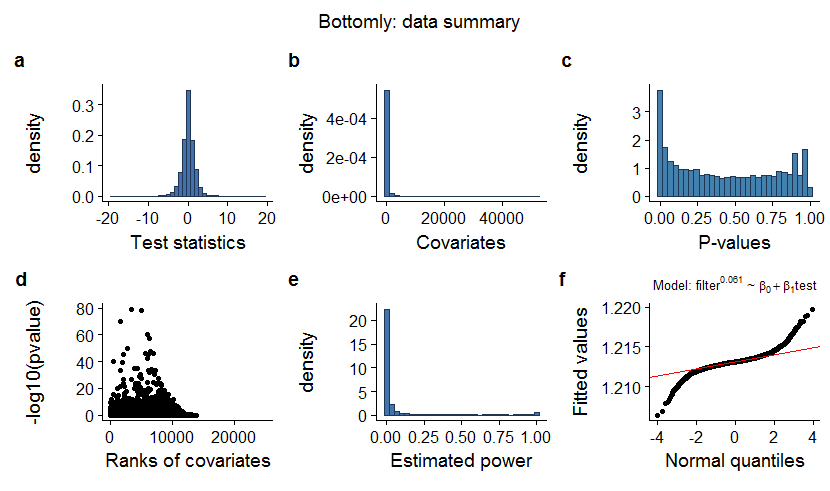}
	\end{center}
	\caption[Bottomly data pre screening summary]{\footnotesize{This Figure shows the summary information of the Bottomly-RNA-seq data. First row (left-right) shows the distribution of the test statistics, the covariates, and the p-values, respectively; and the second row shows the rank of the covariates vs. the p-values (low index is better rank), estimated power of the tests, and the Q-Qplot of the fitted values from the regression model with Box-Cox transformation.}}
	\label{fig:bottmly_pre_screening}
\end{figure}\\
From the pre-screening analysis, we see:\\ 
1) Covariates are not normally distributed; therefore, we applied box-cox transformation in the simple linear regression.\\
2) Bimodal p-values indicates two-tailed test criteria are necessary because p-values close to 1 be the cases that are significant in the opposite direction.\\ 
3) Low p-values are enriched at high covariates or low ranks of the covariates. This indicates that the covariates $baseMean$ are correlated to the power under the alternative hypothesis.\\
4) The estimated power plot shows that there are only a handful of tests which have more than $50\%$ power and the Q-Q plot suggests that only the center part of the data fitted well.

\textbf{Proteomics:} This data set is applied in the article \cite{dephoure2012hyperplexing}. The authors demonstrated an approach to increase the multiplexing capacity of quantitative mass spectrometry. By applying the approach, they quantified $2,666$ proteins in six replicates. We used their $Welch \ t-test$ p-values and the number of quantified peptides as the covariate from their supplementary Table 1. The overall peptide count is an indicator for the quality of the data and for intrinsic interest in the gene. It does not not have any relationship with whether the gene is differentially expressed. Therefore, peptide counts is used as covariates.
\begin{figure}[ht]
	\begin{center}
		\includegraphics[scale=.56]{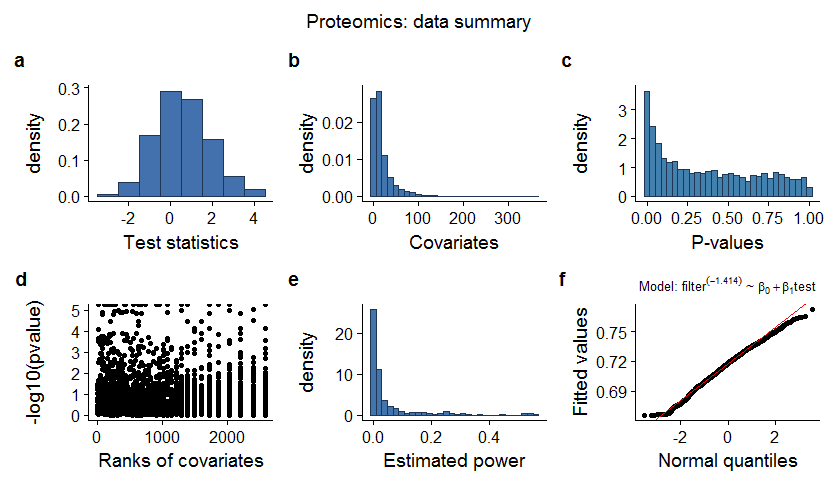}
	\end{center}
	\caption[Proteomics data pre screening summary]{\footnotesize{First row shows the distribution of the test statistics, covariates, and the p-values, respectively; and the second row shows the rank of the covariates vs. the p-values, estimated power of the tests, and the Q-Q plot of the fitted values from the regression model with Box-Cox transformation of the Proteomics data.}}
	\label{fig:proteomics_pre_screening}
\end{figure}\\ \\
From the pre-screening analysis, we see:\\ 
1) Covariate statistics is not normally distributed; therefore, we applied box-cox transformation in the simple linear regression.\\
2) Unimodal p-values indicates one-tailed test criteria are necessary.\\ 
3) There is a very weak relationship between the higher covariates or high rank of the covariate. This indicates that the covariates peptide counts are weakly correlated to the power under the alternative hypothesis test.\\
4) The estimated power plot shows that the power of the tests is very low and the Q-Q plot suggests a weak fit of the data on the tail ends. 

\section{Application of CRW method}\label{sec:crw_data_steps}
\hspace{.22in} Before applying CRW, we performed a pre-screening analysis of the data composed of p-values $(P)$ and the covariates $(Y)$. We then followed several steps to analyze the data, in particular to obtain the ranks probabilities, $P(r_i=k \mid E(\varepsilon_i))$ and the corresponding weights $w_i$. For the reader’s convenience, the steps are summarized below:

\textbf{Input:} A nominal significance level $\alpha \ \epsilon \ (0,1),$ a vector of p-values $P = p_1, \ldots,p_m$ and a vector of covariates $Y=y_1, \ldots,y_m$ corresponding to each p-value, which is independent of $P$ under $H_0$.

\textbf{Step 1:} Denote the vector of test statistics by T. Obtain the test statistics by $T = \Phi^{-1}(1-P)$ for the one-tailed and $T=\Phi^{-1}(1-P/2)$ for the two-tailed p-values, where $\Phi^{-1}$ refers to the inverse of the standard normal Cumulative Density Function (CDF).

\textbf{Step 2:} Define a simple linear regression by $y_i = \beta_0+\beta_1 t_i+\epsilon_i;1=1,\ldots,m$. Obtain the relationship between the test statistics $(T)$ and the covariates $(Y)$ by the linear regression. In the regression, Box-Cox transformation or other transformations need to be used if the covariates are not approximately normal.

\textbf{Step 3:} Denote by $m,m_0$, and $m_1$ the number of hypothesis tests, the number of true null tests, and the number of true alternative tests, respectively. Estimate the numbers of the tests by using the method of \cite{storey2003statistical}. In particular, use the $qvalue$ function from the $R$ package $qvalue$. 

\textbf{Step 4:} Order the test statistics in the decreasing order and pick the top $m_1$ tests. This set of tests is used to estimate the alternative hypothesis test effects.

\textbf{Step 5:} Denote by $\bar T$ and $T_m$ the mean and the median of the true alternative test statistics. Calculate these from the top $m_1$  tests from Step 4. 

\textbf{Step 6:} Apply the estimated model $\hat{y_i} = \hat \beta_0 + \hat \beta_1 t_i$ to compute the predicted covariates $Y$ and $\hat{Y_m}$  for the mean and the median test statistics. The predicted covariates corresponding to the mean and the median tests statistics were considered as the estimate of the mean covariate-effect sizes. 

\textbf{Step 7:} Apply equation (\ref{eq:ranksProbFinal}) to compute the rank probability $P\big(r_i=k \mid E(\varepsilon_i)\big)$, where $E(\varepsilon_i)$ will be replaced by the value of $\hat{Y}$ and $Y_m$ if the effects are assumed continuous and binary, respectively. 

\textbf{Step 8:} Compute weights by applying equation (\ref{eq:ContWeight}) or (\ref{eq:BinWeight}), where $E(\varepsilon_i)$ will be replaced by $\bar T$ and $T_m$ if the effects are assumed continuous and binary, respectively.

\textbf{Step 9:} Apply the weighted Bonferroni method $\sum_{i=1}^{m}I(p_i \le \frac{\alpha w_i}{m})$ or weighted FDR in the \cite{benjamini1995controlling} procedure $\sum_{i=1}^{m}I(p_{ia} \le \alpha)$ to obtain the number of significant tests, where $I(.)$ and $p_{ia}$ refers to the indicator function and the adjusted p-values after applying the weights, i.e $p_{ia} = adjust(\frac{p_i}{w_i})$, respectively.
\begin{figure}[ht]
\begin{center}
	\includegraphics[scale=.48]{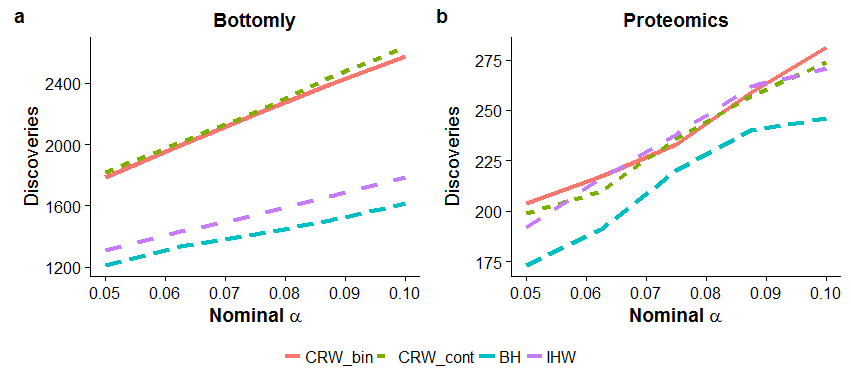}
\end{center}
\caption[Application of the CRW mehod to bottomly and proteomics data ]{\footnotesize{The number of rejected null hypotheses across different significance levels of $\alpha$ for the two datasets: (a) Bottomly and (b) Proteomics. In the legend, CRW\_bin = Proposed Binary, CRW\_cont = Proposed Continuous, BH = Benjamini-Hochberg, and IHW = Independent Hypothesis Weighting methods.}}
\label{fig:real_data_examples}
\end{figure}

Note that, the diagnostic plots of the data suggested that the regression models did not fit well. A pinpoint model diagnosis process can improve the fitness of the models. This leaves a scope of further research, which is beyond the goal of this dissertation. However, the current model information is sufficient for our present purposes, because CRW only requires the centers of the test-effect sizes and the corresponding covariate-effect sizes.

For the Bottomly data, the estimated proportion of the true null hypothesis was $82\%$; mean test-effects were $1.7$ and $2.2$ for the binary and continuous cases, respectively; and the mean covariate-effect was $1.2$ for both the binary and continuous cases. The estimated power is about $1\%$ for the average test effect for the continuous case. The IHW method finds approximately $10\%$ more discoveries than BH, whereas CRW finds greater than $50\%$ more discoveries. For the Proteomics data, the estimated proportion of the true null hypothesis was $66\%$; mean test-effects were $1.8$ and $2$ for the binary and continuous cases, respectively; and the mean covariate-effect was $0.7$ for both the binary and continuous cases. The estimated power is about $1.7\%$ for the average test effect for the continuous case. In this case, IHW and CRW have similar performance, finding about $10\%$ more discoveries than BH. This data set has conditions that are more favorable for the IHW method relative to CRW: the estimated proportion of true nulls is rather low $(66\%)$ and the power at average effect size is higher than for the Bottomly data. Both data sets demonstrate a reality of high throughput data: power is very low to detect most true effects. 
\section{Approximate vs. exact weights}
\hspace{.22in} In this section, we explore how the approximate weights perform compare to the exact weights of the CRW method. To obtain the exact weights, we numerically computed the weights via solving the integral of the objective problem equation (\ref{eq:finalLikelihood}) and to obtain the approximate weights we applied equation (\ref{eq:ContWeight}). Then, we compared them in the following plots (Figure\ref{fig:exact vs. approx weights}). Since the exact approach is computationally very expensive, we applied the exact and the approximate procedures only on the two data sets.
\begin{figure}[ht]
	\begin{center}
		\includegraphics[scale=.62]{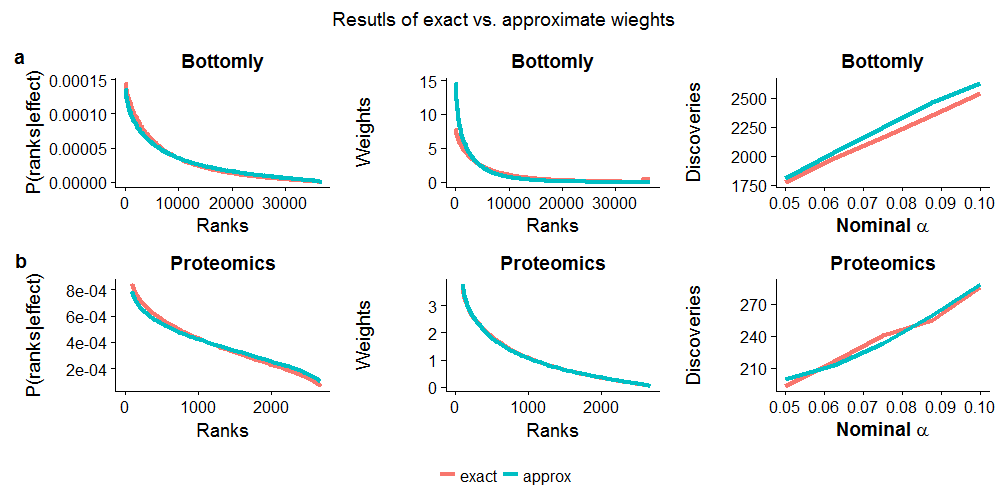}
	\end{center}
	\caption[Comparison of approximate weights and exact weights]{\footnotesize{This Figure shows the comparison between the exact and the approximate weights and the corresponding results. a) Results of the \textbf{Bottomly data:} ranks probabilities, weights at $\alpha = .05$, and the number of rejected tests at the different significance level of $\alpha$. b) Results of the \textbf{Proteomics data:} ranks probabilities, weights at $\alpha = .05$, and the number of rejected tests at the different significance level of $\alpha$.}}
	\label{fig:exact vs. approx weights}
\end{figure}
To apply the exact method, we transformed the covariates by Box-Cox transformation since covariates are not approximately normal. Then, we estimated the true alternative mean and the standard deviation to compute the ranks probabilities and the corresponding weights. In this procedure, we assumed that the transformed covariates are normally distributed. On the other hand, to compute the approximate weights we followed the data analysis steps described in Section \ref{sec:crw_data_steps}. From the results, we see that the approximate version of the weights is very accurate and close to the exact method. The approximate method is also computationally much faster than the exact numerical solution. Therefore, applying the approximate weights will not lead to a significant deviance from the actual outcomes.


\chapter{Gaussian covariate weighting (GCW)}\label{ch:gcw_method} 
\section{Backgrounds}
\hspace{.22in} The high multiple testing burden inherent to high throughput data makes detection of true positive features difficult. As a result, there has been a great deal of interest in independent information to identify the most promising features of the data and reduce multiple testing. Weighted p-values provide perhaps the most promising avenue for doing this.

Consider a situation in which there are $m$ hypothesis tests, and we want to test $H_{0i}:\varepsilon_i=0$ vs. $H_{1i}:\varepsilon_i>0;\ i=1,\ldots,m$. For illustrative purpose, we consider only right-sided test; however, one can easily generalize to two-sided test. Furthermore, we define a set of hypotheses $H=\{H_1,\ldots,H_m\}$, and corresponding test statistics $Z=\{Z_1,\ldots,Z_m\}$ and p-values $P = \{p_1,\ldots,p_m\}$. Then, the simple Bonferroni weighted procedure \citep{spjotvoll1972optimality, holm1979simple, benjamini1997false} for a set of non-negative weights $w=\{w_1,\ldots,w_m\}$ rejects the null hypothesis for tests
\begin{equation}
i \ \epsilon \ R=\Big\{i:\frac{P_i}{w_i} \le \frac{\alpha}{m}\Big\},
\end{equation}
Since under the null hypothesis the p-values are $uniform(0,1)$, the expected number of false rejections is $\sum_{i=0}^m P\big(p_i \le \frac{\alpha w_i}{m}\big) =  \frac{\alpha}{m}  \sum_{i=1}^{m}w_i$. Thus, this weighting scheme will control FWER if the average weight equals to $1$ \citep{roeder2009genome} or $\sum_{i=1}^{m}w_i = m$. In other words, if $R$ and $H_0$ refer to the total number of rejected null and true null hypothesis, respectively, then $R$ controls the Family Wise Error Rate (FWER) if $P(R \cap H_0) \le \alpha$.

\cite{roeder2009genome} and \cite{rubin2006method} considered the model where $\varepsilon_i$  is known exactly from the prior data, and the weights are allowed to depend on $\varepsilon_i$. They showed that the weights that maximize the average power over tests are given by 
\begin{equation}
\hat w_i= \Big(\frac{m}{\alpha}\Big)\bar\Phi \Big(\frac{\varepsilon_i}{2} + \frac{c}{\varepsilon_i}\Big) I(\varepsilon_i>0).
\end{equation}             
This is known as an oracle weight, because it requires knowledge of the effect size. Because this is never known in practice, various authors have proposed methods for using independent data to estimate effect sizes \citep[e.g.][]{westfall2004weighted,ionita2007genomewide,roeder2009genome}.

\cite{rubin2006method} proposed a data splitting technique, in which part of the data is used to estimate effect sizes and the other part is used for hypothesis testing with weights based on the estimated effect sizes. However, \cite{roeder2009genome} showed that the power gain from this weighting procedure cannot compensate the loss of power resulting from the split data. 

\cite{roeder2009genome}(RDW) proposed a method of breaking the hypothesis tests into groups by the covariates and estimating effect sizes for each group from the data and then calculating the weights for the group based on these effect sizes. In order to maintain FWER control, the group sizes must be large enough that individual null features with chance large test statistics will not inflate estimates of effect size, thus boosting them erroneously. However, sufficiently large groups will also likely contain mixtures of true negative and true positive features so that estimates of the effect sizes for the positive features will be diluted and the resulting weights may be poor.
 
\cite{dobriban2015optimal} extended the RDW weights to the case where the effect sizes follow a Gaussian distribution with known mean and variance. These weights allow for uncertainty in effect sizes. \cite{dobriban2015optimal} denote the resulting weights as Bayes Weights (BW). In practice, the mean effect size for each test is estimated from prior data and the variance is chosen by the user.
 
Hasan and Schliekelman (submitted) extended the RDW weights so that the weight for each test is conditioned on the ranking of that test among all tests by an external covariate. They denote these as covariate rank weights (CRW). This covariate is assumed to have an association with effect size so that tests for which the alternate hypothesis is true tend to be ranked higher than null tests. This method requires knowledge of the probabilistic relationship between covariate rank and effect size. Hasan and Schliekleman showed how to calculate this for the case of a normally distributed covariate. This relationship can be estimated for other cases. 

This method proposed in this paper has similarities to both the BW weights and the CRW weights. Rather than conditioning the weights on the covariate rank as in the CRW weights, we condition the weights directly on the covariate. These weights are shown to have a close relationship with the BW weights, including them as a special case. We derive expressions for optimal weights, explore their properties via simulations and applications to data, and show comparisons with other weighting methods. 

\section{Theoretical results}
In the following sections, we provide a brief description of the background of the Gaussian covariate weighting (GCW) and the platform on which the method is built.
\subsection{Theoretical Background}
\hspace{.22in} Consider a multiple hypothesis testing setups. Suppose we want to test $H_{i0}:\varepsilon_i=0$ against $H_{i1}:\varepsilon_i>0; \ i=1,\ldots,m$. In addition, there is vector of test statistics $Z=\{z_1, \dots,z_m\}$ and the corresponding vector of covariates $X = \{x_1,\ldots,x_m\}$. We assume that the conditional distribution of the covariates given the effect sizes is Gaussian in addition to there is Gaussian prior effect sizes. These covariates are easily available from independent data. It is assumed that the $i^{th}$ element of the covariate-vector is related to the $i^{th}$ hypothesis test in a sense that the corresponding covariates will tend to be more informative to more promising tests. Unlike our method, \cite{dobriban2015optimal} assumed that there are Gaussian prior effect sizes but no covariates. 

We will call our proposed method from now Gaussian Covariate Weighting (GCW). The starting point for the GCW method is the optimal weights of \cite{wasserman2006weighted}, \cite{roeder2009genome}. In the articles, the authors showed that the power of tests can be reformulated in terms of the weights as $\beta(w_i;\varepsilon_i) = P\big(Z_i > Z_{\frac{\alpha w_i}{m}}\big) =\bar\Phi\Big(\bar \Phi^{-1}\big(\frac{\alpha w_i}{m}\big)  -  \varepsilon_i\Big)$, then they maximize the following objective problem to obtain the optimal weights: 
\begin{equation}\label{eq:roeder_power}
\begin{aligned}
& \underset{w_k}{\text{maximize}}
& & \sum_{i=1}^{m} \bar \Phi\Big(\bar \Phi^{-1}\big(\frac{\alpha w_i}{m}\big)  -  \varepsilon_i\Big)\\
& \text{subject to}
& & \sum_{i=1}^{m}w_i=m.
\end{aligned}
\end{equation}
Consequently, they were able to obtain a closed form equation of the weights for the known effect sizes, which is
\begin{equation}\label{eq:roeder_weight}
w_i= \Big(\frac{m}{\alpha}\Big)\bar\Phi \Big(\frac{\varepsilon_i}{2} + \frac{c}{\varepsilon_i}\Big) I(\varepsilon_i>0).
\end{equation}

These oracle weights require knowledge of the true effect sizes $\varepsilon_i$, which is not a reasonable requirement. Therefore, \citep{dobriban2015optimal} considered the uncertainty of the effects sizes. They assumed that the effect sizes $\varepsilon_i$ are not fixed but are from Gaussian prior information with mean $\eta_i$ and standard deviation $\sigma_i$, i.e., $\varepsilon_i \sim N(\eta_i, \sigma_i^2)$. Thus, they introduced $\gamma_i^2 = 1 + \sigma_i^2$ and expressed the power as $\beta(w_i;\eta_i,\gamma_i)=\int\bar\Phi\Big(\bar \Phi^{-1}\big(\frac{\alpha w_i}{m}\big)  -  \varepsilon\Big)f(\varepsilon)d\varepsilon = \bar \Phi\Big(\frac{\bar \Phi^{-1}(\frac{\alpha w_i}{m})  -  \eta_i}{\gamma_i}\Big)$. Consequently, they obtained the weights by maximizing the following objective problem:
\begin{equation}\label{eq:BA}
\begin{aligned}
& \underset{w_k}{\text{maximize}}
& & \sum_{i=1}^{m} \bar \Phi\bigg(\frac{\bar\Phi^{-1}\big(\frac{\alpha w_i}{m}\big)  -  \eta_i}{\gamma_i}\bigg)\\
& \text{subject to}
& & \sum_{i=1}^{m}w_i=m.
\end{aligned}
\end{equation}
They maximized the problem with respect to the test statistics $Z_i$ and the effect size $\varepsilon_i$ to obtain the optimal weights, which they referred to the Bayes weights (BW). Since the prior information is Gaussian, they were able to obtain an explicit and efficient solution of the weights for a large $m$ when $\frac{\alpha}{m}$ is sufficiently small and nearly an exact solution for an arbitrary $\frac{\alpha}{m}$: 
\begin{equation}\label{eq:BA_weight}
w_i = \Big(\frac{m}{\alpha}\Big)\bar\Phi\bigg(\frac{-\eta_i \pm \gamma\sqrt{\eta_i^2+2(\gamma_i^2-1)log(\lambda\gamma_i)}}{\gamma_i^2-1}\bigg).
\end{equation} 

The advantage of the method is that it does not over fit the model, hence automatically control the Type-I error \citep{fortney2015genome}. From a practical point of view, when the available information is not strong enough to identify the true features, this method can be very handy to discover a handful number of true features \citep{ignatiadis2017covariate}; however, the BW method is not efficient enough to increase the power quite significantly. On the other hand, we believe that obtaining the probabilistic relationship between the true effect sizes and the independent covariates can surpass this limitation. We derive optimal weights that take into consideration this relationship. We also provide exact mathematical results for estimating this relationship. We show that the new method GCW can outperform other methods in many situations, especially if the number of true alternate tests are moderate to low and the effect sizes are small.

\subsection{Approximate weight} 
\hspace{.22in} Suppose there are $m$ total hypothesis tests with $m_1$ true effects and $m_0 = m-m_1$ null effects. The goal is to test the null hypothesis $H_{0i}:\varepsilon_i=0$ against the alternative hypothesis $H_{1i}:\varepsilon_i>0;i=1, \ldots,m$. There is a test statistic $Z_i$ associated with each hypothesis test and we will assume that this test statistic follows a normal distribution. In addition, each hypothesis test $i$ has an associated covariate $x_i$. This covariate is obtained from independent data, which is believed to be correlated to test statistics under the alternative model and contain information about the effect sizes of the hypothesis tests. It is assumed that the covariates with larger values are more likely to be paired with the true alternate tests. 

Furthermore, we assume that we know the probabilistic relationship between the covariate and the effect size $f(\varepsilon \mid x_i)$. It is possible to calculate this relationship explicitly under certain distributional assumptions, or estimate this empirically. Our goal is to derive p-value weights based on this relationship. Define $\beta(x_i;w_i)$ as the power of the $i^{th}$ test with covariate $x_i$ and weight $w_i$, then after incorporating the test effect size we can reformulate the power as 
\begin{equation}\label{eq:powerLikelihood_gcw}
\beta(w_i;x_i )=\frac{1}{f(x_i)}\int{\bar\Phi\big(Z_{\frac{\alpha w_i}{m}}-\varepsilon\big)f(x_i \mid \varepsilon)f(\varepsilon)I(\varepsilon > 0)d\varepsilon.}
\end{equation}
where $f(x_i \mid \varepsilon)$ is the conditional probability density function of the covariate given the corresponding effect size, and $f(\varepsilon)$  and $f(x_i)$ are the marginal probability density functions of the effect size and the covariate, respectively.  

In order to make this weighting scheme meaningful and valid the average weight needs to be equal to $1$. Therefore, after considering the weighting constraint the objective problem becomes: 
\begin{equation}\label{eq:Likelihood_gcw}
\sum_{i=1}^{m}\frac{1}{f(x_i)}\int{\bar\Phi\big(Z_{\frac{\alpha w_i}{m}}-\varepsilon\big)f(x_i \mid \varepsilon)f(\varepsilon)I(\varepsilon > 0)d\varepsilon} - \delta\Big(\sum_{i=1}^{m}w_i-m\Big),
\end{equation}
where $\delta$ refers to the Lagrange multiplier. Our goal is to find the vector of weights $\bar w$ that maximizes the average power subject to this constraint. Applying Lagrangian optimization and the first order Taylor's series expansion and performing simple algebraic manipulation an approximate version of the weight is obtained: 
\begin{equation}\label{eq:ContWeight_gcw}
w_i \approx \Big(\frac{m}{\alpha}\Big) \bar \Phi \Bigg (\frac{E(\varepsilon)}{2} + \frac{1}{E(\varepsilon)} log\Big(\frac{\delta m f(x_i)}{\alpha P\big(x_i \mid E(\varepsilon)\big)}\Big)\Bigg).
\end{equation}
This weight equation is equivalent to the equation (\ref{eq:ContWeight}) except instead of using the rank of the covariate, here we used covariate itself. From now, we will refer this approximate weight to GCW2. The Lagrange multiplier $\delta \ge 0$, can be obtained by numerical optimization so that $\sum_{i=1}^{m} w_i = m$. We solved $\sum_{i=1}^{m} w_i = m$ for $\delta \ \epsilon \ [0, \infty)$ by Newton-Raphson algorithm to obtain the optimal value of $\delta$. The benefit of this weight is that it does not depend on the distributional assumption of the covariates and can be widely applied. Note that $E(\varepsilon)$ refers to the expected value of the true alternative test-effect sizes, i.e., $E(\varepsilon)= E(\varepsilon \mid \varepsilon >0)$. 

\subsection{Gaussian Covariate Weight (GCW)}
\hspace{.22in} In this section, we will present the incorporation of a Gaussian covariate in addition to the prior information of the effect sizes. We developed the theoretical results along the results showed in \cite{dobriban2015optimal} except we incorporated an additional Gaussian covariate. 
\begin{lemma}\label{lma:power_gcw}
Suppose $X=\{x_1, \ldots, x_m\}$ is a vector of covariates and the conditional density function of the covariate given the effect size $\varepsilon$ is normally distributed, i.e., $x \mid \varepsilon \sim N(\varepsilon, \nu^2)$, and the marginal density of the effect size is $\varepsilon \sim N(\eta, \sigma^2 )$. Then the power for the weighted case becomes
\begin{equation}
\bar \Phi\Bigg(\frac{\bar \Phi^{-1}\big(\frac{\alpha w}{m}\big) - \frac{\eta \nu^2 + x \sigma^2}{\tau^2}}{\frac{\gamma}{\tau}}\Bigg)\frac{1}{f(y)},
\end{equation}\\ \\
where $\gamma^2=\sigma^2 \nu^2+\sigma^2+\nu^2$ and $\tau^2=\sigma^2+\nu^2$; and $f(y)=\frac{f(x )}{f(x;\eta,\tau)}$, the ratio of the prior and the posterior densities of the covariate. 
\end{lemma}
\begin{proof}
From the previous section (equation \ref{eq:powerLikelihood_gcw}) we know the covariate based weighted power:   
\begin{equation}
\beta(w; x) = \frac{1}{f(x)}\int{\bar\Phi\big(Z_{\frac{\alpha w}{m}}-\varepsilon\big)f(x \mid \varepsilon)f(\varepsilon)I(\varepsilon > 0)d\varepsilon}.
\end{equation}
Define a new variable $u := Z_{\frac{\alpha w}{m}}=\bar \Phi^{-1}\big(\frac{\alpha w}{m}\big)$, then the power equation can be reformulated in terms of the standard normal CDF and PDF as
\begin{equation}
\beta(w; x) = \frac{1}{f(x)}\int{\bar\Phi\big(u-\varepsilon\big)\frac{1}{\nu}\phi\bigg(\frac{x-\varepsilon}{\nu}\bigg)\frac{1}{\sigma}\phi\bigg(\frac{\varepsilon-\eta}{\sigma}\bigg)}d\varepsilon.
\end{equation}
Denote by $z = \frac{\varepsilon-\eta}{\sigma}$, then substituting $\frac{\varepsilon-\eta}{\sigma}$ by $z$ introduces the integral of the power as
\begin{equation}
\beta(w; x) = \frac{1}{f(x)}\int{\bar\Phi\big(u-\eta-z\sigma\big)\frac{1}{\nu}\phi\bigg(\frac{x-\eta-z\sigma}{\nu}\bigg)\phi(z)}dz.
\end{equation}
The above integral is differentiable. Therefore, differentiating $\beta$ w.r.t $u$ and rearranging the parameters provides
\begin{equation}\label{eq:power2_gcw}
\frac{d\beta}{du} = -\frac{1}{f(x)}\int \frac{1}{\sqrt{8\pi^3\nu^2}}e^{ -\frac{1}{2}\left[(u-\eta-z\sigma)^2+\frac{(x-\eta-z\sigma)^2}{\nu}+z^2\right]} dz
\end{equation}
In order to obtain a closed form solution we need to rearrange the parameters. Define by $D:=(u-\eta-z\sigma)^2+\frac{(x-\eta-z\sigma)^2}{\nu}+z^2$, then after performing algebraic manipulation we obtain
\begin{equation}\label{eq:D}
D = \frac{(\nu^2\sigma^2+\sigma^2+\nu^2)}{\nu^2}\left[z^2-2z\frac{(\nu^2u\sigma-\nu^2\sigma\eta-\sigma\eta+x\sigma)}{(\nu^2\sigma^2+\sigma^2+\nu^2)}+\frac{\nu^2(u-\eta)^2+(x-\eta)^2}{(\nu^2\sigma^2+\sigma^2+\nu^2)}\right]
\end{equation}
We further rearrange equation $D$ to obtain a kernel of the normal PDF. Let us denote by $E:=\frac{(\nu^2u\sigma-\nu^2\sigma\eta-\sigma\eta+x\sigma)}{(\nu^2\sigma^2+\sigma^2+\nu^2)}$. Then plugging $E$ into equation (\ref{eq:D}) and then $D$ into equation (\ref{eq:power2_gcw}) introduces
\begin{equation}\label{eq:int_power}
\frac{d\beta}{du} = -\frac{1}{f(x)}\int \frac{1}{\sqrt{8\pi^3\nu^2}}e^{-\frac{1}{2}\frac{(\nu^2\sigma^2+\sigma^2+\nu^2)}{\nu^2}\left[z^2-2zE+\frac{\nu^2(u-\eta)^2+(x-\eta)^2}{(\nu^2\sigma^2+\sigma^2+\nu^2)}\right]}dz.
\end{equation}
Since integrating of the normal PDF over the entire limit is $1$, therefore,    
\begin{equation}\label{eq:int_gcw}
\int \frac{1}{\sqrt{2\pi}\sqrt{\big(\frac{\nu}{\nu^2\sigma^2+\sigma^2+\nu^2}\big)}}e^{-\frac{1}{2}\frac{(\nu^2\sigma^2+\sigma^2+\nu^2)}{\nu^2}}(z-E)^2dz = 1.
\end{equation}
Thus we can extract the integral (\ref{eq:int_gcw}) out from equation (\ref{eq:int_power}), which leaves to
\begin{equation}\label{eq:int_u}
\frac{d\beta}{du} = -\frac{1}{f(x)}\frac{1}{2\pi\sqrt{\nu^2\sigma^2+\sigma^2+\nu^2}}e^{-\frac{1}{2}\frac{(\nu^2\sigma^2+\sigma^2+\nu^2)}{\nu^2}\left[\frac{\nu^2(u-\eta)^2+(x-\eta)^2}{(\nu^2\sigma^2+\sigma^2+\nu^2)} - E^2\right]}.
\end{equation} 
Similarly rearranging equation (\ref{eq:int_u}) reveals a part which is a normal PDF. Then, integrating w.r.t $u$ makes the normal PDF to $1$ and leaves
\begin{equation}
\beta(w;x)=\frac{1}{f(x)}\bar \Phi \Bigg(\frac{u-\frac{\nu^2\eta+\sigma x}{\sigma^2+\nu^2}}{\sqrt{\frac{\nu^2\sigma^2+\sigma^2+\nu^2}{\sigma^2+\nu^2}}} \Bigg).\frac{1}{\sqrt{\sigma^2+\nu^2}}\phi\bigg(\frac{x-\eta}{\sqrt{\sigma^2+\nu^2}} \bigg)
\end{equation}
Finally introducing $\gamma^2=\sigma^2 \nu^2+\sigma^2+\nu^2$ and $\tau^2=\sigma^2+\nu^2$; and $f(y)=\frac{f(x )}{f(x;\eta,\tau)}$, where $f(x;\eta,\tau)=\frac{1}{\tau}\phi(\frac{x-\eta}{\tau})$ finishes the proof.
\end{proof}
Now we need to maximize the expected power with respect to $w_k$ in order to obtain the optimal weight subject to $\sum_{i=1}^{m}w_i=m$. Thus, after considering the above Lemma \ref{lma:power_gcw} the objective problem of Gaussian Covariate Weight (GCW) (equation \ref{eq:powerLikelihood_gcw}) becomes

\begin{equation}\label{eq:objective_gcw}
\begin{aligned}
& \underset{w_k}{\text{maximize}}
& & \sum_{i=1}^{m} \bar \Phi\Bigg(\frac{\bar \Phi^{-1}\big(\frac{\alpha w_i}{m}\big) - \frac{\eta_i \nu_i^2 + x_i \sigma_i^2}{\tau_i^2}}{\frac{\gamma_i}{\tau_i}}\Bigg)\frac{1}{f(y_i)}\\
& \text{subject to}
& & \sum_{i=1}^{m}w_i=m.
\end{aligned}
\end{equation}
For convenience, we can assume that the prior and the posterior distributions of the covariates are the same or the ratio of the random variable is $y_i \sim uniform(0,1)$, thus $f(y_i )=1$, to make the above expression simple. However, this adaption may introduce loss of vital information. We will explore this later in the simulation section.  

The GCW method is based on the covariates, i.e., we estimate the optimal weight and the corresponding optimal power from the probabilistic relationship between the effect sizes and the covariates. If we allow the distribution of the effect sizes $f(\varepsilon)$ and the distribution of the covariates $f(x_i)$ to be $uniform(0,1)$, then the objective problem of the GCW method becomes exactly the same as the objective problem of BW method disused in \cite{dobriban2015optimal}. This is evident from the next result.
\begin{lemma}
Suppose the distribution of the covariate $x$ and the conditional distribution of the covariate given the effect size $x \mid \varepsilon$ are  $uniform(0,1)$ and the effect size $\varepsilon \sim N(\mu, \sigma^2)$, then the objective problem of GCW (\ref{eq:objective_gcw}) reduces to the objective problem of BW (\ref{eq:BA}). 	
\end{lemma}
\begin{proof}
If the given conditions hold then equation (\ref{eq:powerLikelihood_gcw}) reduces to 
\begin{equation}\label{eq:gcw_for_uniform}
\beta(x;w)=\int_\varepsilon{\bar\Phi\big(Z_{\frac{\alpha w}{m}}-\varepsilon\big)\frac{1}{\sigma}\phi \bigg(\frac{\varepsilon-\eta}{\sigma} \bigg) I(\varepsilon > 0)d\varepsilon,}
\end{equation}
where $\Phi = 1-\bar \Phi$ and $\phi$ refer to the standard normal CDF and PDF, respectively. Let, $c:= Z_{\frac{\alpha w}{m}} = \bar \Phi^{-1}(\frac{\alpha w}{m})$ and $y = \frac{\varepsilon-\eta}{\sigma}$. Then after differentiating $\beta$ w.r.t $c$ and rearranging the parameters, equation (\ref{eq:gcw_for_uniform}) produces

\begin{equation}
\int_y\frac{1}{\sqrt{\frac{1}{1+\sigma^2}}}\phi\Bigg(\frac{y - \frac{c \sigma-\sigma \eta}{1+\sigma^2}}{\sqrt{\frac{1}{1+\sigma^2}}}\Bigg).\frac{1}{\sqrt{1+\sigma^2}}\phi\bigg(\frac{c - \eta}{\sqrt{1+\sigma^2}}\bigg)dy.
\end{equation}
Both the first and the second parts under the integral are the standard normal PDFs. Integrating w.r.t $y$ reduced the first part to $1$. Consequently, integrating w.r.t $c$ leads to the following objective problem:
\begin{equation}
\bar \Phi\bigg(\frac{c - \eta}{\sqrt{1+\sigma^2}}\bigg).
\end{equation}
This finishes the proof.
\end{proof}
The above result is exactly the same objective problem that \cite{dobriban2015optimal} (\ref{eq:BA}) introduced when the effect sizes are from the Gaussian prior distribution. This is more obvious in terms of application. For instance, the author showed in the data application that the mean of the Gaussian prior is $\eta_i=t_{i0}$ if the sample sizes of the prior and the posterior study are the same, where $t_{i0}$ refers to the prior test statistics. In our case, $t_{i0}$ is the covariate, i.e.,  $x_i = t_{i0}$. 

The fundamental difference is that BW requires prior information for each effect size, i.e., $\eta_i$ and $\sigma_i^2;i=1,\ldots,m,$ whereas GCW, essentially, does not require prior information of the effect sizes, which are $\eta_i, \sigma_i^2$. However, the GCW requires the covariates and would provide better outcomes if the prior information for each effect sizes are available. In addition, GCW can use generic information about the prior effect sizes instead of all effect sizes such as $\eta, \sigma^2$ and variance of the covariate $\tau^2$. These simple parameter estimates are easily accessible as opposed to the prior information for each test. 

The BW method also heavily depends on the prior and the posterior sample sizes, and the simulation and data examples \citep{dobriban2015optimal} show that the posterior sample sizes need to be higher than the prior sample size in order to obtain the improved power. In contrast, the GCW method can obtain the same power as the BW method by simply assuming a flat prior. GCW also does not require the sample sizes and the estimation of the prior effect sizes for all tests; however, allowing Gaussian prior in addition to the covariate certainly will increase the power. 

The main advantage of the GCW method is that it creates a bridge between the prior information of the effect sizes and the p-values via covariates. Thus, unlike BW, GCW can overcome the adverse situations, for an example, when the prior information is not Gaussian. We will explore this in the simulation. The GCW method also leads to an idea of incorporating two sets of independent information. This is possible if we incorporate different prior information for the different effect sizes (like the BA method \citealt{dobriban2015optimal}) instead of a generic prior. In that situation, we will allow $\eta_i$ and $\sigma_i^2$ in addition to the covariates $x_i; \ i=1,\ldots,m$. This may lead to a new perspective of analyzing multiple hypothesis tests. 
 
We solve the GCW objective problem (equation \ref{eq:objective_gcw}) efficiently for a large $m$. The objective problem is solved efficiently when $\frac{\alpha}{m}$ is sufficiently small and approximately for an arbitrary  $\frac{\alpha}{m}$. In the next result we show the optimal weight of the GCW method.
\begin{theorem}
If $\lambda \ge 1$ and $\alpha \ \epsilon \ (0,1)$ then the optimal weights that maximize the average power is $w_i= \big(\frac{m}{\alpha}\big)\bar \Phi\big(u_i(x_i, \eta_i, \nu_i, \sigma_i; \lambda)\big)$, where $\lambda$ is a tuning parameter so that $\sum_{i=1}^{m}w_i=m$, and
\begin{equation}
u_i = \frac{-\bigg(\frac{\eta_i \nu_i^2 + x_i \sigma_i^2}{\tau_i^2} \bigg) + \frac{\gamma_i}{\tau_i}\sqrt{\bigg(\frac{\eta_i \nu_i^2 + x_i \sigma_i^2}{\tau_i^2}\bigg)^2 + 2(\frac{\gamma_i^2}{\tau_i^2}-1)\log\bigg(\frac{\lambda \gamma_i m f(y_i) }{\alpha \tau_i}\bigg)}}{\frac{\gamma_i^2}{\tau_i^2}-1}. 
\end{equation} 
\end{theorem}
\begin{proof}
The objective problem \ref{eq:objective_gcw} can be expressed in terms of the Langrangian maximization problem such that
\begin{equation}\label{eq:objective_gcw2}
\begin{aligned}
& \underset{w_k}{\text{maximize}}
& & \sum_{i=1}^{m} \bar \Phi\Bigg(\frac{\bar \Phi^{-1}\big(\frac{\alpha w_i}{m}\big) - \frac{\eta_i \nu_i^2 + x_i \sigma_i^2}{\tau_i^2}}{\frac{\gamma_i}{\tau_i}}\Bigg)\frac{1}{f(y_i)} - \lambda \bigg(\sum_{i=1}^{m}w_i-m \bigg).
\end{aligned}
\end{equation}
The idea is to study the penalized objective problem for the different $\lambda$ and find a suitable $\lambda$ that maximizes $w_i(\lambda)$ and satisfy the constraint. Let us define the above function by $f(w_i)$ and consider $w_k$ for a specific test $k$. Then, we need to maximize the expression with respect to $w_k$. Differentiating $f(w_i)$ with respect to $w_k$ produces
\begin{equation}\label{eq:objective_gcw3}
\begin{aligned}
\frac{df}{dw_k} = \frac{\phi\Bigg(\frac{\bar \Phi^{-1}\big(\frac{\alpha w_i}{m}\big) - \frac{\eta_i \nu_i^2 + x_i \sigma_i^2}{\tau_i^2}}{\frac{\gamma_i}{\tau_i}}\Bigg)\frac{\alpha}{m}\frac{1}{f(y_i)}}{\phi\bigg(\bar \Phi^{-1}\big(\frac{\alpha w_i}{m}\big)\bigg)\frac{\gamma_i}{\tau_i}}  - \lambda
\end{aligned}
\end{equation}
where $\phi(.)$ is the standard normal density. Let us denote a new variable $u_i := \bar \Phi^{-1}\big(\frac{\alpha w_i}{m}\big)$. Then, the derivative is maximized, i.e., $\frac{df}{dw_k} \ge 0$ if and only if 
\begin{equation}\label{eq:quadratice_ineq}
\begin{aligned}
u_i^2 - \Bigg(\frac{u_i^2 - \frac{\eta_i \nu_i^2 + x_i \sigma_i^2}{\tau_i^2}}{\frac{\gamma_i}{\tau_i}}\Bigg)^2 \ge 2\log\bigg(\frac{\lambda \gamma m f(y_i) }{\alpha \tau}\bigg).
\end{aligned}
\end{equation}
This is a quadratic inequality. Therefore, after performing simple algebra, we obtain the weights:
\begin{equation}
w_i= \big(\frac{m}{\alpha}\big)\bar\Phi(u_i) 
\end{equation}
where,
\begin{equation}\label{eq:roots_gcw}
u_{ik} = \frac{-\bigg(\frac{\eta_i \nu_i^2 + x_i \sigma_i^2}{\tau_i^2} \bigg) \pm \frac{\gamma_i}{\tau_i}\sqrt{\bigg(\frac{\eta_i \nu_i^2 + x_i \sigma_i^2}{\tau_i^2}\bigg)^2 + 2(\frac{\gamma_i^2}{\tau_i^2}-1)\log\bigg(\frac{\lambda \gamma_i m f(y_i) }{\alpha \tau_i}\bigg)}}{\frac{\gamma_i^2}{\tau_i^2}-1}; \ k=1,2. 
\end{equation}
$u_k$  has two roots $u_1 < u_2$. By assumption, $\frac{\gamma_i^2}{\tau_i^2} > 1$ and $\lambda > 0$. For $\lambda \ge 1$, both roots are real, and for $0 < \lambda < 1$, the roots below a certain point of $\lambda$ are not real.

Let us denote the generic Lagrangian function of the objective problem as   
\begin{equation}
f(u) = \bar \Phi\Bigg(\frac{u - \frac{\eta \nu^2 + x \sigma^2}{\tau^2}}{\frac{\gamma}{\tau}}\Bigg)\frac{1}{f(y)} - \lambda \frac{m}{\alpha}\bar \Phi(u)  .
\end{equation}

Consider the situation when $\lambda \ge 1$. For $\lambda \ge 1$, the larger root is given by
\begin{equation}
u_2 = \frac{-\bigg(\frac{\eta_i \nu_i^2 + x_i \sigma_i^2}{\tau_i^2} \bigg) + \frac{\gamma_i}{\tau_i}\sqrt{\bigg(\frac{\eta_i \nu_i^2 + x_i \sigma_i^2}{\tau_i^2}\bigg)^2 + 2(\frac{\gamma_i^2}{\tau_i^2}-1)\log\bigg(\frac{\lambda \gamma_i m f(y_i) }{\alpha \tau_i}\bigg)}}{\frac{\gamma_i^2}{\tau_i^2}-1}, 
\end{equation}
which is the global maximum, because $f$ is increasing on $\left(-\infty, u_1 \right] $ and $\left[u_2, \infty\right)$, and decreasing on $\left[ u_1, u_2\right] $. Furthermore, $f < 0$ at $-\infty$ if $f(y)$ is a probability density function (PDF) and $f = 0$ at $\infty$. Therefore, $u_2$ is the  global maximum of $f$. Thus, $f(u)$ is maximized at $u_2$. Consequently, to obtain the optimal weight, we need to find $\lambda \ge 1$ so that $\sum_{i=1}^{m}w_i=m$.
This finishes the proof
\end{proof}

For $\lambda \ \epsilon \ \left[0,1\right)$ the roots are complex if $\lambda < l(\eta,\sigma,\nu,y)$, where $l(\eta,\sigma,\nu,y)$ is the upper limit of $\lambda$. The values below this limit will produce complex roots. The roots will be complex due to the negativity of the discriminant of the quadratic equation, i.e., if
\begin{equation}
\bigg(\frac{\eta \nu^2 + x \sigma^2}{\tau^2}\bigg)^2 + 2\Big(\frac{\gamma^2}{\tau^2}-1 \Big)\log\bigg(\frac{\lambda \gamma m f(y) }{\alpha \tau}\bigg)
 < 0 
\end{equation}
which is equivalent to the following upper bound of $\lambda$ in terms of the new function $l(\eta,\sigma,\nu,y)$:
\begin{equation}
\lambda < l(\eta,\sigma,\nu,y) := \frac{\alpha \tau}{\gamma m f(y)} exp\Bigg(-\frac{\frac{\eta \nu^2 + x \sigma^2}{\tau^2}}{2(\frac{\gamma^2}{\tau^2}-1)}\Bigg)
\end{equation}
Therefore, if $\lambda < l(\eta,\sigma,\nu,y)$, then the supremum of $f(w)$ is at $-\infty$; otherwise, there are two supremums: $u=-\infty$ and $u=u_2$.  We want $u_2$ as large as possible, and this is equivalent to evaluating $f$ at $\lambda=1$. By Lemma $9.2$ in \citep{dobriban2015optimal}, it can be shown that there is a unique value of $\lambda$ such that $\lambda \ \epsilon \ \left[l(\eta,\sigma,\nu,y_k),1 \right]$. We obtain $\lambda$ by applying the Newton-Raphson method when $f(\lambda) > 0$ ; otherwise, we applied Brent’s algorithm to obtain $\lambda$ so that $\sum_{i=1}^{m}w_i=m$. More specifically, to obtain the weights and $\lambda$, we reparametrized the roots $u_2$  and followed the similar procedures described in \citep{dobriban2015optimal}, especially we followed the Section 9.5.2 from the paper. 

A special case of the GCW weight can be presented. Let $\sigma \longrightarrow 0$; then $\frac{\gamma^2}{\tau^2} \longrightarrow 1$. Thus, equation (\ref{eq:quadratice_ineq}) reduces to
\begin{equation}
u_i \ge \frac{\eta_i}{2} + \frac{1}{\eta_i} \log\bigg(\frac{\lambda mf(y_i)}{\alpha}\bigg).
\end{equation} 
Therefore, in this particular context, the weight depends on the ratio of the prior and the posterior densities functions of the covariate $f(y_i)$ in addition to the prior mean $\eta_i$ of the effect sizes. Furthermore, if the prior and the posterior density are assumed to be the same, then this weight reduces to the \citet{spjotvoll1972optimality} or to the \cite{roeder2009genome} weights.

\section{Simulation study}\label{sec:simulation_gcw}
\hspace{.25in} In this section, we present the performance of the proposed GCW method via simulations. To conduct the simulations we considered various combinations of test parameters, effect sizes, covariates etc.; however, we present only a portion of the simulations. 

We used simulations to compare the statistical performance of four methods: the proposed GCW and GCW2 (approximation of GCW) methods, the Benjamini and Hochberg FDR procedure with no weighting (BH) \citep{benjamini1995controlling}, and the Bayes weight (BW) method \citep{dobriban2015optimal}. To obtain the results of GCW and GCW2, we followed the similar procedure described in the following data analysis section \ref{sec:gcw_data_steps}; in order to obtain the results corresponding to the BH method, we used the $p.adjust$ function from $R$ software and followed the Benjamini and Hochberg FDR procedure \citep{benjamini1995controlling}; and in order to implement the BW method, we applied the $bayes\_weights$ function from the $R$ package $pweight$ and maintained the default settings.

Simulated data were generated for various scenarios and used to estimate power, FDR, and FWER for each method. The simulations were divided into three groups based on the proportion of the true null hypothesis. The three groups were composed of $\pi_0 =\{50\%, 90\%, 99\%\}$ true null tests. For each group of simulations, we considered two types of covariates: 1) normally distributed test statistics and 2) pooled variances of two sample t-test statistics, where the sample size per test was $n=5$. We considered equal sample sizes. Then for each type of covariates we considered a set of effect sizes $ = \{0,.2,.4,.6,.8, 1, 2, 3, 5, 8\}$. 

If we were interested in observing the performance of the normally distributed covariates, then the vector of the effect sizes was considered as the vector of mean effect sizes; i.e.; $E(\tau_i) = \{0,.2,.4,.6,.8, 1, 2, 3, 5, 8\}$; otherwise, the vector was considered as the vector of the differences between the effect sizes of the two samples; i.e.; $E(\tau_{i1} - \tau_{i2}) = \{0,.2,.4,.6,.8, 1, 2, 3, 5, 8\}$. 

Note that, for simplicity, we kept the mean test effect and the mean covariate effect the same and assumed that $f(y_i)$ is the PDF of $uniform(0,1)$. We conducted a simulation of $1,000$ replicates and assumed that there were $m=10,000$ hypotheses tests. For each replicate, we determined the proportion of true positive (for power) and false positive (for FDR) results and calculated the average over the replicates. For FWER, we calculated the number of false positives and averaged across the replicates. 

\subsection{Power}
\begin{figure}[ht]
	\begin{center}
		\includegraphics[scale=.58]{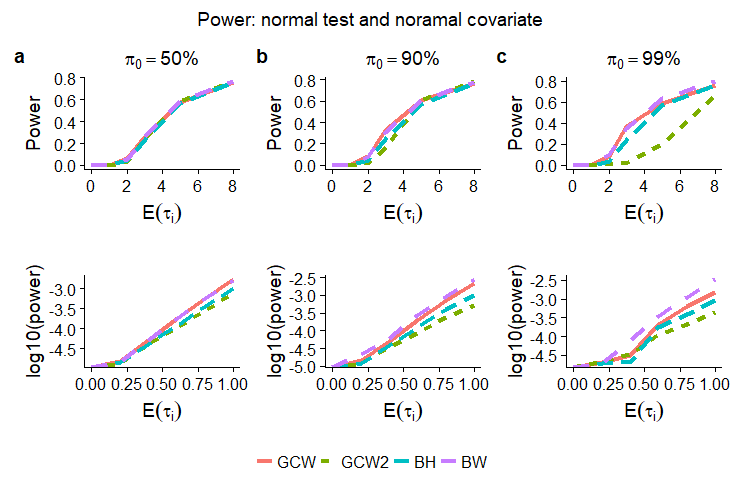}
	\end{center}
	\caption[Power of the GCW method for the different effect sizes when both the covariates and the test statistics are normal]{\footnotesize{This Figure shows the simulated power of four methods when both the test and the covariate effect sizes are normal. Each plot consists of four curves that are based on proposed GCW and GCW2, Benjamini-Hochberg (BH), and Bayes weights (BW) methods. The first row shows the power for low to high effect sizes, and the second row shows power for the low effect sizes. Three columns are based on three groups composed of $50\%,90\%$, and $99\%$ true null hypothesis.}}
	\label{fig:power_normal}
\end{figure}
\hspace{.22in} Figure \ref{fig:power_normal} shows the power of the different methods when the covariates are Gaussian. If the number of the true null hypothesis is low (Figure \ref{fig:power_normal}a), all methods perform almost equally; however, for all other situations, GCW and BW perform better than BH and GCW2, which is expected since GCW and BW are based on the Gaussian covariates. For a large number of true null hypotheses GCW2's performance is below average. The overall performance of the methods stayed the same across low to high effect sizes. 

\begin{figure}[ht]
	\begin{center}
		\includegraphics[scale=.58]{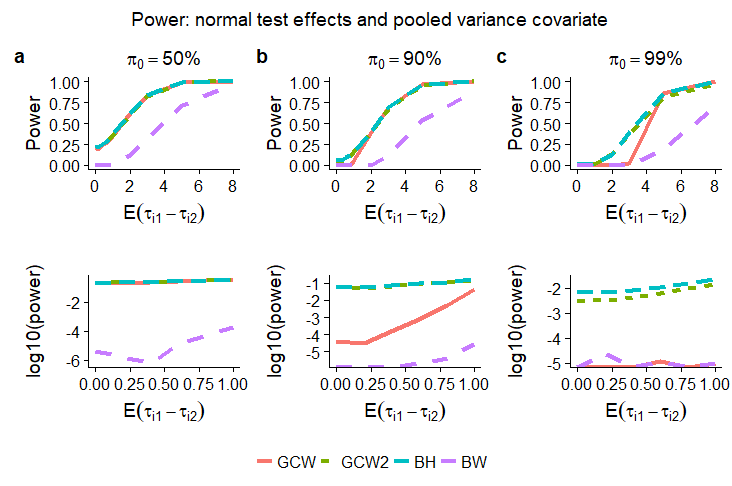}
	\end{center}
	\caption[Power of the GCW method for the different effect sizes when the covariates are pooled variances but the test statistics are normal]{\footnotesize{This Figure shows the simulated power of four methods when the test effect size is normal but the covariate is the pooled variance. Each plot consists of four curves that are based on the proposed GCW and GCW2,  Benjamini-Hochberg (BH),  and Bayes Weights (BW) methods. The first row shows the power for low to high effect sizes, and the second row shows power for the low effect sizes. The three columns are based on three groups composed of $50\%,90\%$, and $99\%$ true null hypothesis.}}
	\label{fig:power_poolVar}
\end{figure}
Figure \ref{fig:power_poolVar} shows the power of the different methods when the covariates are not Gaussian as they are pooled variances of the corresponding two sample t-tests. For the moderate to the low proportion of the true null hypothesis, GCW and GCW2 performed similarly to the BH method. Although the performance of GCW2 stayed pretty much the same as BH method, GCW did not perform well for the low effect sizes and very high proportion of the true nulls. On the other hand, the performance of the BW method is below average regardless of the proportion of the true nulls and the effect sizes except for the very low effects and proportion of the true nulls. This result is expected because the BW method assumed that the prior effect sizes are Gaussian; therefore, it could not perform well when the prior is not Gaussian. However, the GCW method performed comparatively better than BW in this adverse scenario because GCW used covariates statistics, therefore robust to misspecified prior information. The covariates also create a bridge between the test statistics and the prior effect sizes and thus can overcome the adverse situations. The performance of GCW and GCW2 would be much better, if we would have known the PDFs $f(y_i)$ and $f(x_i)$. 

\subsection{FWER}
\begin{figure}
	\begin{center}
		\includegraphics[scale=.5]{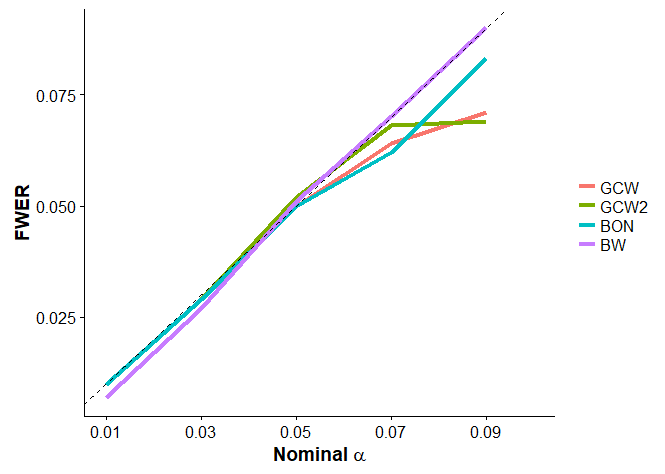}
	\end{center}
	\caption[Type-I error rate of the GCW method]{\footnotesize{This Figure shows the FWER for different significance levels of $\alpha$. In the legend, the representation is: GCW = Gaussian covariate, GCW2 = GCW approximation, BON = Bonferroni, , and BW = Bayes weights.}}
	\label{fig:fwr_gcw}
\end{figure}
\hspace{.22in} We conducted simulations to verify that the methods control the Family Wise Error Rate (FWER). For the simulation, we conducted $1,000$ replications, and for each replication, we generated a data set composed of $m=10,000$  observations of three variables: 1) test statistics, 2) p-values, and 3) covariate statistics. We generated two sets of p-values from $unifrom(0,1)$ then one set of p-values is used to compute the covariate statistics and the other set is used to compute the test statistics. We converted the p-values into the statistics by the inverse of the standard normal CDF, i.e., $T = \Phi^{-1}(1-P)$. We obtained the significant tests by following the data application procedures described below. Since all test statistics are from the null models, we expect the number of false positive to be below the significance level of $\alpha$. From Figure \ref{fig:fwr_gcw}, we see the GCW and GCW2 methods control the FWER.

\section{Data application}
\hspace{.22in} In this section, we showed a data example. We applied two GWAS data sets: SCZ and Lipid. The SCZ data set \citep{schizophrenia2011genome} was downloaded from \url{http://www.med.unc.edu/pgc/downloads} and the Lipid data set \citep{teslovich2010biological} was downloaded from \url{http://www.broadinstitute.org/mpg/pubs/lipids2010/TC_ONE_Eur.tbl.sorted.gz}. \cite{schizophrenia2011genome} and \cite{teslovich2010biological} conducted the studies to identify the SNPs that are associated to schizophrenia and bipolar disorder, respectively. \cite{andreassen2013improved} showed an improvement of the power by applying their own method called controlled the Bayesian conditional false discovery rate on these data sets. This pair of data sets is also used in \cite{dobriban2015optimal} paper to show their proposed method's (BW) performance. A summary of the data sets is provided in Figure \ref{fig:lipid_scz_summary}. 

We considered p-values from the both data sets as covariates and statistics interchangeably to observe the effect of changing the covariates. There were $m =1,252,901$ SNPs and $N = 21,856$ observations per SNPs to conduct a hypothesis tests in the SCZ data set. In the Lipid data set, there were $m =2,692,560$ SNPs and $N = 96,598$ observations per SNPs to conduct a hypothesis tests. However, we used $m=1,162,376$ SNPs that are common in the both data sets to produce the analyses. 
\begin{figure}[ht]
	\begin{center}
		\includegraphics[scale=.53]{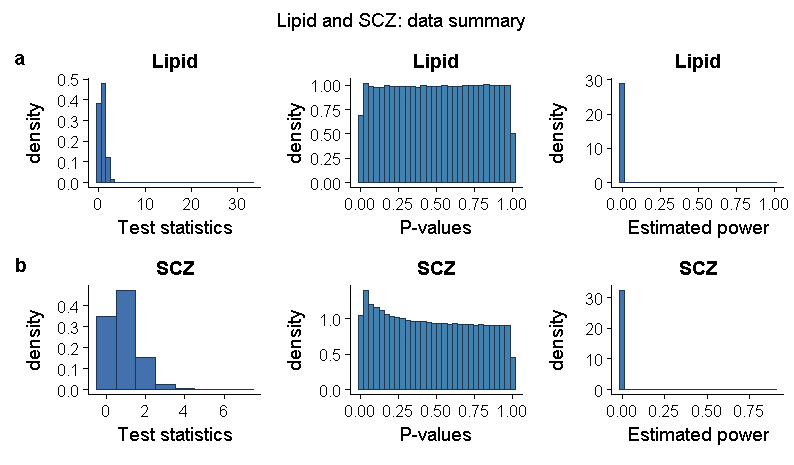}
	\end{center}
	\caption[Lipid and SCZ data pre screening summary]{\footnotesize{This Figure shows the summary information of the Lipid and SCZ data. First row (left-right) shows the distribution of the test statistics, p-values, and estimated power of the Lipid data; and the second row shows the distribution of the test statistics, p-values, and estimated power of the SCZ data.}}
	\label{fig:lipid_scz_summary}
\end{figure}

When the p-values of Lipid data were used as the covariates, the estimated proportion of the true null hypothesis was $90\%$; the mean and the standard deviation of the estimated test-effects were $\eta=2.4$ and $\sigma=0.49$, respectively; and the standard deviation of the estimated covariate-effects was $\tau=0.49$. Furthermore, the estimated power was about $0.2\%$ for the average test effect size. When the p-values of SCZ data were used as the covariates, the estimated proportion of the true nulls and the estimated power did not change. However, the mean and the standard deviation of the estimated test-effects were $\eta=2.2$ and $\sigma=0.88$, respectively; and the standard deviation of the estimated covariate-effects was $\tau=1.0$.

Since both studies \cite{schizophrenia2011genome} and \cite{teslovich2010biological} are independent, we can interchangeably use their p-values as independent covariates udder the null hypothesis. However, the p-values of the data sets are correlated under the alternatives because schizophrenia and bipolar disorder have overlapped polygenic signals, which allows to discover common SNPs that are associated to schizophrenia and bipolar disorder.

\subsection{Application of GCW and GCW2 methods}\label{sec:gcw_data_steps}
\hspace{.22in} In order to apply the proposed methods (both GCW and GCW2) we followed the several steps. For the reader’s convenience, the steps are summarized below:

\textbf{Input:} A nominal significance level $\alpha \ \epsilon \ (0,1),$ a vector of p-values $P = \{p_1, \ldots,p_m\}$ and a vector of the covariates $Y=\{y_1, \ldots,y_m\}$ corresponding to each p-value.

\textbf{Step 1:} Denote the vector of test statistics by $T$. Obtain the test statistics by $T = \Phi^{-1}(1-P)$ for the one-tailed and $T=\Phi^{-1}(1-P/2)$ for the two-tailed p-values.

\textbf{Step 2:} Denote by $m_1$ the number of true alternative tests. Estimate $m_1$ by using the method of \cite{storey2003statistical}. 

\textbf{Step 3:} Sort the test statistics in the decreasing order and pick the top $m_1$ tests. This set of tests is used to estimate the alternative test effects.

\textbf{Step 4:} Denote by $\eta$ and $\sigma$ the mean and standard deviation of the true alternative test statistics. Calculate these from the top $m_1$  tests from Step 3.

\textbf{Step 5:} Denote by $\tau$ and $\nu$ the standard deviations of the covariates ($x$) and the conditional covariates given the effect sizes ($x \mid \varepsilon$). Estimate $\nu^2$ by $\nu^2=\tau^2 - \sigma^2$.

\textbf{Step 6:} Reparameterize by $\mu = \frac{\eta \nu^2 + x \sigma^2}{\tau^2}$ and $\sigma^2 = \frac{\gamma^2}{\tau^2}-1$, and assume $f(y)=1$. Apply $bayes\_weights$ function from the $R$ package $pweight$ to compute the weights ($w_i$) for the GCW method.    

\textbf{Step 7:} For the GCW2 method, define a simple linear regression by $y_i = \beta_0+\beta_1 t_i+\epsilon_i;1=1,\ldots,m$. Obtain the relationship between the test statistics $(t)$ and the covariates $(x)$ by the linear regression. 

\textbf{Step 8:} Apply the estimated model $\hat{y_i} = \hat \beta_0 + \hat \beta_1 t$ to compute the predicted covariate $\hat y$ for the mean test statistics $\eta$.  

\textbf{Step 9:} Apply $dnorm$ function from $R$ software to compute the probability $P\big(x \mid E(\varepsilon)\big)$, where the mean is $E(\varepsilon)=\hat y$ and the standard deviation is $\tau$.  

\textbf{Step 10:} Compute weights ($w_i$) for the GCW2 method by applying equation (\ref{eq:ContWeight_gcw}), where $E(\varepsilon)$ will be replaced by $\eta$.

\textbf{Step 11:} Apply the weighted Bonferroni method $\sum_{i=1}^{m}I(p_i \le \frac{\alpha w_i}{m})$ or BH \citep{benjamini1995controlling} weighted FDR procedure $\sum_{i=1}^{m}I(p_{ia} \le \alpha)$ to obtain the number of significant tests, where $I(.)$ and $p_{ia}$ refers to the indicator function and the weighted adjusted p-values, i.e., $p_{ia} = adjust(\frac{p_i}{w_i})$, respectively.
\begin{table}[ht]
	\caption[Significant SNPs by the GCW method when p-values of Lipid data are covariates]{\footnotesize{Significant SNPs when covariates are Lipid data}}
	\begin{center}
		\begin{tabular}{lrrrrr}
			\hline
			Methods/Sig. level & .05 & .0625 & .075 & .0875 & 0.1 \\ 
			\hline
			GCW & 232 & 245 &  256 &  264 & 272 \\ 
			GCW2 & 39 &  45 &  48 &  51 & 57 \\ 
			BON & 116 & 126 & 131 & 136 & 145 \\ 
			BW & 227 & 235 & 246 & 253 & 261 \\ 
			\hline
		\end{tabular}
	\end{center}
\footnotesize{\parbox[t]{6in} {To identify the SNPs, we applied the four methods. There were $1,162,376$ common SNPs in these two data sets. The estimated proportion of the true null hypothesis was $90\%$; the mean and the standard deviation of the estimated test effects were $\eta=2.4$ and $\sigma=0.49$, respectively; and the standard deviation of the estimated covariate effects was $\tau=0.49$. The estimated power was about $0.2\%$ for the average test effect size.}}
\label{table:sig_SNPs}
\end{table}
\begin{table}[ht]
	\caption[Significant SNPs by the GCW method when p-values of SCZ data are covariates]{\footnotesize{Significant SNPs when covariates are SCZ data}}
	\begin{center}
		\begin{tabular}{lrrrrr}
			\hline
			Methods/Sig. level & .05 & .0625 & .075 & .0875 & 0.1 \\ 
			\hline
		GCW & 668 & 695 & 715 & 736 & 756 \\ 
		GCW2 & 41 &  46 &  49 &  53 & 60 \\
		BON & 855 & 877 & 893 & 900 & 917 \\ 
		BW & 528 & 561 & 572 & 586 & 597 \\ 
			\hline
		\end{tabular}
	\end{center}
\footnotesize{\parbox[t]{6in} {The number of common SNPs, the estimated proportion of the true nulls, and the estimated power did not change from the previous Table. However, the mean and the standard deviation of the estimated test effects were $\eta=2.2$ and $\sigma=0.88$, respectively; and the standard deviation of the estimated covariate effects was $\tau=1.0$.}}
\label{table:sig_SNPs2}
\end{table}

We compared our methods GCW, GCW2 to Bonferroni and BW methods. To apply the proposed methods we followed the data analysis steps described above and to apply the BW method we used $bayes\_weights$ function from the $R$ package $pweight$. \cite{dobriban2015optimal} computed the standard deviation by $\sigma_i = \sqrt{\frac{N_i}{N_{0i}}}; \ i=1,\ldots,m$, where $N_i$ and $N_{0i}$ refer to the posterior and the prior sample sizes; and argued that their method performs well when the overdispersion $\phi=1$. In addition, from the data analysis, we found that their method performs well when $\sigma_i=1$. For example, at $5\%$ significance level, the number of significant SNPs is $210$ and $19$ when the Lipid and SCZ, respectively, are used as the covariates and $\sigma_i$ is calculated by $\sqrt{\frac{N_i}{N_{0i}}}$, which is lower than the number shown in the tables. Therefore, we always kept both $\phi=1$ and $\sigma_i=1$ to obtain the optimal performance of the BW method.  

The data analysis results are summarized in Table \ref{table:sig_SNPs} and \ref{table:sig_SNPs2}. From the Tables, we see GCW performed better than the BW method and both GCW and BW performed much better than the Bonferroni, where as GCW2 could not perform better than the Bonferroni method. For the GCW and GCW2 methods, we assumed that $f(y_i)$, the ratio of the prior and the posterior PDFs of the covariates, and $f(x_i)$ are $uniform(0,1)$. We believe that including correct information about the covariates can surpass the GCW method's performance over BW method by a large margin and improve the power of the GCW2 method as well. 

In order to include the correct information about the covariates, we may need another set of data or independent information. Although theoretically, we can incorporate a third data set, in practice, we need to explore the idea how we collect and incorporate a third data set to obtain the p-value weights. From now we can conclude that GCW is as good as BW or slightly better than BW, however, it will perform much better if we include correct prior and posterior information about the covariates. It can be noted that this example also signifies the importance of properly selecting covariates since changing the covariates also change the outcomes. It is also evident from the data analysis results that GCW performs better when the standard deviations of the covariates given the effect sizes and the covariates are close to $1$.


\chapter{Data-driven covariate weighting (DCW)}\label{ch:data_driven_method} 
\section{Backgrounds} 
\hspace{.22in} We introduced a method in Chapter \ref{ch:method_crw} for calculating approximate optimal weights conditioned on the ranking of tests by external covariates. This approach uses the probabilistic relationship between the ranks of tests and the test effect sizes to calculate more informative weights. Since this procedure obtains a weight for each test, the weights do not dilute by null effects as is common with group-based methods. We also demonstrated that this probabilistic relationship can be calculated theoretically for normally distributed covariates and argued that this relationship can be estimated empirically in other cases. In this chapter, we will show a data-driven method of estimating ranks probabilities and the corresponding weights.

Controlling the error rate and obtaining the optimal power is a major challenge for the multiple hypothesis test. Therefore, the p-value weighting has been introduced. However, the major challenge and question are how to estimate the weights. A popular method is incorporating external information frequently referred to covariate, which is independent of the test statistics under the null model but informative to power under the alternative model. 

The existing methods of obtaining the optimal weight and hence the optimal power requires knowledge of the effect sizes, which is not generally known. Almost in all cases, data-driven methods have been adapted by most authors to address this limitation, which leads to group-based analysis. Thus, these methods suffer because of the dilution of the effect if there is a tiny fraction of the hypotheses that are the true alternative tests (see Figure \ref{fig:group_effect} and \ref{fig:group_effect_ihw}). To overcome this limitation we proposed a modified power identity in terms of a probabilistic relationship between the test rank and the effect size, $f(\varepsilon \mid r_i )$, in which ranking of the test statistics is done by an independent external covariate. Consequently, we showed an approximate weight for the continuous effect sizes and an exact weight for the binary effect sizes. The weight for the continuous effects is
\begin{equation}\label{eq:contWeight_data_driven}
w_i \approx \Big(\frac{m}{\alpha}\Big) \bar \Phi \Bigg (\frac{E(\varepsilon)}{2} + \frac{1}{E(\varepsilon)} log\bigg(\frac{\delta}{\alpha P\big(r_i \mid E(\varepsilon)\big)}\bigg)\Bigg).
\end{equation}
and for the binary effect is
\begin{equation}\label{eq:BinWeight_data_driven}
w_i = \Big(\frac{m}{\alpha}\Big) \bar \Phi \Bigg (\frac{\varepsilon}{2} + \frac{1}{\varepsilon} log\bigg(\frac{\delta m}{\alpha m_1 P(r_i \mid \varepsilon)}\bigg)\Bigg)I(\varepsilon = \varepsilon).
\end{equation}
The continuous and the binary effects mean that the effect sizes of the alternative tests are continuous and the same (a fixed value), respectively. We also demonstrated an exact mathematical formula of the probabilistic relationship used in the weight equations. 
 
This method works well and outperforms other methods in many scenarios and is at least comparable to others in all scenarios. This method also does not require estimating the effect sizes or group analysis. However, it depends on the correct estimation of the mean of the true alternative effect sizes. A misspecification of the true mean of the covariate effects and the test effects can lead to poor weight and consequently fails to improve the power. A data driven method with a combination of group analysis can suppress the limitation of the method. Thus, we propose a data driven method that does not estimate the effect sizes from the group structure; instead, it uses CRW's probabilistic relationship to compute the weight by deriving the probabilistic relationship from the data. It is also computationally faster because of the grouping. In the proposed method, the input variables are p-values and the corresponding covariates are frequently termed as filter statistics, which are independent of p-values under the null hypothesis but informative under the alternatives. We proposed two versions of the weight estimation: 1) continuous case when the alternative effect sizes are continuous and 2) binary case when all effects are assumed to be either zero under the null or a fixed common value under the alternatives. 

\section{Methods}
\hspace{.22in} In this section, we will show how to compute ranks probability from the data for both the continuous effects and the binary effects. First, we will describe the procedures, then we will provide algorithms to conduct the methods. 
 
\subsection{Ranks probability given continuous effects $P\big(r_i \mid E(\varepsilon_i)\big)$} 
\hspace{.22in} Consider a multiple hypothesis testing setup. Suppose we want to test $H_{i0}:\varepsilon_i=0$ against $H_{i1}:\varepsilon_i>0; \ i=1,\ldots,m$. In addition, there is vector of test statistics $Z=\{z_1, \dots,z_m\}$ and the corresponding vector of covariates $X = \{x_1,\ldots,x_m\}$, where each hypothesis test $i$ has an associated covariate statistic $x_i$. It is believed that the test statistics corresponding to the larger covariate statistic values are more likely to be the true alternative tests. Furthermore, these covariates are independent of the p-values $(P)$ under the null hypothesis $(H_i  = 0)$ but informative for the power of the tests. 

To compute the ranks probabilities of the test statistics ranking by the covariates given the mean effect size $P\big(r_i \mid E(\varepsilon_i)\big)$, we consider the ordered p-values of the test statistics. This ordering is reasonable because we assumed that there is a relationship between the covariates and test statistics under the alternative hypothesis $(H_i  = 1)$. This relationship can be either linear or non-linear. Hence, we incorporate the ranking information into the test statistics, which will ultimately lead to the construction of the weight computed from the ranks probabilities of the tests. 

Denoted by $P_{x_{(i)}} ;i=1,\ldots,m$ the ordered p-values sorted by the ordered covariates $x_{(i)}$. P-values usually are a mixture of at least two distributions with support $0$ to $1$. The p-values of the true alternative hypothesis are likely to be close to zero. A higher peak close to zero is an indication of the substantial power of the hypothesis tests. We split the ordered p-values $P_{x_{(i)}}$ into G groups. Each group contains $m_g$ number of p-values so that $\sum_{g=1}^{G}m_g=m$. For the sake of simplicity, the number of p-values per group is assumed to be equal. Both the group $(g)$ and the group size ($m_g$) are variable and require optimization. We provide an algorithm to optimize the number of groups and the group size. We adapt the sizes that maximize the number of rejections. Several studies have been conducted in regard to this optimization problem \citep{ionita2007genomewide,roeder2009genome,kim2015prioritizing,ignatiadis2016data}. 

The groups represent the ranks $r_i; i=1,\ldots,G$ and the group with enrichment of the lower p-values represent the higher ranks. As the group or the rank index increases, the p-values distributions would be more uniform, because the lower p-values tend to come from the alternative models, whereas the higher p-values come from the null model. These groupings were conducted by the ranking of the covariates. The group with lower indexes contains mostly lower p-values are considered higher ranked and are informative with respect to the covariates. On the other hand, the groups with lower ranked contain mostly higher p-values shows uninformative to the covariates. In the next step, we estimate the ranks probabilities of the test statistics having assumed that we have $g$ groups or $r_g$ ranks. 

Denote by $P_{x_{(i)}}^g; g=1,\ldots,G$ the vector of ordered p-values in the $g^{th}$ group. For each group of  $G$, split the p-values $P_{x_{(i)}}^g$ into $K$ bins on the interval $\left[0,1 \right]$. Denote by $m_k^g$ the number of p-values in the $k^{th}$ bin of the $g^{th}$ group so that $\sum_{k=1}^{K}m_k^g =m_g$. Then define by $f_k^g$  the relative frequency in the $k^{th}$ bin of the $g^{th}$ group. Then the estimated relative frequencies or the densities of the first bin $f_1^g$ across the groups $g=1,\ldots,G$ will be the ranks probabilities of the groups. These ranks probabilities correspond to the quantile point of the first bin, which is, in other words, the estimated mean covariate effect size. The quantile points of the first bin across all $g$ groups are the same quantile points; however, the densities are not the same unless the p-values are strictly uniformly distributed across the groups. 

To compute the weights, we need the ranks probabilities corresponding to only one effect size because the weight equation is based on the information captured from the mean of the alternative effect sizes, which is most likely corresponding to the lower p-values. Therefore, the density of the first bin is most likely representing the ranks probabilities of the true alternatives' mean.

Next, we smooth the probabilities by using a non-linear regression approach. We recommend using a smoothing spline since it is popular yet simpler and can provide a good approximation of the complicated functions. Consider the data points $(p_1,r_1),\ldots,(p_G,r_G)$, where $p_g$ and $r_g; \ g=1,\ldots,G$ refer to the ranks probabilities and the corresponding ranks. Then the smoothing spline can be defined as
\begin{equation}
S\big(\mu(r_g ), \gamma\big)= \frac{1}{G} \sum_{g=1}^{G}\big(p_g-\mu(r_g )\big)^2+\gamma \int\big(\mu''(r_g )\big)^2dr_g.
\end{equation}
The first term expresses the mean square error (MSE) and the second term is the measure of the average curvature of $\mu$ at $r_g$. If two functions are with the same MSE, the function with less curvature is preferable. More broadly, increasing one unit MSE is equivalent to reducing the average curvature by $\gamma$. The solutions of the function are piece-wise cubic polynomials and are continuous and have continuous first and second derivatives. The function is also continuous beyond the smallest and the largest point, but this is always linear \citep{shalizi2013advanced}. In other words, we assume that the boundaries (knots) between the pieces are known, and the spline function is smooth at the knots, i.e., the spline function is differentiable at the knots. One should use at least four groups $(G \ge 4)$ in order to apply the cubic spline to smooth the probability curve. 

Smoothing can improve the power, especially if the alternative effect sizes are low; because the rank corresponding to the lower effect size tends to produce lower probability, consequently generates weights that can be close to zero. Thus, assigning a very low common weight to a group may not produce any rejection of the hypothesis test; however, smoothing the probabilities balance this situation by providing relatively more weights to the lower probabilities. On the other hand, the smoothing is not effective if the alternative effect sizes are very high from the background because the probabilities that correspond to the high effect sizes is discernibly more separated than the probabilities of the zero effect sizes, thus smoothing reduces these differences. 

Regardless of how the ranks probabilities are smoothed, one should normalize them to ensure that the probabilities are greater than zero and sum to $1$ and also normalize the weights so that the weights sum to the number of tests $m$. The difficulties of this procedure include finding the optimal group size $G$ and the optimal test size $m_g$ per group. Later we provide an algorithm to obtain the optimal number of groups and the group size.

\subsection{Ranks probability given binary effects $P(r_i \mid \varepsilon)$}
\hspace{.22in} The computing procedure for the ranks probabilities for the binary effects follows different steps than the steps that are described for the continuous effects. Instead of splitting the $m_g$  hypothesis into the bins, we estimate the proportion of the true alternative hypothesis in each group. Define by $f_1^g$ the proportion of the true alternatives in the $g^{th}$ group. Then the estimated proportion $f_1^g$ of the true alternatives is the rank probability of the $g^{th}$ group. This proportion can be efficiently obtained by the cubic spline regression \citep{storey2003statistical}. Spline regression is a better choice because this approach ensures an unbiased and less variable cut-off point. 

In the event of a cut-off point $\lambda \ \epsilon \ \left[0, 1 \right]$ that splits the p-values into the groups: true alternatives' p-values and true nulls' p-values, where $0$ to $\lambda$ stands for the range of the true alternative p-values and $\lambda$ to $1$ stands for the range of the true null p-values. Higher $\lambda$  lower the bias in the estimate of $f_1^g$ because the probability of the alternatives’ p-values exceeds $\lambda$ becomes lower; however, increasing  $\lambda$ increases the variability of estimating $f_1^g$ because of the fewer p-values to estimate the true proportion of the nulls. This problem can be solved by spline regression because spline produces estimate $f_1^g$ by considering all data points. All the remaining procedures such as ordering test statistics by covariates, then splitting p-values into several groups, etc., are the same as in the procedure described earlier for the continuous effects. 

For the reader’s convenience, the above methods are summarized in the following three algorithms: Algorithm~\ref{alg:data_driven_opt_groups} -- method to compute optimal number of groups, Algorithm~\ref{alg:data_driven_ranks_prob} -- method to compute ranks probability, and Algorithm~\ref{alg:data_driven_weights} -- method to compute weights.\medskip\\

\begin{algorithm}[H]
	\caption{Method to compute optimal number of groups}
	\label{alg:data_driven_opt_groups}
	\SetAlgoLined
	\textbf{Input:}\\
	A nominal significance level $\alpha \ \epsilon \ (0,1)$;\\
	A vector of p-values $P=\{p_1,\ldots,p_m \}$;\\
	A vector of covariates $X=\{x_1,\ldots,x_m \}$, where $i^{th}$ element is corresponded to the $i^{th}$ p-value;\\
	Maximum number of groups $G_{max}$ to be used.\\
	
	\textbf{Steps:}\\
	For the group $g=2,\ldots,G_{max}$, compute the ranks probabilities $f_r; r=1,\ldots,g$ for the $G$ groups.\\
	Compute the normalized weights $w_i; \ i=1,\ldots,m$.\\
	Denote by $\gamma \ \epsilon \ \left(1, g \right]$ the smoothing spline degrees of freedom.\\ 
	Define by $\gamma_{opt}$ the optimal degrees of freedom that maximizes the number of rejections.\\
	Define by $g_{opt}$ the optimal group that maximizes the number of rejections.\\ 
	Use $g_{opt}$ and $\gamma_{opt}$ to compute the ranks probabilities and the corresponding weights.\\
	
	\textbf{Return:} optimal number of groups
\end{algorithm}

\newpage

\begin{algorithm}[H]
\caption{Method to compute DCW ranks probability}
\label{alg:data_driven_ranks_prob}
\SetAlgoLined
\textbf{Input:}\\
A nominal significance level $\alpha \ \epsilon \ (0,1)$; a vector of p-values $P = \{p_1,\ldots,p_m\}$; and a vector of covariates $X=\{x_1,\ldots,x_m\}$, which is independent of $P$ under $H_0$.\\

	Let $x_{(i)}; \ i=1,\ldots,m$ be the ordered covariate statistics.\\
	Define by $P_{x_{(i)}}$ the ordered p-values sorted by the ordered covariate statistics $x_{(i)}$.\\
	Group the ordered p-values $P_{x_{(i)}}$  into $G$ groups.\\ 
	Denote by $m_g$ the number of p-values in the $g^{th}$ group so that $\sum_{g=1}^{G}m_g =m.$\\
	Let $P_{x_{(i)}}^g$ be the ordered p-values in the $g^{th}$ group.\\
	
		\eIf{effect type is “continuous”}
		{
			For each group of $G$, split the p-values $P_{x_{(i)}}^g \ \epsilon \ \left[0,1\right]$ into bins.\\ 
			Denote by $m_k^g; \ k=1,\ldots,K$ the number of p-values in the $k^{th}$ bin of the $g^{th}$ group so that $\sum_{k=1}^{K}m_k^g =m_g.$\\
			Define by $f_k^g$ the relative frequency in the $k^{th}$ bin of the $g^{th}$ group.\\ 
			Take estimated relative frequency $f_1^g$  of the first bin. Then, the estimated densities $f_1^g; \ g=1,\ldots,G$ across groups will be the ranks probabilities of the groups. 		
		}{
			Define by $f_1^g; \ g=1,\ldots,G$ the proportion of the true alternatives in the $g^{th}$ group. Then the estimated proportion $f_1^g$ of the true alternatives across groups will be the ranks probabilities of the groups.
		}
			
	Apply smoothing such as cubic spline to obtain the optimized probabilities.\\
	Normalize the probabilities so that the probabilities sum to one, i.e., $\sum_{g=1}^{G} f_1^g=1$\\
\textbf{Return:} normalized ranks probabilities
\end{algorithm}

\newpage

\begin{algorithm}[H]
\caption{Method to compute DCW weights}
\label{alg:data_driven_weights}
\SetAlgoLined
\textbf{Input:}\\
A nominal significance level $\alpha \ \epsilon \ (0,1)$;\\
The mean test effect $E(\varepsilon_i); \ i=1,\ldots,m_1$ or the median test effect $\varepsilon$ of the true alternatives;\\
A vector of ranks probabilities $f_1^g; \ g=1,\ldots,G$ for the $G$ groups.\\

\textbf{Steps:}\\
Generate a sequence of $\delta = \delta_1,\ldots,\delta_n$, where $\delta \ \epsilon \ (0,1)$ and $\lvert\delta_j-\delta_k\rvert \le .001; \ j,k=1,\ldots,n;j \ne k$.\\

\eIf{effect type is “continuous”}
	{
		For each $\delta_j$, compute the weights  
		$w_i \approx (\frac{G}{\alpha}) \bar \Phi \Big (\frac{E(\varepsilon_i)}{2} + \frac{1}{E(\varepsilon_i)} log(\frac{\delta_j}{\alpha f_1^g})\Big).$
	}{
		Estimate the number of true alternative tests by $\frac{m_1}{m}G$.\\
		For each $\delta_j$, compute the weights $w_i = (\frac{G}{\alpha}) \bar \Phi \Big (\frac{\varepsilon}{2} + \frac{1}{\varepsilon} log(\frac{\delta_j m^2}{\alpha m_1 G f_1^g})\Big)I(\varepsilon = \varepsilon)$.
	}
	
Compute the sum of weights $sw_j(\delta_j)=\sum_{i=1}^{m}w_i$  for each $\delta_j$.\\	
Denote by $\delta_{opt}$ the optimal value of $\delta$, which satisfy $\underset{\delta_{opt}}{min} \lvert sw_j (\delta_j)- G\rvert$.\\
Compute the weights $w_g$ by using $\delta_{opt}$.\\
Replicate the weights $w_g$ so that the p-values in the $g^{th}$ group receive the same weight and the size of the weights' vector equal to $m$.\\
Normalize the weights so that the weights sum to $m$.

\textbf{Return:} weights
\end{algorithm}

\newpage

\section{Simulation study} 
\hspace{.22in} In this section, we provided simulation results to demonstrate a preliminary understanding about the p-value distribution and showed how to compute ranks probabilities and the corresponding weights. Later we also demonstrated the Power of the DCW method via simulations.
\subsection{Ranks probability}
P-values usually are a mixture of at least two distributions. If the p-values are from the null model then these follow $uniform(0,1)$; otherwise, some distributions with support from $(0,1)$.
\begin{figure}[ht]
	\begin{center}
		\includegraphics[scale=.6]{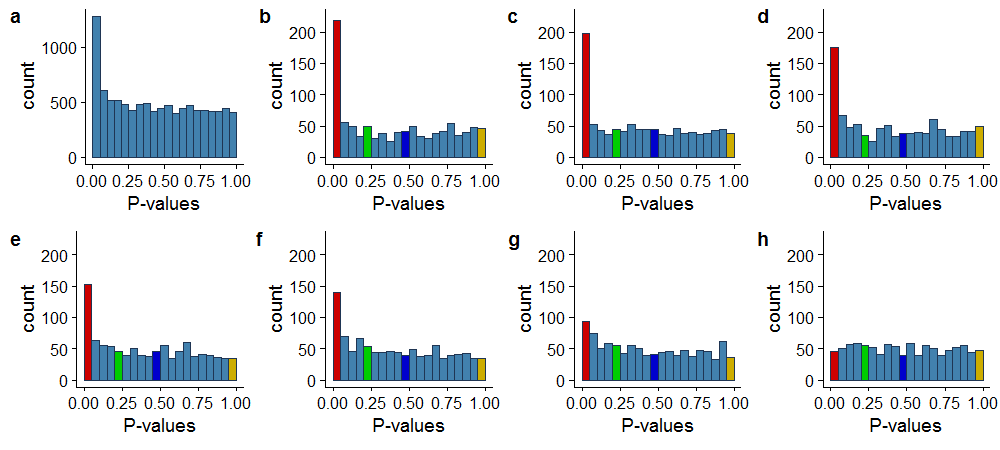}
	\end{center}
	\caption[Distribution of the p-values and the stratified p-values by groups]{\footnotesize{Histograms stratified by the covariates' ranks. (a) The histogram of all p-values and an enrichment of small p-values to the left. (b–h) Histograms of the p-values by groups $g=\{1,2,3,4,5,7,10\}$ after splitting the p-values into 10 groups by the ranks of the covariates. We assumed there were $80\%$ true null and $20\%$ true alternative hypotheses and each group is composed of $1000$ p-values. The p-values in the red-colored bin are most likely from the true alternative hypotheses.}}
	\label{fig:pval_hist}
\end{figure}
For a quick overview, Figure \ref{fig:pval_hist}(a) can be referenced. We generated $m=10,000$ p-values in which approximately $80\%$ of the p-values are from the true null models. We also generated $m=10,000$ covariates in which $i^{th}$ covariate is corresponded to the $i^{th}$ p-value. These covariates are related to the p-values under the alternative models but are independent under the null models. Most of the p-values of the true alternative hypothesis are on the left side, and those of the true null hypotheses are on the right side. A higher peak on the left side indicates that a substantial power can be obtained from the tests. 

In order to compute the ranks probabilities for the tests and the corresponding weights empirically from the data, we used the covariates and the corresponding p-values of the test statistics. First, we ordered the p-values by the covariates, i.e., the covariates’ ranks were the ranks of the p-values denoted by $P_{x_{(i)}}$. Ordering by covariates allowed us to incorporate the ranking information into the test statistics, which ultimately leads to the construction of the weight computed from the ranks probabilities of the tests. Once the p-values are ordered, we divided the p-values of the test statistics into $G=10$ groups. The indexes of the groups represent the rank of the tests; in other words, we will call these indexes as $10$ test ranks. Figure \ref{fig:pval_hist}(b-h) show the ordered p-value distribution of the groups $g=\{1,2,3,4,5,7,10\}$. Each group contains $m_g=1000; \ g=1,\ldots,10$ p-values so that $\sum_{g=1}^{10}m_g=m.$ For simplicity, we assumed group sizes to be equal. 

As the group index or the rank increases, the p-value distributions become more uniform. This is always true because the lower p-values intend to come from true alternative models; however, the higher p-values are from the true null models, and this pattern shows that the p-values are informative under the alternative model. Note that, the groups with lower indexes are considered higher ranked, which shows informative with respect to the covariates. On the other hand, groups with lower ranks are uniformly distributed, which shows uninformative with the covariates.
\begin{figure}
	\begin{center}
		\includegraphics[scale=.6]{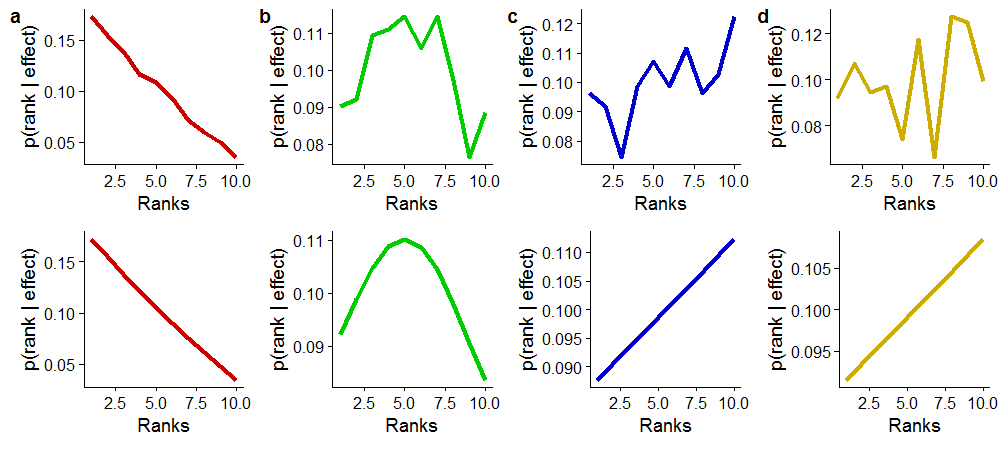}
	\end{center}
	\caption[Ranks probability of the Data-driven Covariate Weighting (DCW) method]{\footnotesize{Simulated ranks probabilities of the tests for the four particular effect sizes or quantiles. Different colors represent the different bins of the histograms of Figure \ref{fig:pval_hist}. The first row and the second row represent the ranks probabilities before and after applying the smoothing. Plot (a-d) represent the ranks probability of the estimated effect sizes $E(\varepsilon)=\{1.96,1.64,1.15,0.67\}$, respectively, which are computed as $\Phi^{-1} (1-md/2)$, where $md =$ mid-value of the corresponding bin. The leftmost plot is the ranks probabilities of a test statistic across all groups when the mean effect size of the true alternative tests is $1.96$.}}
	\label{fig:ranksProb_datadriven}
\end{figure}

Next, we estimated the ranks probabilities for the p-values having assumed that we have $G=10$ groups or $r_g;\ g=1,\ldots,10$ test ranks. In order to do this, we split each group of p-values into $K=20$ bins; in other words, we specified $21$ points on the interval $\left[0,1\right]$ of the p-values for a particular group. Then, we obtained the relative frequency of the first bin $\hat f_1^g$ (Figure \ref{fig:ranksProb_datadriven}a). The estimated relative frequencies are the ultimate ranks probabilities of the test statistics and the quantile point is the estimated mean test effect size which can be computed as $E(\varepsilon)=\Phi^{-1}(1-.05/2)$. Note that the first bin’s interval is $[0, .05)$, therefore, we considered the mid-value of the bin. For clarification, we also showed the ranks probabilities curves (Figure \ref{fig:ranksProb_datadriven} (b-d)) for the other effect sizes. In the data application, we only need the ranks probabilities that correspond to the first plot. 

The quantile points across all $G$ groups are the same quantile points; however, the relative frequencies are not the same. The relative frequencies across the groups are the ranks probabilities of the tests given the effect sizes. In the above-simulated situation, we have $10$ sets of ranks probabilities of the tests given $20$ different effect sizes. However, we only need the ranks probabilities corresponding to the first bin. For the simulation, we assumed $G=10$ groups and each group has $m_g=1000$ tests, but these numbers are not fixed and can vary depending on the conditions or goal of the data analysis. In order to obtain an optimal number, we provided Algorithm~\ref{alg:data_driven_opt_groups}. 

Figure \ref{fig:ranksProb_datadriven} will provide a better picture about ranks probabilities. The figure also shows the ranks probabilities of the tests of four different effect sizes across all groups. As we see, when the mean effect size decreases toward the zero, the ranks probabilities tend to decrease initially with the lower ranks then increase. This event occurred because as we move toward the higher p-values from left to right, the effect sizes decrease, and the probability of being ranked lower of a test corresponding to the lower effect size becomes higher.

\subsection{Ranks probability and weights}
\hspace{.22in} In this section, we showed simulated ranks probabilities for the tests and the corresponding weights with respect to the proportion of the true null hypothesis and the number of groups for a fixed number of tests $m=10,000$. Note that as an increase in the number of groups decreases the group size or the number of tests per group. 

From Figure \ref{fig:ranks_weight_emp}a, we see the proportion of the true null hypothesis has a great influence on the ranks probabilities. The rank probability tends to be sharply slanted with the low proportion of the true null tests when the number of groups is fixed (e.g. $G=10$), and flatten out with the higher proportion of the true null hypothesis. The curves of the groups’ weight have the similar shape to the ranks probabilities curves (Figure \ref{fig:ranks_weight_emp}b), as expected, and the weight varies approximately from $0$ to $2.5$. 

In order to generate these plots, we adapted the same procedures (Algorithm~\ref{alg:data_driven_ranks_prob} and \ref{alg:data_driven_weights}) described earlier except here we considered the different proportion of the true null tests and group numbers. Figure \ref{fig:ranks_weight_emp}c represents the individual weight for each test where all tests in a particular group received the same weight. From the simulation of the power (provided later), we observed that optimal group size plays a crucial role in maximizing the number of discoveries. Figure \ref{fig:ranks_weight_emp}d provides an overview how the weight changes with the number of the groups. For a fixed proportion of the true null hypothesis (e.g. $\pi_0 = 80\%$), the weights tend to become smoother with the higher number of groups.  
\begin{figure}[ht]
	\begin{center}
		\includegraphics[scale=.5]{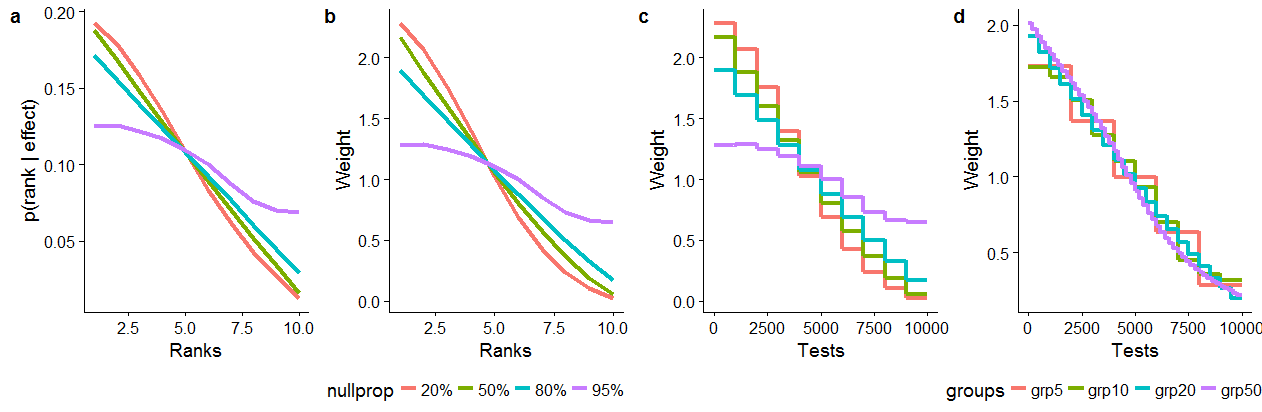}
	\end{center}
	\caption[Ranks probability and the corresponding weights of DCW]{\footnotesize{Figure (a-c) shows information with respect to the proportion of the true null hypothesis when the number of group is $G=10$: a) The ranks (groups) probability, b) the corresponding normalized group-weights and c) individual weight, and d) weight with respect to the number of groups when the proportion of the true null is $\pi_0 = 80\%$.}}
	\label{fig:ranks_weight_emp}
\end{figure}

\subsection{Power}
\hspace{.22in} We also compared the performance of various methods such Bonferroni (BON) or Benjamini and Hochberg (BH) \citep{benjamini1995controlling}, Roeder and Wasserman (RDW) \citep{roeder2009genome}, and Independent Weighted Hypothesis (IHW) \citep{ignatiadis2016data} with our proposed method (DCW) in terms of simulated Family Wise Error Rate (FWER), Power, and False Discovery Rate (FDR). We conducted simulations for the different combinations of parameters; however, only a portion (Figures \ref{fig:power_emp}) is presented here. The entire simulation procedure is divided into three groups based on the proportion of the true null hypothesis. Three groups are composed of $\pi_0 = \{50\%, 90\%,99\%\}$ true nulls. 
\begin{figure}[ht]
	\begin{center}
		\includegraphics[scale=.58]{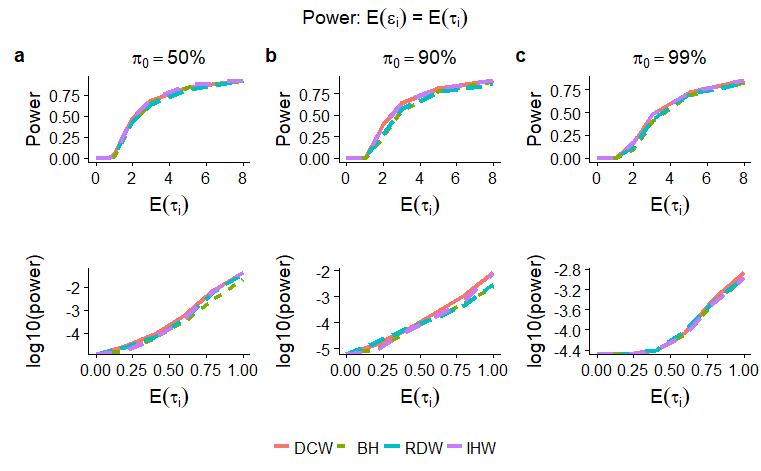}
	\end{center}
	\caption[Power of the tests by the DCW method]{\footnotesize{The Power of four methods when the mean test-effect $E(\varepsilon_i)$ is equal to mean covariate-effect $E(\tau_i)$. Each plot consists of four curves of DCW, BH, RDW, and IHW methods. The first row shows the power for the low to high effect sizes, and the second row shows $log(power)$ for the low effect sizes. Three columns represent three groups of $50\%,90\%,$ and $99\%$ true null hypothesis.}}
	\label{fig:power_emp}
\end{figure}

For each group of simulations, we considered the combination of the correlations and the effect sizes as $\rho \times E(\varepsilon_i) = \{0,.3,.5,.7,.9\} \times \{E(\tau_i),N\big(E(\tau_i),CV.E(\tau_i)\big)\}$, where $E(\varepsilon_i)$, $E(\tau_i)$, and $CV$ refers to the mean test-effect, mean covariate-effect, and coefficient of variation. The correlation was between test statistics. For the non-correlated case, we generated tests from the normal distribution and for the correlated case tests are generated from the multivariate normal distribution, respectively. We formulate a correlation matrix of $10,000$ rows and $10,000$ columns to generate multivariate normal tests. The matrix was diagonally split into $100$ blocks, and each block consisted of $100$ rows and $100$ columns. All the remaining cells are filled with zeros. This process is followed to take into account the group effect, which is very common in high-throughput data.\\ 

Apparently, all of the methods, including the non-weighted BH method, perform equally. We hypothesize this results is due to the small number of test statistics and normally distributed covariates. The data example shows that for the non-normal covariates DCW and other methods significantly perform better than the BH method. Furthermore, we observed the significance of the tests correlation, the proportion of the true nulls, and the controls over the error rate and we also tested via simulations that DCW controls FWER and FDR for the different effect sizes.

\section{Data application}
\hspace{.22in} For comparative purpose, we applied the same Bottomly and Proteomics data sets that are discussed for the CRW method in Chapter \ref{ch:CRW_data_application}. We followed the similar procedure to obtain the p-values and the covariates. However, we applied the data-driven method (DCW) to compute the ranks probabilities and the weights.
 
\subsection{Application of the DCW method}\label{sec:dcw_data_steps}
\hspace{.22in} In order to apply the proposed DCW method, we followed the several steps. For the reader’s convenience, the steps are summarized below:

\textbf{Input:} A nominal significance of level $\alpha \ \epsilon \ (0,1),$ a vector of p-values $P = \{p_1, \ldots,p_m\}$ and a vector of  covariates $X=\{x_1, \ldots,x_m\}$. 

\textbf{Step 1:} Denote by $m_0$ and $m_1$ the number of true null and true alternative tests. Estimate $m_0$ and $m_1$ by the method of \cite{ritchie2015limma}.

\textbf{Step 2:} Denote the vector of test statistics by $T$. Obtain the test statistics by $T = \Phi^{-1}(1-P)$ for the one-tailed and $T=\Phi^{-1}(1-P/2)$ for the two-tailed p-values.

\textbf{Step 3:} Denote by $\varepsilon$ the effect of the true alternative tests. Estimate $\varepsilon$ from the top $m_1$ tests $T$ from Step 2. 

\textbf{Step 4:} Denote by $g_{opt}$ the optimal number of groups. Find the $g_{opt}$ by Algorithm~\ref{alg:data_driven_opt_groups}.

\textbf{Step 5:} Compute ranks probabilities $P\big(r_i \mid E(\varepsilon)\big)$ and the corresponding weights $w_i$ by Algorithm~\ref{alg:data_driven_ranks_prob} and \ref{alg:data_driven_weights}, respectively.

\textbf{Step 6:} Apply the weighted Bonferroni $\sum_{i=1}^{m}I(p_i \le \frac{\alpha w_i}{m})$ or weighted FDR \citep{wasserman2006weighted} procedure $\sum_{i=1}^{m}I(p_{ia} \le \alpha)$ to obtain the number of significant tests, where $I(.)$ and $p_{ia}$ are the indicator function and the weighted adjusted p-values, i.e $p_{ia} = adjust(\frac{p_i}{w_i})$, respectively.

The mean of the true alternative test statistics is taken as the mean test effect size. The mean is computed from top $m_1$ test statistics, which represents the estimated number of true alternative tests. We applied the procedure described above (Algorithm~\ref{alg:data_driven_opt_groups}, \ref{alg:data_driven_ranks_prob}, and \ref{alg:data_driven_weights}) to compute the probabilities of ranks of the tests, $P\big(r_i \mid E(\varepsilon)\big)$. These probabilities is then used to compute the weights. We compared the result to that of using the Benjamini–Hochberg (BH) \citep{benjamini1995controlling}, and Independent Hypothesis Weighting(IHW) \citep{ignatiadis2016data} methods.  
\begin{figure}[ht]
\begin{center}
	\includegraphics[scale=.6]{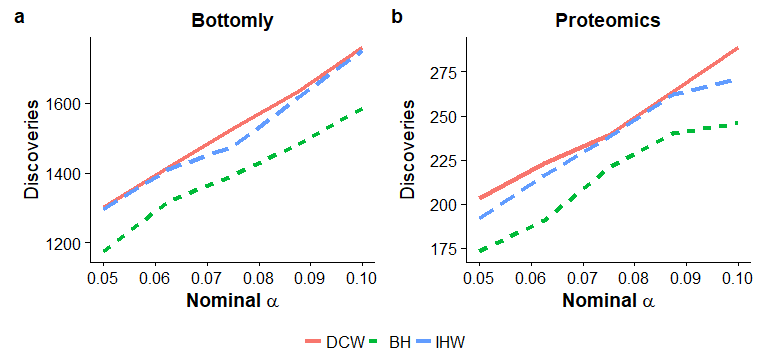}
\end{center}\label{fig:bottomly_emp}
\caption[Data application of the DCW method]{\footnotesize{This plot shows the number of rejected null hypotheses across different significance levels of $\alpha$ for the two data sets: Bottomly and Proteomics. Here, DCW = Proposed data-driven, BH = Benjamini and Hochberg, and IHW = Independent Hypothesis Weighting methods.}}
\end{figure}


\chapter{Effective application of optimal p-value weighting} \label{ch:effective_data_app}
\section{Overview} \label{sec:third_Introduction}
\hspace{.22in} High throughput data is very common in modern science. The main property of these data is high-dimensionality, that is, the number of features is larger than the number of observations. There are many ways to study this kind of data, and multiple hypothesis testing is one of them. In a multiple hypothesis test, generally a list of p-values ($p_i$) are calculated, one for each hypothesis ($H_i$) corresponding to one feature; the p-values are then compared with predefined fixed thresholds or random thresholds to obtain the number of significant features. The key goal is to control the FWER while maximizing the power of the tests. In some cases, it is convenient to control the FDR \citep{benjamini1995controlling} because it considers all the features that are declared significant by the hypothesis test and is, therefore, reasonable when there are a large number of hypotheses to tests. Scientists prefer multiple hypothesis test over other high throughput data analytical methods because it is convenient and provides a simple means of studying many features simultaneously.

Although multiple hypothesis test provides an opportunity to test many features simultaneously, it often requires high compensation for doing so \citep{stephens2016false}. To overcome from this shortcoming, \cite{benjamini1997false} shows a way of using the rank of the p-values frequently termed as adjusted p-values. However, this method, again, solely depends on the p-values and, therefore, provides only sub optimal power of the tests. A decades old and promising method \citep{ferkingstad2008unsupervised,ploner2006multidimensional} recently renovated in the scientific study is p-value weighting in which external information is used in terms of weight and the actual p-value is redefined by incorporating the weight, which is usually accomplished dividing the original p-values by the corresponding weights and creating new p-values called weighted p-values. 

Most of the articles referred to the external information as the covariates or the filter statistics ($y_i$). Generally, covariates provide different prior probabilities of the null hypotheses being true; therefore, a judiciously chosen covariate can significantly improve the power of the test while maintaining the error rate below the threshold. Such covariates are frequently available from various studies and data sets \citep{ignatiadis2016data}, which is generally considered irrelevant for the frequentist in a single hypothesis test. However, this covariate has become very relevant in the multiple hypothesis test because of the concern regarding the power and controlling the error rate of the hypothesis test. 

In this chapter, we discussed an idea of obtaining the covariates from an external source and incorporating them into the method that we proposed in the previous chapters (Chapter \ref{ch:method_crw}). In the chapter \ref{ch:method_crw}, we showed how to compute an optimal weight of the p-values without estimating the effect sizes. We showed that the weights can be computed by applying a probabilistic relationship of the ranking of the tests and the effect sizes. We also showed an exact mathematical formula for obtaining the probabilistic relationship. In the Chapter \ref{ch:data_driven_method}, we provided a data driven method, as well, of obtaining the probabilistic relationship, and thence, computing the weight, empirically. We showed that the CRW method can substantially increase the power while maintaining the FWER. This method can easily be extended to the context of the FDR. The requirements of the method are that the covariates are assumed to be independent under the null hypothesis but informative for the power \citep{bourgon2010independent}. In addition, the weights must be non-negative, and the mean of the weights must be equal to $1$. Furthermore, the weights greater than one and less than one up weight and down weight the p-values, respectively; and a well-defined weight can substantially increase the power, while the loss of the power by a miss-specified weight is minimal.

\section{Mouse QTL data}
\hspace{.22in} In this section, we will present the results that we obtained via applying the CRW method. We will also present a procedure of obtaining covariates from an external source and comparison of the results obtained by the CRW method and the existing $R/qtl$ analysis. 
\subsection{Results} 
\hspace{.25in} We identified $14$ significant markers in chromosomes $1, 2, 3, 4, 6,$ and $11$ by applying the CRW method. It could be suggested that $QTLs$ close to these markers are responsible for a mouse’s growth. The existing $R/qtl$ (an $R$ package developed by \citealt{broman2003r}), the analysis also discovered markers in the same chromosomes but identified only six markers. The CRW method not only detected the same markers in the chromosomal positions, which were already identified by the $R/qtl$ analysis but also discovered additional markers in the different locations of the same chromosomes, which were not identified by the $R/qtl$ analysis. These discoveries clearly indicate that there are $QTLs$ close to the new markers that could be responsible for a mouse’s body growth and size. The newly discovered markers are $D3MIT18, D6MIT132,$ and $D11MIT4$ in chromosomes $3, 6,$ and $11$, respectively, and the possible newly identified $QTLs$ correlated to these markers are $Wt10q2, Mtbcq1,$ and $Lgaq4$ in chromosome $3$; $Dbm1, Tabw2, Plbcq3, W6q5,$ and $Obwq3$ in chromosome $6$; and $Mskt6, Fina1, Skl5,$ and $Stzid3$ in chromosome $11$.
\subsection{Methods}
\hspace{.25in} Scientists are interested in the genetic architecture of the complex polygenic trait of the mouse. To further investigate the genetic architecture of the complex polygenetic trait, \cite{rocha2004large} conducted a study composed of long-term selection lines for the high and low growth of the mouse. The study was composed of approximately $1,000 F2$ cross between inbred mouse strains $M16i$ (the high weight gain strain) and $L6$ (the low weight gain strain). In the study, the authors identified $154$ $QTLs$ that is responsible for the mouse’s body growth. We considered the same data and applied the method we proposed and found additional QTLs which could contribute to mouse growth. In this study, we analyzed the data by the $R/qtl$ method as well as by the CRW method to compare the results obtained from both approaches. From the data, we used only {\itshape weight at weeks 10} as the targeted phenotype among the other phenotypes. 

In this section, we describe the applications of the methodologies: $1) R/qtl$ analysis and $2)$ the CRW method. In order to apply the methods, we first screen the data and computed the LOD scores and the corresponding p-values; then we applied $R/qtl$ and the CRW method. LOD score is a statistical measurement of how likely two loci on a chromosome are near to each other and therefore how likely they are inherited as a package. The $LOD$ scores of the markers are computed by applying $R/qtl$ analysis procedures. These $LOD$ scores and the corresponding p-values are then compared against the thresholds computed by the $R/qtl$-permutation test with Haley-Knott (HK) regression and by the CRW method.\\
\begin{figure}[ht]
\begin{center}
	\includegraphics[scale=.6]{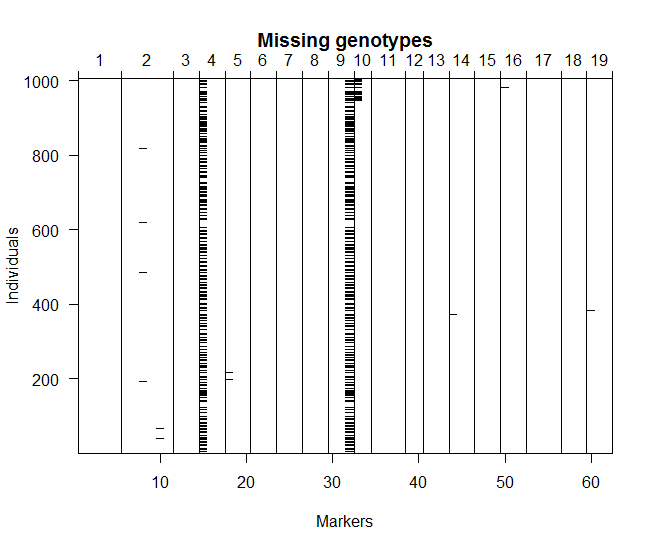}
\end{center}
\caption[Missing genotypes of the mouse markers]{\footnotesize{This plot shows the missing genotypes of the markers. It is obvious that there are many missing genotypes in the first and the third markers of chromosomes $4$ and $9$.}}
\label{fig:qtl_missing}
\end{figure}
\hspace{.2in}We first carefully observed any missing pattern in the data and found problematic missing information in chromosomes $4$ and $9$ (Figure \ref{fig:qtl_missing}). There are many missing genotypes, which can be a concern if there is a significant correlation between the phenotype ({\itshape weight at weeks 10}) and the missing genotypes. In fact, we found that there is a potential pattern of missing genotypes among weightier mice. Therefore, we dropped the first marker $(D4MIT1)$ and the third marker $(D9MIT19)$ from chromosomes $4$ and $9$, respectively.\\
\begin{figure}[ht]
\begin{center}
	\includegraphics[scale=.6]{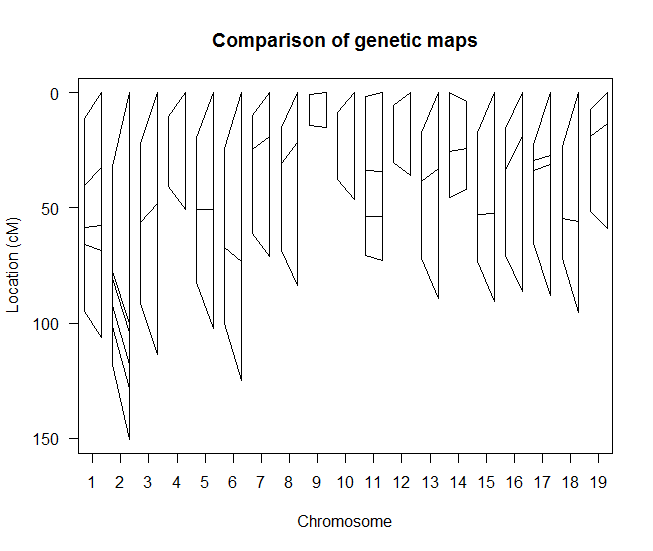}
\end{center}
\caption[Mouse genetic maps side-by-side]{\footnotesize{This plot shows the old and new genetic maps side-by-side. The old genetic map (left) is based on some earlier genetic studies, and the new (right) genetic map is estimated from the current data.}}
\label{fig:qtl_map}
\end{figure}
\hspace{.2in}We also investigated the genetic map (Figure \ref{fig:qtl_map}) to observe major discrepancies, if any, between the expected map distances and the observed map distances between the markers. The plot seemed reasonably fine, although estimated distances from the current data are a little higher. Furthermore, we also investigated the genotyping error for all markers and found little evidence of the genotyping error.

The actual analysis was conducted in $R/qtl$, and HK regression was used to test for the linkage of the trait and the transcripts of the markers. $R/qtl$ provides $LOD$ scores, which were converted into p-values by applying the asymptotic chi-squared distribution with two degrees of freedom. The p-values corresponding to the markers were then compared with the pre-specified threshold to identify the significant markers.   
 
\subsection{Significant markers by $R/qtl$ analysis}
\hspace{.22in} To identify the significant markers under $R/qtl$ analysis, we used $1000$ permutations with HK regression to compute the $5\%$ significance threshold, which is a $3.01$ $LOD$ score. This $LOD$ score then was used as a cutoff point, i.e., the markers with $LOD$ score greater than $3.01$ were labeled as significant markers, and the $R/qtl$ analysis identified six markers (Table \ref{table:sig_marker_qtl}). That is, there is strong evidence of existence of the body-weight $QTLs$ in the regions of each of the six markers. 
\begin{table}[ht]
\caption[Significant markers by the $R/qtl$ analysis]{\footnotesize{Significant markers and the corresponding statistics by $R/qtl$ analysis }}
\begin{center}
\begin{tabular}{lcrcc}
	\hline
	markers & chr & pos & lod & pval \\ 
	\hline
	D19MIT72 & 1 & 56.44 & 5.26 & 0.000 \\ 
	D2MIT224 & 2 & 62.89 & 9.55 & 0.000 \\ 
	D3MIT97 & 3 & 37.06 & 6.27 & 0.000 \\ 
	D4MIT27 & 4 & 42.13 & 4.42 & 0.002 \\ 
	D6MIT138 & 6 & 1.81 & 3.27 & 0.020 \\ 
	D11MIT2 & 11 & 9.79 & 6.73 & 0.000 \\ 
	\hline
\end{tabular}
\end{center}
\footnotesize{\parbox[t]{6in}{To identify the markers we computed significance threshold by permutation with HK regression. Here, chr = chromosome number, pos = marker position, lod = $LOD$ score, and pval = p-value.}}
\label{table:sig_marker_qtl}
\end{table}
\subsection{Significant markers by CRW method}
\hspace{.22in} We applied CRW method to identify the markers. Our proposed CRW method requires a covariate frequently termed as filter statistics in addition to the test statistics, which is used to compute the weight of the p-values. In order for p-value weighting methods to be effective, weights must be estimated from independent data \citep{bourgon2010independent}. Considering the same setup as above, suppose that in addition there is a vector  of covariates generated from independent data, with the $i^{th}$ element of the vector corresponding to the $i^{th}$ hypothesis test. This covariate will tend to be higher for more promising tests and lower for less promising ones. 

From the \cite{rocha2004large} data, we obtained test statistics or p-values using $R/qtl$ analysis. We also need covariates for each marker to compute the weight, which must be independent of the data. In their article, \cite{rocha2004large} reported $154$ $QTLs$, which were identified as significantly related to mouse growth. We used the MGI (Mouse Genome Informatics, \citealt{smith2014mouse}) database to identify the information of the $154$ $QTLs$ and construct the covariates. We basically identified the $LOD$ scores and centimorgan ($cM$) values, a measurement of the locations of the $QTLs$ or the markers in the chromosome, of the $QTLs$. 

We also searched the MGI database against {\sl “growth/size/body”} phenotype and found $337$ (this includes almost all of the $154$ $QTLs$) unique $QTLs$ that were reported as significantly related to the mouse growth in the various articles. For these $337$ $QTLs$, we collected $LOD$ scores and the corresponding centimorgan ($cM$) values. Then we matched the $QTLs’$ $cM$ values with the cM values of the markers of the original data set \citep{rocha2004large}.  We assigned the $LOD$ scores of the $QTLs$ obtained from the independent study as a covariate of the markers closest to $QTLs$ in $cM$. In some cases, multiple $LOD$ scores were assigned to some markers because there are multiple $QTLs$ close to the markers. On the other hand, there was also no $LOD$ score to be assigned to some markers because of the unavailability of the $LOD$ scores of the $QTLs$ close to those markers in the $MGI$ database. For the multiple $LOD$ scores, we used median $LOD$ score as a covariate, and for no-reported $LOD$ score, we used $3.0$. The $LOD$ score $3.0$ is a reasonable choice, because in the articles from the $MGI$ database, it is mentioned that $QTLs$ are reported as significant if the $LOD$ scores are greater than $3.0$. Since markers are reported to the linkage study, there could be some $QTLs$ close to those markers, although they have not been identified yet. Therefore, we assigned a $LOD$ score of $3.0$ for the non-identified $QTLs$.\\
\begin{figure}[ht]
\begin{center}
	\includegraphics[scale=.55]{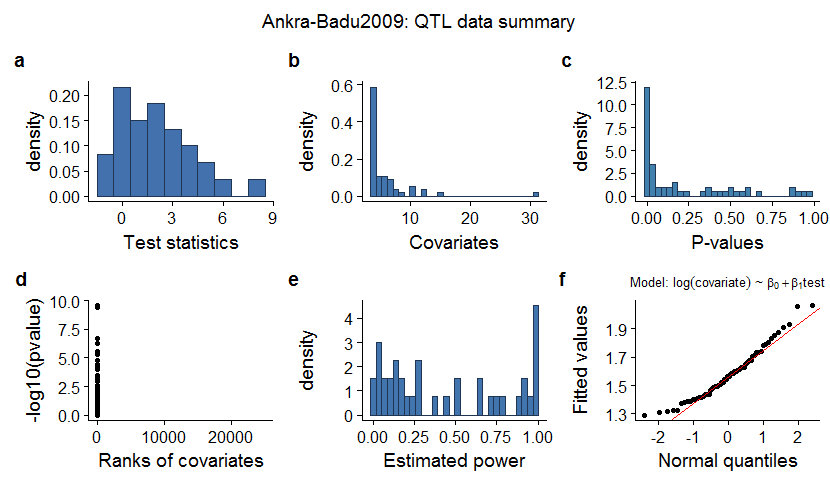}
\end{center}
\caption[Summary plots of the mouse QTL data: test statistics, covariates, and p-values]{\footnotesize{Pre-analysis plots: first row (left-right) shows the distribution of the test statistics, covariates, and the p-values, respectively; and the second row shows the rank of the covariates vs. the p-values (low index is better rank), estimated power of the tests, and the Q-Qplot of the fitted values.}}
\label{fig:summary_plots_filter}
\end{figure}
\hspace{.2in} Before applying our method to the test statistics and the covariates, we performed a pre-screening analysis of the data to obtain the relationship between the test statistics and the covariates because the CRW method requires estimated mean effect size from both types of statistics. From the pre-analysis plots (\Cref{fig:summary_plots_filter}), we see the distribution of the covariates is rightly skewed; therefore, we performed a natural log-based transformation to make the distribution reasonably symmetric. Although the transformed distribution is still non-symmetric, this may not be a severe problem because the CRW method only requires the center of the distribution.

We also observe from the plot that there is a weak relationship between the test statistics and the covariates; however, the ranked-covariates show potential relationships with the p-values, i.e., low p-values are enriched at the higher covariates. This relationship might be informative for the p-value-weighting because we are more interested in the ranking of the covariates. Enrichment of the low p-values at higher covariates also indicates that the covariates or the $LOD$ score covariate are correlated to power under the alternative model. Furthermore, the distribution of the p-values indicates one-tailed test criteria could be applicable because the p-value distribution is unimodal. We also observed the empirical cumulative distribution of the p-values. The empirical cumulative distribution shows that the curve is almost linear for the high p-values, which reveals the lesser importance of the higher p-values, and the size of the low p-values is very small. The estimated power plot shows that there are about 12 QTLs which have more than $50\%$ probability to be significant. 

To obtain the optimal weight for the optimal power, we require the ranks probabilities of the test statistics, $P\big(r_i=k \mid E(\varepsilon)\big)$, which we be obtained by applying the normal approximation formula (equation \ref{eq:prob}):
\begin{equation}\label{eq:prob}
P(r_i=k \mid \varepsilon) = 
\begin{cases}
E_T N(\mu_0,\sigma_0^2), & \text{if } \varepsilon_i=0\\
E_T N(\mu_1,\sigma_2^2), & \text{if } \varepsilon_i > 0\\
\end{cases},
\end{equation}
where
\begin{equation}
\begin{cases}
\mu_0 = (m_0-1)(1-F_0 ) + m_1 (1-F_1 ) + 1\\
\mu_1 = m_0(1-F_0 ) + (m_1-1)(1-F_1 ) + 1\\
\sigma_0 = (m_0-1)(1-F_0 )F_0 + m_1 (1-F_1 )F_1\\
\sigma_1 = m_0(1-F_0 )F_0 + (m_1-1)(1-F_1 )F_1
\end{cases}.
\end{equation}
and $m_0$ and $m_1$ are the number of the true null and the true alternative hypotheses, respectively, and $F_0$ and $F_1$ are the cumulative distributions of the test statistics under the null and alternative models, respectively. The parameters $m_0$ and $m_1$ were estimated by the method of \cite{storey2003statistical}. Since the actual test statistics follow a chi-square distribution with two degrees of freedom, we transformed the test statistics to asymptotically normal test statistics \citep{fisher1925statistical}. One can also work directly by considering $F_0$ and $F_1$ are the cumulative distribution functions of the chi-square distribution with two degrees of freedom under the null and the alternative hypothesis, respectively, if the distributions of the test statistics are chi-square. 

Consequently, by applying the probability in the weight equation (\ref{eq:ContWeight2}), we obtain the optimal weight. Both the ranks probability and the weight equations require the mean of the effect sizes; however, the mean of the effect sizes used in the probability equation and the mean effect size in the weight equation are estimated independently. Although both estimations are independent, we assume that they have a correlation under the alternative model. This relationship can be observed by applying simple linear regression, as the simplest step. Thus, we fitted simple linear regressions by using the covariates as a dependent variable and the test statistics as an independent variable. Then we computed the predicted value of the covariates corresponding to the mean of the test statistics. The mean value of the test statistics was then used in the weight equation (\ref{eq:ContWeight2}) as the estimated mean of the effect sizes, and the predicted value of the covariates was used as the estimated mean effect size of the covariates to compute the ranks probability of the test statistics.
\begin{equation}\label{eq:ContWeight2}
w_i \approx \Big(\frac{m}{\alpha}\Big) \bar \Phi \Bigg (\frac{E(\varepsilon)}{2} + \frac{1}{E(\varepsilon)} log\Big(\frac{\delta}{\alpha P\big(r_i \mid E(\varepsilon_i)\big)}\Big)\Bigg).
\end{equation}

The weighting equation (\ref{eq:ContWeight2}) is designed for the continuous effect size, where $m, \alpha, \varepsilon, \delta,$ and $r_i$ refer to the test size, significance level, effect size, Lagrange multiplier and the test rank, respectively, and $\Phi=1-\bar\Phi $ refers to the cumulative distribution function of the normal distribution. We also proposed a similar optimal weighting scheme (\ref{eq:BinWeight2}) for the case in which the effect sizes are considered as binary; i.e., all effect sizes under the null models are zero, and the effect sizes under the alternative models are a fixed value; this fixed value could be mean, median, or any other value depending on the specifics of the problems. However, the estimation procedures of the probability of the rank and the weight are similar to the procedures followed in the continuous effect case. In the binary case, we used the following weight equation (\ref{eq:BinWeight2}). In summary, we actually followed the data analysis steps described in the Chapter \ref{ch:CRW_data_application}.  
\begin{equation}\label{eq:BinWeight2}
w_i = \Big(\frac{m}{\alpha}\Big) \bar \Phi \Bigg (\frac{\varepsilon}{2} + \frac{1}{\varepsilon} log\Big(\frac{\delta m}{\alpha m_1 P(r_i \mid \varepsilon)}\Big)\Bigg)I(\varepsilon = \varepsilon).
\end{equation}

The $R^2$ value of the linear regression model is $0.1562$. This value, as well as the diagnostic plots (Figure \ref{fig:summary_plots_filter}f) suggested that the model fit can be improved. This leaves a scope of further research, which is beyond the goal of this dissertation. The current model information is sufficient for our present purposes. The simple linear regression also provided the predicted mean, $1.68$, and median, $1.63$, of the covariates for the estimated mean effect size $3.18$ and the estimated median effect size $2.57$ of the test statistics, respectively. We applied $1.68$ and $1.63$ to compute the ranks probability of the test statistics for the continuous and the binary effects, respectively. We plugged the estimated mean effect size, $3.18$, and the estimated median effect size, $2.57$, into the weight equations of the continuous (\ref{eq:ContWeight2}) and binary (\ref{eq:BinWeight2}) effects, respectively, in addition to the ranks probability of the tests to estimate the weights of the p-values. Finally, we used the estimated weights to compute the weighted p-values and compared the new p-values with the threshold. If we are interested in controlling the Family Wise Error Rate (FWER), then we use $\frac{\alpha}{m}$ as a threshold to identify the significant markers, and if we are interested in controlling the False Discovery Rate (FDR), then we apply \cite{benjamini1997false} FDR procedure to identify the significant markers.
\begin{figure}
\begin{center}
	\includegraphics[scale=.6]{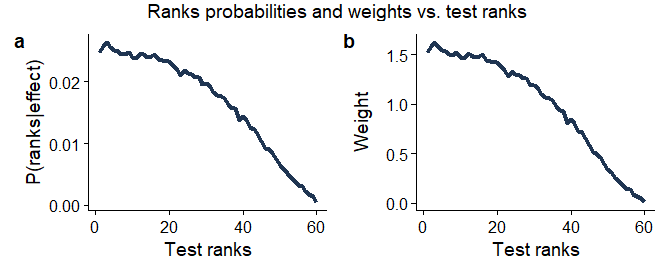}
\end{center}
\caption[Mouse QTL data: ranks probability and the corresponding weights]{\footnotesize{The plot shows the ranks probabilities of the tests and the corresponding weights for the continuous effect sizes.}}
\label{fig:prob_and_weight_third}
\end{figure}

\hspace{.2in}The following markers were identified (Table \ref{table:sig_marker_proposed}) as the significant after applying the CRW method while controlling the FWER. Both the continuous and the binary cases identified the same markers. Comparing \Cref{table:sig_marker_qtl,table:sig_marker_proposed} reveals that the CRW method can identify more markers than the $R/qtl$ analysis. The CRW method not only identified the markers that had already been identified by the $R/qtl$ method but also identified new markers in the different locations of the chromosomes. If we observe the chromosome and the $cM$ position then the markers $D3MIT18, D6MIT132,$ and $D11MIT4$ are seemingly new discoveries after applying the proposed method. In other words, we can expect more $QTLs$ that are significant to mouse growth. These $QTLs$ could be found in the same region of these markers. Consequently, this also reveals that there are $QTLs$ in these locations that are significantly responsible for mouse growth. The possible newly identified $QTLs$ are $Wt10q2, Mtbcq1,$ and $Lgaq4$ in chromosome $3$; $Dbm1, Tabw2, Plbcq3, W6q5,$ and $Obwq3$ in chromosome $6$; and $Mskt6, Fina1, Skl5,$ and $Stzid3$ in chromosome $11$.
\begin{table}[ht]
\caption[Significant markers by the CRW method]{\footnotesize{Significant markers and corresponding statistics by applying the CRW method}}
\begin{center}
\begin{tabular}{lcrrrcr}
	\hline
	markers & chr & pos & lod & stat & p-value & covariate \\ 
	\hline
	D1MIT180 & 1 & 38.20 & 4.19 & 4.48 & 0.00 & 4.00 \\ 
	D19MIT72 & 1 & 56.44 & 5.26 & 5.23 & 0.00 & 4.20 \\ 
	D1MIT200 & 1 & 63.84 & 4.81 & 4.93 & 0.00 & 4.10 \\ 
	D2MIT133 & 2 & 59.97 & 9.42 & 7.58 & 0.00 & 12.40 \\ 
	D2MIT224 & 2 & 62.89 & 9.55 & 7.65 & 0.00 & 3.30 \\ 
	D2MIT22 & 2 & 74.83 & 5.55 & 5.42 & 0.00 & 12.00 \\ 
	D2MIT49 & 2 & 83.83 & 4.01 & 4.35 & 0.00 & 5.10 \\ 
	D3MIT130 & 3 & 2.78 & 3.00 & 3.52 & 0.00 & 4.15 \\ 
	D3MIT97 & 3 & 37.06 & 6.27 & 5.87 & 0.00 & 14.90 \\ 
	D3MIT18 & 3 & 72.32 & 3.44 & 3.89 & 0.00 & 5.90 \\ 
	D4MIT27 & 4 & 42.13 & 4.42 & 4.65 & 0.00 & 4.80 \\ 
	D6MIT138 & 6 & 1.81 & 3.27 & 3.75 & 0.00 & 5.20 \\ 
	D6MIT132 & 6 & 44.78 & 2.95 & 3.48 & 0.00 & 10.10 \\ 
	D11MIT2 & 11 & 9.79 & 6.73 & 6.14 & 0.00 & 7.50 \\ 
	D11MIT4 & 11 & 41.87 & 3.31 & 3.79 & 0.00 & 10.30 \\ 
	\hline
\end{tabular}
\end{center}
\footnotesize{\parbox[t]{6in}{Significant markers while controlling the FWER. To identify the markers we computed significance threshold by permutation with HK regression. Here, chr = chromosome number, pos = marker position, lod = $LOD$ score, stat = test statistics, pval = p-value, covariate = covariates.}}
\label{table:sig_marker_proposed}
\end{table}


\chapter{Discussion}\label{ch:discussion}
\hspace{.22in} The multiple testing problem is fundamental to high throughput data, and the necessary statistical stringency makes detection of features of interest difficult unless their effect size is large. Unless this problem can be overcome, then the promise of high throughput data may never be realized. Weighted p-values provide a framework for using external information to prioritize the data feature that is most likely to be true effects. Many studies  \citep{holm1979simple,westfall2004weighted,kropf2004nonparametric,genovese2006false,rubin2006method,wasserman2006weighted,ionita2007genomewide,roeder2009genome,roquain2009optimal} have proposed methods of p-value weighting and techniques have been steadily advancing. However, all existing methods require either 1) difficult to attain knowledge of effect sizes or effect size distributions or 2) grouping tests by the covariate and then estimating properties of the groups. While group-based methods are powerful in many scenarios, we have shown that their effectiveness may suffer when true features of interest are rare and/or have effect sizes such that the power will be low (Figure \ref{fig:group_effect}).  

There is a good reason to think that very low statistical power for detection of most features of interest is the norm for many types of high throughput biological data. This was observed for the data sets used in this dissertation, and our experience is that this is typical of many p-value weighting methods, that has been conducted and discussed. Some of them have even shown concrete theory to obtain oracle weight and power. However, they are not obtainable in practice because these methods require estimating the true effect sizes.
 
Although there are many methods to up weight and down weight p-values, none of them are completely satisfactory in terms of discovery rate of the true test effect. Almost all of them depend on either group analysis and/or effect sizes, which reduces the efficiency of the analysis. Generally, the requirement of the effect sizes leads to group analysis. Since groups are composed of the true null and the true alternative hypotheses, groups can easily be diluted by a large number of true null tests, which ultimately produce poor weight. We proposed methods that do not require estimating the effect sizes and, therefore, do not require group analysis. 

We proposed three methods: 1) Covariate rank weighting (CRW), 2) Gaussian covariate weighting (GCW), and 3) Data-driven covariate weighting (DCW) and an approximation of GCW called GCW2. All methods are based on a probabilistic relationship between the effect sizes and the test statistics by an independent covariate or the function of the covariate such as rank. 

The CRW method is most beneficial when the number of true alternative tests is small, and the effect sizes are low, although it outperforms other methods for the other cases as well. However, when the other methods are not effective in detecting a very small true effect size, CRW performs quite well in that situation. This is partly true because CRW does not depend on the group analysis and can estimate a distinct weight for each test. Thus, the true effects can easily stand out from the false effects, although CRW also performs well in the grouping formulation.
  
CRW depends on the mean of the covariate and test effects. Therefore, estimating correct mean effect is crucial in order to estimate the correct weight. The proposed weight equation requires the same mean effect for covariate and test statistics. However, in practice, it is unlikely to have the same effects in both cases. The simulation studies suggest that CRW will perform similarly as long as the difference between the two mean effects is not high. If the difference between the mean covariate-effect and the mean test-effect varies within a small range, for example 1 - 2, CRW performs almost as well as if they are equal. Simulation studies also suggest that if the estimated mean test-effect is lower than the estimated mean covariate-effect, CRW under performs and vice-versa as long as the difference is not very high and varies within the small range of magnitudes. Our experience from the simulations and data analyses suggests that it is highly likely to have a larger mean test-effect size than the mean covariate-effect size.

We derived mathematical results of the probabilistic relationship, power, and weight for both cases: 1) for continuous effects and 2) for binary effects. Simulation results and the data examples suggest that the continuous version generally performs better than the binary version. Our method also performs better when the true number of alternative tests is small. CRW requires estimating the mean of the covariate and test effect sizes and the proportion of the true null hypothesis. Simulation results and data analysis suggest that correct estimation of these parameters can significantly improve the results.

We also observed the power of CRW when the mean test-effect and the mean covariate-effect are different. The power of CRW is not affected significantly by the higher mean covariate-effect and the proportion of the true null, and if the coefficient of variation $(CV)$ of the test-effect is close to $1$ (see Appendix \ref{ch:addtional_plots}); however, for the low effects and higher $CV$, the power decreases compared to other methods.

Furthermore, we observed the significance of the tests correlation, the proportion of the true nulls, and the controls over the error rate. The power is increased if the tests correlation is high and the mean effect size is low; however, for the large effect sizes, the effect of correlations is minimal. Generally, CRW performs well when the proportion of the true alternatives is low, i.e., below $20\%$. For the low mean covariate-effect and tests correlation, the proportion of null tests does not affect the CRW method. For the moderate mean effect size and the high correlation, CRW performs similarly to existing methods as long as the proportion of the true nulls is above $60\%$. However, CRW always performs well if the null proportion is high, approximately above $90\%$, regardless of the test correlations and the sizes of the mean covariate-effect. Moreover, we tested via simulations that CRW controls FWER and FDR for the different effect sizes

The fundamental difference between the CRW method and other methods is that almost all p-value weighting methods require estimating effect sizes or group analysis. To compensate for this problem \cite{dobriban2015optimal}, instead, accounted for the uncertainty by considering $\varepsilon_i$  is from the Gaussian model, i.e., $\varepsilon_i \sim N(\eta_i, \sigma_i^2)$. Their method is a special case of CRW. Replacing the rank of the covariate $(r_i)$ by the covariate $(x_i )$ itself introduces a new form of the optimal power and the weights. Consequently, assuming the conditional covariates are normally distributed introduces the method called Gaussian Covariate Weighting (GCW). Furthermore, allowing the prior distribution of the effect sizes and the distribution of the covariates to be $uniform(0,1)$, GCW becomes exactly the same as the BW method that was discussed in \cite{dobriban2015optimal}. 

The caveat is that the BW method assumes the same prior and posterior effect sizes. It also may not obtain the optimal power in many cases since the method uses only the prior information without intervening the p-values, especially when the prior and the posterior distribution are not the same. For example, Proteomics or GWAS data where the prior information is count and minor allele frequency, respectively, and the posterior distribution is a Gaussian. Another disadvantage of Dobriban’s method \citep{dobriban2015optimal} is that the posterior sample sizes need to be higher than the prior sample sizes in order to obtain the improved results because means and the standard deviations depend on the sample sizes.

To compensate for this limitation we propose the GCW method, where we assumed an independent Gaussian covariate in addition to the prior information. We also suggested an approximate version of GCW to avoid the Gaussian assumption about the covariates. For GCW, instead of assuming that the effect sizes are Gaussian, we assume that there is a vector $\bar X = \{x_1,\ldots,x_m\}$ of Gaussian covariates. These covariates are available from independent data. It is assumed that the $i^{th}$ element of the covariate-vector corresponds to the $i^{th}$ hypothesis test in a sense that the covariates will tend to be higher for more promising tests and lower for less promising ones. 

The benefit of the GCW method is that it can provide a bridge between the prior effects and the test statistics via covariates. GCW methods can perform as well as BW \citep{dobriban2015optimal} by assuming just a flat prior. This method allows a new way of analyzing the p-valued multiple hypothesis test by incorporating two sets of independent external information such as the covariates and the priors.

Lastly, we proposed the DCW method. For the existing methods, the requirement of the effect sizes leads to group analysis. Since groups are composed of the true null and the true alternative hypotheses, groups can easily be diluted by a large number of true null tests, which ultimately produce poor weight. We proposed a data driven method (DCW) that does not require estimating the effect sizes. This method is also based on group analysis; however, the groups do not dilute because the DCW does not require estimating the effect sizes. DCW is based on a probabilistic relationship between the effect size and the ranking of the test by an independent covariate. DCW depends on only the mean of the test effect sizes if the effect sizes follow a continuous distribution. If there is evidence that the effects are particularly divided into two groups, then the DCW method only requires a fixed effect size. Most of the time the fixed effect size can be assumed to be the median of the true effect sizes. The DCW method is simple and straightforward and does not require any mathematical manipulations, and it is also easier to apply.

DCW is most beneficial when the number of tests is high although it outperforms other methods for some other cases as well. DCW is less sensitive to poor estimation of the mean effect sizes of the tests. This is true partly because the proposed method is based on a probabilistic relationship which is bounded by $0$ to $1$. The DCW method also performs well in the grouping formulation, because it also allows grouping structure as well. The greatest benefit is that DCW does not require estimating effect sizes of the covariates at all.


\begin{appendices}
\chapter{Proofs}\label{ch:proofs}
\section{Compute $P(Y_l > t)$}
\hspace{.22in} Suppose the test statistic, $Y_l$, under the alternative model is from a normal distribution, i.e.,   $Y_l \sim N(\varepsilon_l, 1)$, where $\varepsilon_l$ is a random variable of the effect size; then $P(Y_l>t)=1-\int F_1 (t,\varepsilon_l)f(\varepsilon_l)d\varepsilon_l=1-\int \Phi(t-\varepsilon_l)f(\varepsilon_l )d\varepsilon_l$, where $f(\varepsilon_l)$ is the probability density function of $\varepsilon_l$.

\textbf{i)} If the effect size is from the $uniform (a, b)$, then $P(Y_l>t)$ is computed as:
\begin{equation*}
P(Y_l>t)=\frac{1}{b-a} \int_a^b\Phi(\varepsilon_l-t)d\varepsilon_l.
\end{equation*}
Let $x = \varepsilon_l-t$; then by applying integration by parts we have
\begin{equation*}
	\begin{split}
P(Y_l>t) &=\frac{1}{b-a}\int_{a-t}^{b-t}\Phi(x)dx\\
&=\frac{1}{b-a} [x\Phi(x)+\phi(x)]_{a-t}^{b-t}\\
&=\frac{1}{b-a} [(b-t)\Phi(b-t)-(a-t)\Phi(a-t)+\phi(b-t)-\phi(a-t)].
	\end{split}
\end{equation*}

\textbf{ii)} If the effect size is $\tau_l \sim Normal(\eta, \sigma^2),$ then $P(Y_l > t)$ is computed as:
\begin{equation*}
P(Y_l < t)=\int \Phi(t-\tau_l ) \frac{1}{\sigma}\phi\bigg(\frac{\tau-\eta}{\sigma}\bigg)d\tau_l.
\end{equation*}
Denote $\frac{dP(Y_l < t)}{dt} = \frac{dp}{dt}$. Differentiating w.r.t $t$ produces
\begin{equation*}
\frac{dp}{dt} =\int \phi(t-\tau_l ) \frac{1}{\sigma}\phi\bigg(\frac{\tau-\eta}{\sigma}\bigg)d\tau_l.
\end{equation*}
For the simplicity, denote $x = \frac{\tau-\eta}{\sigma}$, thus $dx=\frac{d\tau_l}{\sigma}$. Then, after rearranging the parameters and performing algebraic manipulations, we obtain
\begin{equation*}
\frac{dp}{dt} = \int \frac{1}{\sqrt{\frac{1}{\sigma^2+1}}}\phi\Bigg(\frac{x-\frac{\sigma t-\sigma \eta}{\sigma^2+1}}{\sqrt{\frac{1}{\sigma^2+1}}}\Bigg)dx \frac{1}{\sqrt{\sigma^2+1}}\phi\bigg(\frac{t-\eta}{\sqrt{\sigma^2+1}}\bigg)
\end{equation*}
First part is the normal PDF, therefore integrates to $1$, which reduces to 
\begin{equation*}
\frac{dp}{dt} = \frac{1}{\sqrt{\sigma^2+1}}\phi\bigg(\frac{t-\eta}{\sqrt{\sigma^2+1}}\bigg).
\end{equation*}
Integrating $\frac{dp}{dt}$ w.r.t $t$ produces a normal CDF. Therefore,
\begin{equation*}
P(Y_l>t) =1-\Phi\bigg(\frac{t-\eta}{\sqrt{\sigma^2+1}}\bigg)
\end{equation*}

\textbf{iii)} If the effect the size is from the $exponential (\lambda)$, then $P(Y_l>t)$ is computed as:
\begin{equation*}
P(Y_l > t) =\int_0^\infty \Phi(\varepsilon_l-t)\lambda e^{\lambda \varepsilon_v}d\varepsilon_v
\end{equation*}
Let $x = \varepsilon_l-t$; then we have
\begin{equation*}
\int_{-t}^{\infty} \Phi(x)\lambda e^{\lambda( x+t)}dx=e^{-\lambda t}\int_{-t}^{\infty}\Phi(x)\lambda e^{-\lambda x}dx.
\end{equation*}
Applying the integration by parts provides
\begin{equation*}
=e^{-\lambda t}\Bigg[\Phi(x)\int \lambda e^{-\lambda x}dx - \int \bigg(\frac{d}{dx}\Phi(x)\int \lambda e^{-\lambda x}dx\bigg)dx\Bigg]_{-t}^{\infty}
\end{equation*}
After simple integration and algebraic manipulation, we obtain
\begin{equation*}
P(Y_l > t)=\Phi(-t)+e^{-\lambda t}e^\frac{\lambda^2}{2}\Phi(t-\lambda).
\end{equation*}

\section{Relationship between covariate and test effects}
\hspace{.22in} Suppose $X =$ test effect and $Y=$ covariate effect, then
\begin{equation*}
E\Big(P(r \mid Y) \mid X \Big)=E\Bigg(\frac{P(r,Y)}{f(y)} \mid X\Bigg)=E\Bigg(\frac{P(Y \mid r)P(r)}{f(y)} \mid X \Bigg). 
\end{equation*}
$Y$ is a variable depending on $X$ but $r$ is fixed, therefore, this expectation is conducted with respect to the PDF of $y$ given $x$. Thus,
\begin{equation*}
E\Bigg(\frac{P(Y \mid r)P(r)}{f(y)} \mid X \Bigg)=\int_y \Bigg(\frac{P(Y \mid r)P(r)}{f(y)} \mid X \Bigg)f(y \mid x)dy=\int_y P(r \mid y)f(y \mid x)dy.
\end{equation*}  
Because $X$ is already conditioned in the PDF, $f(y \mid x)$; thus, there is no need to include again. Therefore,
\begin{equation*}
E\Big(P(r \mid Y) \mid X \Big)=\int_y P(r \mid y)f(y \mid x)dy=P(r \mid X).
\end{equation*}
\section{Weight for the two-tailed p-values}
\hspace{.22in} The likelihood equation for the two-tailed test can be expressed as
\begin{equation*}
L(w_i;r_i) =\frac{2}{m}\sum_{i=1}^{m}\int{\bar\Phi \Big( Z_{\frac{\alpha w_i}{2m}}-\varepsilon\Big)mP(r_i \mid \varepsilon)f(\varepsilon)I(\varepsilon > 0)d\varepsilon} - \delta\bigg(\frac{1}{m}\sum_{i=1}^{m}w_i-1\bigg),
\end{equation*}
Differentiating with respect to $w_i$ produce
\begin{equation*}
\frac{dL}{dw_i}=\frac{2}{m}\int_{\varepsilon=0}^\infty \frac{-\phi\bigg(\bar\Phi^{-1}(\frac{\alpha w_i}{2m})-\varepsilon\bigg)\Big(\frac{\alpha}{2m}\Big)}{-\phi\bigg(\bar\Phi^{-1}(\frac{\alpha w_i}{2m})\bigg)}mP(r_i \mid \varepsilon)f(\varepsilon)d\varepsilon-\frac{\delta}m.
\end{equation*}
After equating $\frac{dL}{dw_i}=0$ and performing simple algebra, we obtain
\begin{equation*}
\int_{0}^{\infty}\left( e^{Z_{\frac{\alpha w_i}{2m}\varepsilon}-\frac{\varepsilon^{2}}{2}}\right)P(r_i \mid \varepsilon)f(\varepsilon)d\varepsilon=\frac{\delta}{\alpha}.
\end{equation*}
Suppose,
\begin{equation*}
g(\varepsilon)=\left( e^{Z_{\frac{\alpha w_i}{2m}\varepsilon}-\frac{\varepsilon^{2}}{2}}\right) P(r_i \mid \varepsilon).
\end{equation*}
Since $\varepsilon$ is a random variable and $g(\varepsilon)$  is differentiable, by the first order Taylor series expansion of $g(\varepsilon)$ we have 
$E(g(\varepsilon)) \approx g(E(\varepsilon)).$
Consequently, the above equation reduces to 
\begin{equation*}
Eg(\varepsilon) \approx \left( e^{Z_{\frac{\alpha w_i}{2m}E(\varepsilon)}-\frac{\big(E(\varepsilon)\big)^{2}}{2}}\right) P\Big(r_i \mid E(\varepsilon)\Big) \approx \frac{\delta}{\alpha}.
\end{equation*} 
Consequently, an approximate version of the weight can be obtained, which is
\begin{equation*}
w_i \approx \Big(\frac{2m}{\alpha}\Big) \bar \Phi \Bigg (\frac{E(\varepsilon)}{2} + \frac{1}{E(\varepsilon)} log\bigg(\frac{\delta}{\alpha P\big(r_i \mid E(\varepsilon)\big)}\bigg)\Bigg).
\end{equation*}

\chapter{Software}\label{ch:R_packages}
\hspace{.22in} Three \textbf{R} packages are developed and are available on \textbf{Bioconductor} and \textbf{GitHub} to reproduce the results:\\
 
1) $OPWeight$ (\url{https://bioconductor.org/packages/OPWeight/} or \url{https://github.com/mshasan/OPWeight}),

2) $OPWpaper$ (\url{https://github.com/mshasan/OPWpaper}), and

3) $empOPW$ (\url{https://github.com/mshasan/empOPW}).\\

The packages come with detailed documentations and vignettes that show the application procedures of the proposed methods. The executable documents, code, and results of the data and the simulations can be obtained from the above links.

\chapter{Additional plots}\label{ch:addtional_plots}
\section{Ranks probability}
\begin{figure}
	\begin{center}
		\includegraphics[scale=.58]{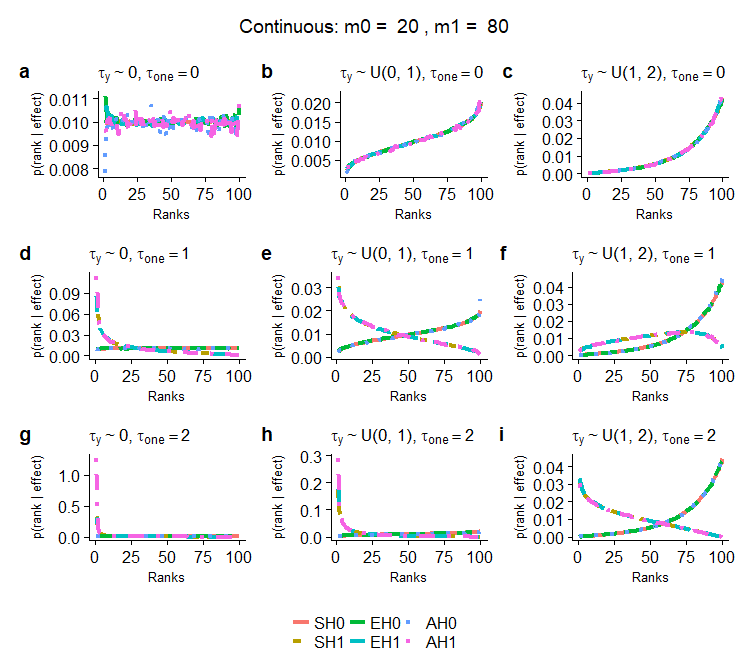}
	\end{center}
	\caption[Plots of $P(r_i=k \mid \tau)$ given the continuous effects when there is $20\%$ true null tests]{\footnotesize{This Figure shows $P(r_i=k \mid \tau_i)$ for the continuous effects for three different scenarios: 1) from the simulation, 2) from the exact CRW method, and 3) from the CRW normal approximation. To generate these plots, we assumed that there are m=100 tests of which $m_0=20$ are true nulls, and $m_1=80$ are true alternatives. The true null tests have mean $0$ and the true alternate tests have the distribution $\tau_y$ as shown at the top of each plot. Each plot consists of six curves, SH0, SH1, EH0, EH1, AH0, and AH1, where the first letter represents the method (S = simulated, E = exact, and A = approximate), and H0 and H1 represent the hypothesis type. Three curves (SH0, EH0, and AH0) starting from the bottom-left represent $P(r_i=k \mid \tau_i=0)$, and the remaining three curves (SH1, EH1, and AH1) starting from top-left represent $P(r_i=k \mid \tau_i=\tau_{one})$, where $\tau_{one}=\{0,1,2\}$. All the curves show the probability of the rank of a statistic with effect size either $\tau_i=0$ or $\tau_i=\tau_{one}$ across all tests. For example, for the curves of the sixth plot (starting from left-right), there are $m_0 = 20$ test statistics from the null models with effect size $0$; $79$ test statistics with effect size $\tau_y \sim uniform(1,2)$ and one test statistic with the effect size $\tau_{one}=1$ are from the alternative models. Thus, the curves SH1, EH1, and AH1 show the probability of the alternative test statistic with effect size $\tau_{one}  =1$ being higher in rank than the other test statistics, and SH0, EH0, and AH0 show the probability of a null test statistic with effect size $\tau_i= 0$ being higher in rank than the other test statistics. All plots of the simulation suggest nearly perfect alignment with the CRW (exact and approximate) methods.}}
	\label{fig:ranksProb_cont_20_null}
\end{figure}
\begin{figure}
	\begin{center}
		\includegraphics[scale=.58]{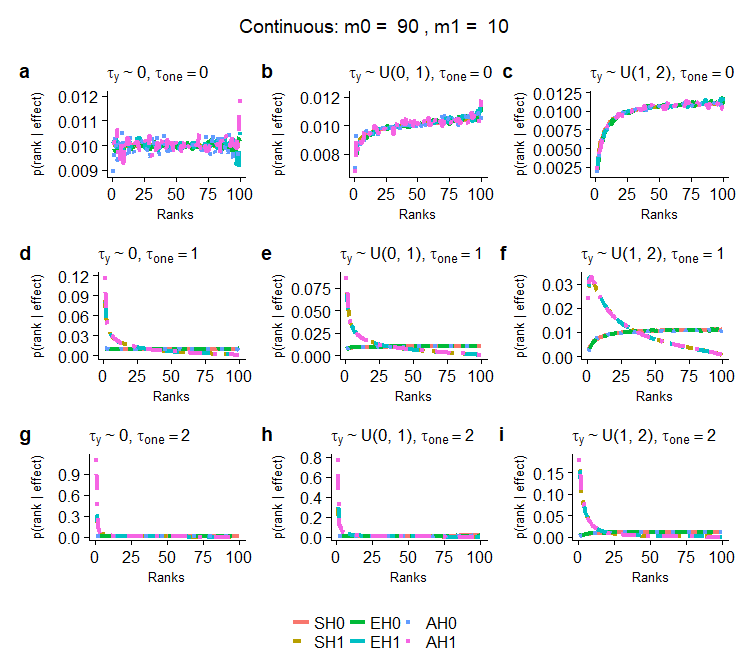}
	\end{center}
	\caption[Plots of $P(r_i=k \mid \tau)$ given the continuous effects when there is $90\%$ true null tests]{\footnotesize{This Figure shows $P(r_i=k \mid \tau_i)$ for the continuous case in three different scenarios. To generate these plots, we assumed that there are $m=100$ tests of which $m_0=90$ are true nulls and $m_1=10$ are true alternatives.}}
	\label{fig:ranksProb_cont_90_null}
\end{figure}

\newpage

\section{Ranks probabilities and weights}
\begin{figure}[ht]
	\begin{center}
		\includegraphics[scale=.6]{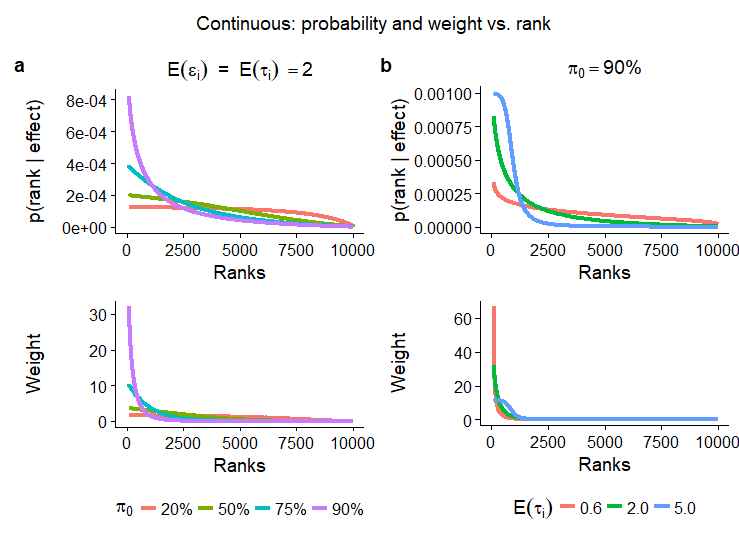}
	\end{center}
	\caption[Plots of $P(r_i=k \mid \varepsilon)$ and corresponding normalized weight for the $90\%$ true nulls]{\footnotesize{The ranks probabilities, $P(r_i=k \mid E(\varepsilon_i ))$, and the corresponding normalized weights, $w_i$ (a) for the different proportion of the true null hypothesis $(\pi_0)$ and (b) for the different mean covariate-effect $E(\tau_i)$. The effect size is the same for all alternate hypothesis tests and equal to the value shown. There are $m=10,000$ total tests. The curves were smoothed to eliminate random variation due to Monte Carlo integration.}}
	\label{fig:prob_and_weight_vs_rank90}
\end{figure}

\newpage

\section{Power across different variances of the test effects}
\hspace{.22in} This section shows the influence of the variance of the mean test effect size $E(\varepsilon_i)$ across different mean covariate effect sizes $E(\tau_i )$.
\begin{figure}[ht]
\begin{center}
	\includegraphics[scale=.55]{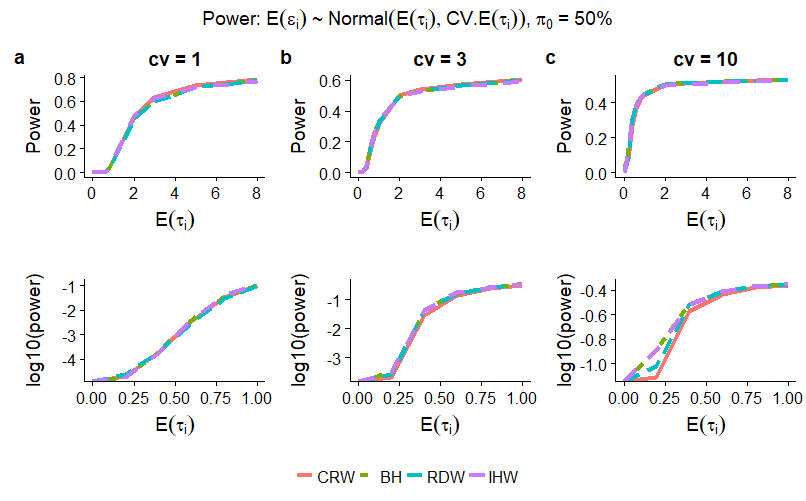}
\end{center}
\caption[Effect of the variance of the test effects when the true nulls is $50\%$]{\footnotesize{This Figure shows the simulated Power when the mean test effect $E(\varepsilon_i)$ is not equal to mean covariate effect $E(\varepsilon_i)$; rather $E(\varepsilon_i) \sim Normal(E(\tau_i),CV.E(\tau_i)$ , where $CV$ = coefficient of variations. Each plot consists of four curves based on CRW, Benjamini-Hochberg (BH), Roeder and Wasserman (RDW), and Independent Hypothesis Weighting (IHW) methods. The first row shows the power for low to high effect sizes, and the second row shows power for the low effect sizes. Three columns are based on three groups composed of $CV=\{1,3,10\}$ true null hypotheses. To generate these plots, we conducted $1,000$ replications and assumed that there were $m=10,000$ hypotheses tests of which $50\%$ are from the true null models.}}
\label{fig:power_vs_CV50}
\end{figure}
\begin{figure}
\begin{center}
	\includegraphics[scale=.55]{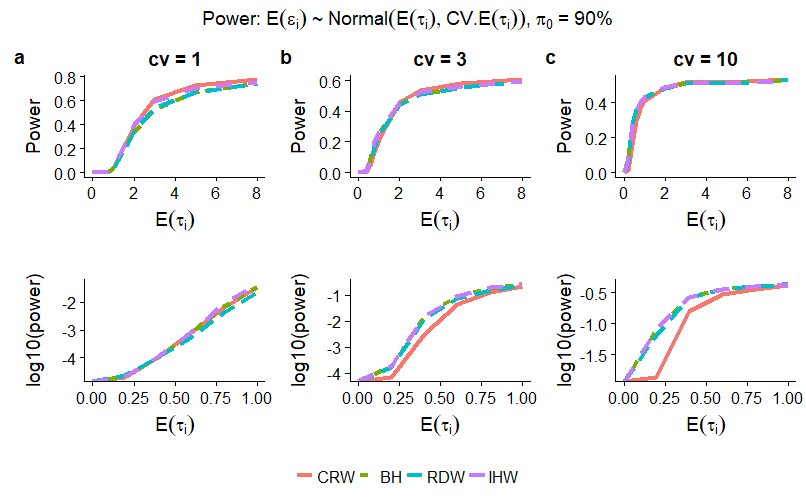}
\end{center}
\caption[Effect of the variance of the test effects when the true nulls is $50\%$]{\footnotesize{This Figure shows the simulated Power when the mean test effect $E(\varepsilon_i)$ is not equal to mean covariate effect $E(\varepsilon_i)$; rather $E(\varepsilon_i) \sim Normal(E(\tau_i),CV.E(\tau_i)$ , where $CV$ = coefficient of variations. To generate these plots, we conducted $1,000$ replications and assumed that there were $m=10,000$ hypotheses tests of which $90\%$ are from the true null models.}}
\label{fig:power_vs_CV90}
\end{figure}

\newpage

\section{Relationship between test-effect and covariate-effect}
\begin{figure}[ht]
	\begin{center}
		\includegraphics[scale=.5]{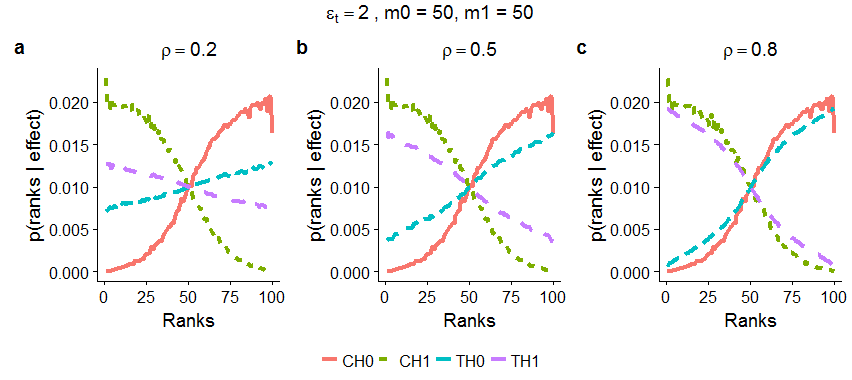}
	\end{center}
	\caption[Relationship between the test effect and the covariate effect when the true null is $50\%$]{\footnotesize{This figure shows the relationship between the test effect $(T)$ and the covariate effect $(C)$ in terms of the probability of the rank of test given the test effect size, $P(r_y=k \mid \varepsilon_t)$. In the legend, the first letter represents the source of the effects, and $H0$ and $H1$ represent the null and the alternative hypothesis, respectively. To generate the plots, we assumed that the number of hypothesis tests was $m=100$, of which $m_0= 50$ are true null and $m_1=50$ are true alternative tests; the mean test effect size of the alternative test is $\varepsilon_t=2$; and the correlation varies by $\rho=\{.2,.5,.8\}$. We performed $10,000$ replications to compute the probability of a specific rank, and for a specific rank, we generated $5,000$ observations of $\varepsilon_y$ from $Normal(\rho \varepsilon_t,1-\varepsilon^2 )$ then computed the expectation of  $P(r_y \mid \varepsilon_y)$  to obtain $P(r_y=k \mid \varepsilon_t)$.}}
	\label{fig:filterVsTest_ralation_.5}
\end{figure}
\begin{figure}
	\begin{center}
		\includegraphics[scale=.55]{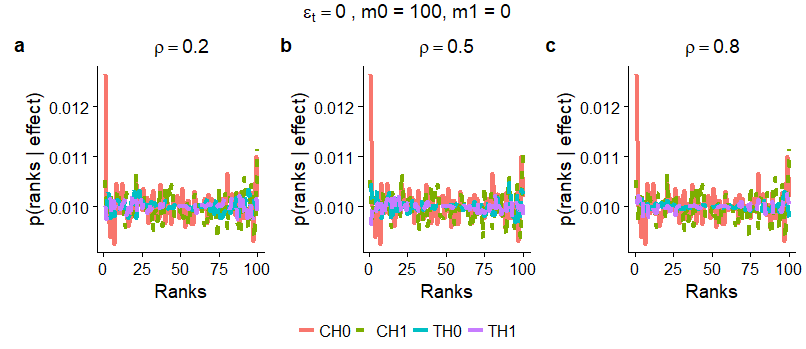}
	\end{center}
	\caption[Relationship between the test effect and the covariate effect when the true null is $100\%$]{\footnotesize{This figure shows the relationship between the test effect $(T)$ and the covariate effect $(C)$ in terms of the probability of the rank of test given the test effect size, $P(r_y=k \mid \varepsilon_t)$. To generate the plots, we assumed that the number of hypothesis tests was $m=100$, of which $m_0= 100$ are true null and $m_1=0$ are true alternative tests; the mean test effect size of the alternative test is $\varepsilon_t=0$; and the correlation varies by $\rho=\{.2,.5,.8\}$.}}
	\label{fig:filterVsTest_ralation_et0}
\end{figure}

\end{appendices}

\hspace{.25in}
\bibliographystyle{apa}
	\bibliography{References}

\end{document}